\documentclass{lmcs}
\pdfoutput=1
\usepackage[utf8]{inputenc}
\usepackage{lastpage}
\lmcsdoi{17}{4}{15}
\lmcsheading{}{\pageref{LastPage}}{}{}%
{Dec.~01,~2020}{Dec.~02,~2021}{}

\theoremstyle{definition}
\newtheorem{definition}[thm]{Definition}
\newtheorem{ex}[thm]{Example}
\newtheorem{remark}[thm]{Remark}
\newenvironment{example}{\begin{ex}\rm}{\end{ex}}
\theoremstyle{plain}
\newtheorem{theorem}[thm]{Theorem}
\newtheorem{lemma}[thm]{Lemma}
\newtheorem{corollary}[thm]{Corollary}


\title{A Probabilistic Higher-order Fixpoint Logic}

\author[Y.~Mitani]{Yo Mitani}[a]
\address{The University of Tokyo, Tokyo, Japan} %%{mitaniyo@kb.is.s.u-tokyo.ac.jp}{}{}
\author[N.~Kobayashi]{Naoki Kobayashi}[a]
\author[T.~Tsukada]{Takeshi Tsukada}[b]
\address{Chiba University, Chiba, Japan} %%{tsukada@kb.is.s.u-toyko.ac.jp}{}{}

 %TODO mandatory; please add comma-separated list of keywords

\newif\iffull\fulltrue
\newif\ifdraft\drafttrue
\ifdraft
\newcommand\nk[1]{\textcolor{red}{[#1 -nk]}}
\newcommand\todo[1]{\textcolor{blue}{[Todo: #1]}}
\newcommand\tk[1]{\textcolor{magenta}{[#1 -tsukada]}}
\else
\newcommand\nk[1]{}
\newcommand\todo[1]{}
\newcommand\tk[1]{}
\fi

\usepackage{amssymb}
\usepackage{stmaryrd}
\usepackage{bussproofs}
\usepackage{bcprules}
\usepackage{multicol}
\usepackage{proof,myproof}
\newcommand\muform{\varphi}
\newcommand\Ptr{P}
\newcommand\TP{\mathit{ProbT}}
\newcommand\redswith[2]{\stackrel{#1}{\longrightarrow}_{#2}}
\newcommand\GRAM{\mathcal{G}}
\newcommand\NONTERMS{\mathcal{X}}
\newcommand\RULES{\mathcal{R}}
\newcommand\C[1]{\oplus_{#1}}
\newcommand\Te{\mathtt{e}}
\newcommand\PHORS{PHORS}
\newcommand\cmpl[1]{\overline{#1}}
\newcommand\COL{\mathbin{:}}
\newcommand{\tyProp}{\Quant}
\newcommand\LUB[1]{\mathbin{\textstyle{\bigsqcup_{#1}}}}
\newcommand\GLB[1]{\textstyle{\bigsqcap_{#1}}}
\newcommand\LEQ{\sqsubseteq}
\newcommand\GEQ{\sqsupseteq}

\newcommand{\bnf}{\, | \,}
\newcommand{\AP}{\mathit{AP}}
\newcommand{\Var}{\mathit{Var}}
\newcommand{\iAP}{\rho_{\AP}}
\newcommand{\iAPM}[1]{\rho_{\AP,#1}}
\newcommand{\Qual}{{\mathit{Prop}}}
\newcommand{\Quant}{\mathit{Prop}}
\newcommand{\Prop}{\mathit{Prop}}

\newcommand{\semant}[1]{\llbracket #1 \rrbracket}
\newcommand{\semantM}[2]{\llbracket #1 \rrbracket_{#2}}
\newcommand\msem[1]{\mathop{(\!|}#1\mathop{|\!)}}
\newcommand\tofun[1]{#1^\dagger}

\newcommand{\defarrow}{\overset{\text{def}}{\Longleftrightarrow}}
\newcommand{\toQuant}{}

\newcommand{\true}{\text{true}}

\newcommand{\tr}{\textit{tr}}

\newcommand{\ty}{\tau}

\newcommand{\val}{\text{val}}

\newcommand{\toPHFL}[1]{\left< #1 \right>}

\newcommand{\PropTU}[2]{\Prop^{#1, #2}}

\newcommand{\Func}{\mathcal{F}}
\newcommand{\Funcb}{\mathcal{G}}

\newcommand{\interpret}{\rho}
\newcommand{\muInterpret}{\theta}

\newcommand{\muSemant}[1]{\semant{#1}_\mu}
\newcommand{\dom}{\mathit{dom}}
\newcommand\Avg{\mathbin{\bigcirc}}
\newcommand\Some{\lozenge}
\newcommand\All{\square}
\newcommand\lamTE{\Delta}
\newcommand\imply{\Rightarrow}
\newcommand{\LFP}{\mathit{LFP}}
\newcommand{\GFP}{\mathit{GFP}}
\newcommand{\muVdash}{\vdash_\mu}
\newcommand{\tyEnv}{\Gamma}
\newcommand{\muProp}{\Omega}

\newcommand{\optionValue}[3]{\text{if \(#1\) then \(#2\) else \(#3\)}}
\newcommand{\semantRel}[1]{\sim_{#1}}

\newcommand{\leqR}{\leq_{\mathbb{R}}}
\newcommand{\tyQual}{\sigma}
\newcommand{\decTypeSys}{\mathcal{T}}

\newcommand{\sqleq}{\sqsubseteq}

\newcommand{\bigJoin}{\bigsqcup}

\newcommand{\inState}{s_{\text{in}}}
\newcommand{\inStateM}[1]{s_{\text{in},#1}}

\newcommand{\decTyEnv}{\mathcal{K}}
\newcommand{\decTy}{\kappa}

\newcommand{\muTy}{A}
\newcommand{\muTyT}{T}

\newcommand\pM{\vdash_M}
\newcommand\TA[1]{\mathcal{T}_{#1}}
\newcommand\set[1]{\{#1\}}
\newcommand\form{\phi}
\newcommand\A{\automaton}
\newcommand\M{M} % Markov chain
\newcommand\PHFL{\textrm{PHFL}}
\newcommand\fspace[2]{#1\to#2}
\newcommand\p{\vdash}
\newcommand\D{D} % semantic domain
\newcommand\order{\mathit{order}}
\newcommand\mup{\mu^p}
\newcommand{\automaton}{A}
\newcommand{\alphabet}{\Sigma}

\allowdisplaybreaks

\begin{document}

\maketitle

\begin{abstract}
We introduce PHFL, a probabilistic extension of higher-order fixpoint logic, which can also be regarded as a higher-order extension of
probabilistic temporal logics such as PCTL and the \(\mu^p\)-calculus.
We show that PHFL is strictly more expressive than the \(\mu^p\)-calculus, and
that the PHFL model-checking problem for finite Markov chains is
undecidable even for the \(\mu\)-only, order-\(1\) fragment of PHFL.
Furthermore the full PHFL is far more expressive: we give a translation from Lubarsky's \( \mu \)-arithmetic to
PHFL, which implies that PHFL model checking is \(\Pi^1_1\)-hard and \( \Sigma^1_1 \)-hard.
As a positive result, we characterize a decidable fragment of the PHFL
model-checking problems using a novel type system.
\end{abstract}

\section{Introduction}\label{section:intro}
Temporal logics such as CTL and CTL* have been playing important roles
in system verification.
%%A variety of logics have been proposed for a variety of purposes.
Among the most expressive temporal logics is the \emph{higher-order fixpoint logic} (\emph{HFL} for short)
proposed by Viswanathan and Viswanathan~\cite{DBLP:conf/concur/ViswanathanV04}, which is a higher-order extension of the \emph{modal \( \mu \)-calculus}~\cite{DBLP:journals/tcs/Kozen83}.
HFL is known to be strictly more expressive than the modal $\mu$-calculus but the model-checking problem against finite models is still decidable.

In view of the increasing importance of probabilistic systems,
temporal logics for probabilistic systems (such as PCTL~\cite{DBLP:journals/fac/HanssonJ94})
% and PCTL*\tk{citation needed}) 
and their model-checking problems have been studied and applied to verification and analysis of probabilistic systems and randomized distributed algorithms~\cite{KNP11}.
Recently Castro et al.~\cite{DBLP:conf/stacs/CastroKP15} have proposed a probabilistic extension of the modal $\mu$-calculus, called the \emph{$\mu^p$-calculus}.
They showed that the $\mu^p$-calculus is strictly more expressive than PCTL and
that the model-checking problem for the $\mu^p$-calculus belongs to NP $\cap$ co-NP.

In the present paper, we introduce \PHFL{},
a probabilistic higher-order fixpoint logic, and studies the model checking problem.
\PHFL{} can be regarded as a probabilistic extension of HFL and
as a higher-order extension of the $\mu^p$-calculus.
%%<<<<<<< HEAD
%%This paper studies the expressiveness of \PHFL{} and its model-checking problem.
%%
%%We first show that \PHFL{} is a quite expressive logic.
\PHFL{} strictly subsumes the \( \mu^p \)-calculus~\cite{DBLP:conf/stacs/CastroKP15}, 
which coincides with order-\( 0 \) PHFL.

We prove that 
\PHFL{} model checking for
finite Markov chains is undecidable even for the order-1 fragment of \PHFL{} without fixpoint alternations,
by giving a reduction from the \emph{value problem} of probabilistic automata~\cite{rabin1963probabilistic, paz1971}.
In the presence of fixpoint alternations (i.e., with both least and greatest fixpoint operators),
 \PHFL{} model checking is even harder: the order-1 \PHFL{} model-checking problem is 
\( \Pi^1_1 \)-hard and \( \Sigma^1_1 \)-hard.
%%%=======
%%%
%%%This paper studies the expressiveness of \PHFL{} and its model-checking problem.
%%%Note that the more expressive a logic is, the more difficult the model-checking problem is.
%%%
%%%We show that \PHFL{} is a quite expressive logic.
%%%It strictly subsumes the \( \mu^p \)-calculus, which coincides with the order-\( 0 \) PHFL, and it is in fact far more expressive than the \( \mu^p \)-calculus and HFL.
%%%%Introduction of fixpoints for first-order functions makes the logic more expressive than \(\mu^p\)-calculus
%%%%and causes undecidability of the model-checking problem:
%%%Even for the first-order fragment without the greatest fixpoint operator, the model-checking problem for finite Markov chains is undecidable; we prove this result by giving a reduction of the \emph{value problem} of probablistic automata~\cite{rabin1963probabilistic, paz1971}.
%%%%\tk{Why is the first-order fragment more expressive than the \( \mu^p \)-calculus?  Which connective is problematic?}
%%%The greatest fixpoint operator further increases the expressivity: the first-order \PHFL{} model-checking problem is \( \Pi^1_1 \)-hard and \( \Sigma^1_1 \)-hard.
%%%>>>>>>> 146d29e60a1e364b434fa5675712d11af1d07b3b
The proof is by a reduction from the validity checking problem for \(\mu\)-arithmetic~\cite{DBLP:conf/lics/Lubarsky89}
to \PHFL{} model checking.
This may be surprising, because both
 order-0 \PHFL{} model checking (i.e. \(\mu^p\)-calculus model checking)
 for finite Markov chains~\cite{DBLP:conf/stacs/CastroKP15} and 
 HFL model checking for finite state systems~\cite{DBLP:conf/concur/ViswanathanV04} are decidable.
The combination of probabilities and higher-order predicates suddenly makes the model-checking problem
highly undecidable.

As a positive result, we identify a decidable subclass of \PHFL{} model-checking problems.
To characterize the subclass, we introduce a type system for \PHFL{} formulas, 
which is parameterized by a Markov chain \( M \).
We show that the model-checking problem \( M \models \varphi \) is decidable provided that \( \varphi \) is typable in the type system for \( M \), by giving a decision procedure using the decidability of existential theories of reals.
%The proof of the decidability has been inspired by the decidability of
%In the fragment the semantics of every formula can be expressed
%by a finite number of real numbers.
%%%Moreover, the number of reals required depends only on the Markov chain over which the formula
%is interpreted, and the type of the formula.
%Using this property, we can give a model-checking algorithm by appealing to existential theory of reals,
%which is known to be decidable in PSPACE.
The decidable subclass is reasonably expressive: the problem of computing termination probabilities of \emph{recursive Markov chains}~\cite{DBLP:journals/jacm/EtessamiY09}
can be reduced to the subclass.

The rest of this article is organized as follows.
 Section~\ref{section:definition} introduces \PHFL{} and shows that it is strictly more expressive than the $\mu^p$-calculus.
Section~\ref{section:undecidability} proves undecidability of the model-checking problem for
\(\mu\)-only and order-\(1\) PHFL.
 Section~\ref{section:hardness} proves that the \PHFL{} model-checking problem is both \(\Pi^1_1\)-hard and \(\Sigma^1_1\)-hard.
 Section~\ref{section:fragment} introduces a decidable subclass of \PHFL{} model-checking problems, and shows that 
%%We also give a model-checking algorithm explicitly and
 the subclass is reasonably large.
 %in this section.
%show reasonably strong expressive power in this section.
 Section~\ref{section:related} discusses related work, and
 Section~\ref{section:conclusion} concludes the paper.
 A preliminary summary of this article has been published 
 in Proceedings of FSCD 2020~\cite{phflFSCD}. This article contains details omitted in
 the preliminary summary, and also significantly extends the decidable fragment of \PHFL{}
 in Section~\ref{section:fragment}.

%Proofs are in the appendix.
%%Because of the space limitation, some of the proof details have been omitted;
%The omitted proof details due to the space restriction 
%%Proofs omitted in the paper
%%are found in a longer version of this paper~\cite{MitaniFull}.
%%, available from
%%\url{http://www.kb.is.s.u-tokyo.ac.jp/~mitaniyo/papers/fscd2020-long.pdf}.

\section{PHFL: Probabilistic Higher-order Fixpoint Logic} \label{section:definition}
This section introduces \PHFL{}, a probabilistic extension of HFL~\cite{DBLP:conf/concur/ViswanathanV04}.
\PHFL{} is a logic used for describing properties of Markov chains.
We define its syntax and semantics and show that it is more expressive than
the \( \mu^p \)-calculus~\cite{DBLP:conf/stacs/CastroKP15}.

\subsection{Markov Chains}
We first recall the standard notion of Markov chains. Our definitions follow those in~\cite{DBLP:conf/stacs/CastroKP15}.
%%As $\mu^p$-calculus, a formula of PHFL expresses a property of a {\textit {Markov chain}}.
\begin{definition}
  Let \(\AP\) be a set of atomic propositions.
A \emph{Markov chain} over \(\AP\) is a tuple $(S, \Ptr, \iAP, \inState)$, where:
\begin{itemize}
\item $S$ is a finite set of states,
\item $\Ptr : S \times S \to [0, 1]$ satisfying $\sum_{s' \in S} \Ptr(s, s') = 1$
  for every \(s\in S\), describes transition probabilities,
\item $\iAP : \AP \to 2^S$ is a labeling function, and
\item $\inState \in S$ is an initial state.
\end{itemize}
For a Markov chain $M = (S, P, \iAP, \inState)$, its \emph{embedded Kripke structure} is $K = (S, R, \iAP, \inState)$ where
%% \begin{itemize}
%% \item $S, \iAP$ and $\inState$ are the same as those in the Markov chain, and
%% \item
  $R \subseteq S \times S$ is a relation such that $R = \{(s, s') | \Ptr(s, s') > 0\}$.
%%\end{itemize}
\end{definition}
Intuitively, \(P(s,s')\) denotes the probability that the state \(s\) transits
to the state \(s'\), and \(\iAP(p)\) gives the set of states where \(p\) is true.
Throughout the paper, we assume that the set \(\AP\) of atomic propositions
is closed under negations, in the sense that for any \(p\in\AP\),
there exists \(\overline{p}\in \AP\) such that \(\iAP(\overline{p})=
S\setminus \iAP(p)\).

Given a Markov chain \(M\), we often write
\(S_M, P_M, \iAPM{M}, \inStateM{M}\) for its components; we omit
the subscript \(M\) when it is clear from the context.
%%Our definitions of Markov chains and their embedded Kripke structures are based on those in~\cite{DBLP:conf/stacs/CastroKP15}.

\subsection{Syntax of \PHFL{} Formulas}
\label{sec:syntax-of-phfl}
As in HFL~\cite{DBLP:conf/concur/ViswanathanV04, DBLP:conf/popl/KobayashiLB17}, we need the notion of types to 
define the syntax of \PHFL{} formulas.
%%In this section we first give the syntax of types in PHFL.
%%Then we give the syntax of the logic and type derivation rules.

The set of types, ranged over by \(\ty\), is given by:
\[
\ty ::= %\Qual \bnf
\Quant \bnf \ty_1 \rightarrow \ty_2.
\]
%%The type \(\Qual\) is
%%for \emph{qualitative} propositions, which take truth values (\(0\) for false, and \(1\) for true). In contrast,
The type \(\Quant\) describes \emph{quantitative} propositions, whose values
range over \([0,1]\). Intuitively, the value of a quantitative proposition represents the \emph{probability}
that the proposition holds. The type \(\ty_1\to\ty_2\) is for functions from \(\ty_1\) to \(\ty_2\).
For example, \((\Quant\to\Quant)\to \Quant\) represents the type of (higher-order, quantitative) predicates
on unary predicates.

\begin{remark}
  In the \(\mup{}\)-calculus~\cite{DBLP:conf/stacs/CastroKP15}
  and the previous version of this paper~\cite{phflFSCD}, two kinds of propositions were considered:
  \emph{quantitative} propositions,
  which take values in \([0,1]\), and
  \emph{qualitative} propositions, which take truth values.
  In the present paper, we consider only quantitative propositions for the sake of simplicity,
  and regard qualitative propositions as a special case of the former
  by treating \(0\) and \(1\) as the truth values ``false'' and ``true'' respectively.
\end{remark}
%%For example, \((\Qual\to\Qual)\to \Quant\) represents the type of (higher-order) quantitative predicates
%%on a qualitative predicate.
%%The set of types in PHFL is the set of simple types with two ground types $\Qual$ and $\Quant$.
%%A formula $\phi$ of type $\Qual$ is a qualitative formula
%%expressing whether $\phi$ holds in each state $s \in S$, while
%%a formula $\psi$ of type $\Quant$ is a quantitative formula
%%expressing the probability of holding $\psi$ in each state $s \in S$.

We assume a countably infinite set
\(\Var\) of variables, ranged over by $X_1, X_2, \dots$.
The set of PHFL (pre-)formulas, ranged over by \(\phi\), is given by:
\[
\phi ::=
p \bnf
X \bnf
\phi_1 \lor \phi_2 \bnf
\phi_1 \land \phi_2 \bnf
[\phi]_J \bnf
%\toQuant{\phi} \bnf
\square \phi \bnf
\lozenge \phi \bnf
\Avg \phi \bnf
\mu X. \phi \bnf
\nu X. \phi \bnf
\lambda X. \phi \bnf
\phi_1 \, \phi_2.
\]
Here,
\(p\) ranges over the set \(\AP\) of atomic propositions
(of the underlying Markov chains; we thus assume that \(\AP\) is closed under negations).
The subscript \(J\) of \([\phi]_J\)
is either ``\( > r \)'' or ``\( \ge r \)'' for some rational number \( r \in [0,1] \).
%%In the formulas \(\form_1\lor \form_2\) and
%%\(\form_1\land\form_2\), 
%%\(\form_1\) and \(\form_2\) may be quantative formulas: in that case,
%%\(\form_1\lor \form_2\) (\(\form_1\land\form_2\), resp.) denotes
%%the maximum (minimum, resp.) of the probability that \(\form_1\) holds
%%and the probability that \(\form_2\) holds.
We often identify \( J \) with an interval:
for example, ``\( > r \)'' is regarded as \( (r, 1] = \{\, x \mid r < x \le 1 \,\} \).
  Given a quantitative proposition \( \phi \), the formula \( [\phi]_{>r} \) (resp.~\( [\phi]_{\ge r} \)) is a qualitative formula, which is true just if the probability that \( \phi \) holds is greater than \( r \) (resp.~no less than \( r \)).
  We exclude trivial bounds ``\(>1\)'' and ``\(\ge 0\)'';
  note that \([\form]_{>1}\) and \([\form]_{\ge 0}\) are equivalent
  to false and true respectively.
%% In the formula \([\phi]_J\), \(\phi\) is expected to be a quantitative proposition
%% (which will be enforced by the type system introduced below),
%% and \([\phi]_J\) is a qualitative proposition that expresses whether \(J(p)\)
%% holds for the probability \(p\)
%% that \(\phi\) holds.
The formulas \(\square\phi\), \(\lozenge\phi\), and
\(\Avg\phi\) respectively mean the minimum, maximum, and average
probabilities that \(\phi\) holds after a one-step transition.
%%(Note that if the probability that \(\phi\) ranges over \(\set{0,1}\),
%%then the meaning s of
%%\(\square\phi\) and \(\lozenge\phi\) coincide with those in
%%a non-probabilistic temporal logic:
%%\(\square\phi\) and \(\lozenge\phi\) respetively mean ...
%%The formula \(\square\phi\) describes that \(\phi\) holds after \emph{any} one-step transition,
%%and \(\lozenge\phi\) describes that \(\phi\) holds after \emph{some} one-step transition.
%%The formula \(\Avg\phi\) represents the probability that \(\phi\) holds after one-step transition.
The formulas \(\mu X. \phi\) and \(\nu X.\phi\) respectively
denote the least and greatest fixpoints
of \(\lambda X.\phi\).
Note that \(\phi\) may denote higher-order predicates, as in HFL~\cite{DBLP:conf/concur/ViswanathanV04} (but unlike in the modal \(\mu\)-calculus and
its probabilistic variants~\cite{DBLP:conf/stacs/CastroKP15, DBLP:journals/corr/MioS13, morgan1997probabilistic},
%,\tk{what are variants?  Is there any extension other than \( \mu^p \)?}
where fixpoints are restricted to propositions).
We have also \(\lambda\)-abstractions and applications,
which are used for manipulating higher-order predicates.
The prefixes
\(\mu X\), \(\nu X\) and \(\lambda X\) bind the variable \(X\).
As usual,
we identify formulas up to the renaming of bound variables and
implicitly apply \(\alpha\)-conversions. We write
\([\form_1/X]\) for the capture-avoiding substitution of \(\form_1\) for
\(X\), and \([\form_1/X]\form_2\) for the formula obtained by
applying the substitution \([\form_1/X]\) to \(\form_2\).

In order to exclude out ill-formed formulas like \((p_1\lor p_2)(\phi)\),
%(where a proposition \(p_1\lor p_2\) is wrongly used as a function),
we restrict the shape of
formulas using a simple type system.
A \emph{type environment} is a map from a finite set of variables to the set of types.
A \emph{type judgment} is of the form \( \tyEnv \vdash \phi : \ty \).
The typing rules are shown in Figure~$\ref{fig:PHFLTypes}$.
%%In the figure, \( \tyProp \) is a meta-variable ranging over
%%the set \(\set{\Qual,\Quant} \) of proposition types.
%%For example, the rule for \(\phi_1\land\phi_2\)
%%means that \( \tyEnv \vdash \phi_i \colon \Qual \)
%%for each \(i\in\set{1,2}\) implies
%%\( \tyEnv \vdash \phi_1 \wedge \phi_2 \colon \Qual \) and that
%%\( \tyEnv \vdash \phi_i \colon \Quant \) for each \(i\in\set{1,2}\)
%%implies \( \tyEnv \vdash \phi_1 \wedge \phi_2 \colon \Quant \).
%In the figure, \(\tyEnv\) denotes a type environment, i.e., 
A formula \(\phi\) is \emph{well-typed}
if \(\tyEnv \vdash \phi: \ty\) is derivable for some \(\tyEnv\) and \(\ty\).
Henceforth, we consider only well-typed formulas.

\begin{example}
  \label{ex:syntax}
  For a proposition \(p \in \AP\), the formula \(\phi = (\mu F. \lambda X. X \lor F (\Avg X)) p\)
  %\, \toQuant{p}\)
  is
  a well-typed formula of type \(\Quant\).
  By unfolding the fixpoint formula (i.e., replacing \(\mu X.\form\) with
  \([\mu X.\form/X]\form\), which will be justified by the semantics introduced later)
  and applying \(\beta\)-reductions, we obtain:
  \begin{align*}
    \phi
    &\equiv (\lambda X.X\lor (\mu F.\lambda X.X\lor F(\Avg X))(\Avg X))
    \toQuant{p}\\
    &\equiv \toQuant{p}\lor (\mu F.\lambda X.X\lor F(\Avg X))(\Avg\toQuant{p})\\
    &\equiv \toQuant{p}\lor \Avg\toQuant{p}\lor
    (\mu F.\lambda X.X\lor F(\Avg X))(\Avg\Avg\toQuant{p})\\
    &\equiv \toQuant{p}\lor \Avg\toQuant{p}\lor\Avg\Avg\toQuant{p}\lor\cdots
  \end{align*}
  Thus, intuitively,
  the formula represents
  the function that maps the current state \(s\) %%each state \(s\)
  to the value \(\sup_{k \geq 0} q_k\) where \(q_k\) is the probability
  that a $k$-step transition sequence
  starting from the state \(s\) ends in a state satisfying \(p\). \qed
\end{example}

\begin{remark}
Following the definition of HFL by Kobayashi et al.~\cite{DBLP:conf/popl/KobayashiLB17},
we have excluded out negations. 
By a transformation similar to that in~\cite{DBLP:journals/corr/Lozes15}
and our assumption that the set of atomic propositions is closed under negations,
any closed ground-type formula of \PHFL{} extended with 
negations
can be transformed to an equivalent negation-free formula
 (as long as the occurrences of negations are restricted
as in the original HFL~\cite{DBLP:conf/concur/ViswanathanV04}
so that fixpoint operators are applied to only monotonic functions). \qed
\end{remark}

\begin{figure}\label{fig:PHFLtype}
\begin{minipage}{85pt}
\AxiomC{\mathstrut}
\UnaryInfC{\(\tyEnv \vdash p : \tyProp\)}
\DisplayProof
\end{minipage}
\vspace{10pt}
\begin{minipage}{100pt}
\AxiomC{\mathstrut}
\UnaryInfC{\(\tyEnv,X\COL\ty \vdash X:\ty\)}
\DisplayProof
\end{minipage}
\vspace{10pt}
\begin{minipage}{120pt}
\AxiomC{\(\tyEnv \vdash \phi : \Quant\)}
\UnaryInfC{\(\tyEnv \vdash [\phi]_J : \tyProp\)}
\DisplayProof
\end{minipage}
%%\begin{minipage}{90pt}
%%\AxiomC{\(\tyEnv \vdash \phi : \Qual\)}
%%\UnaryInfC{\(\tyEnv \vdash \toQuant{\phi} : \Quant\)}
%%\DisplayProof
%%\end{minipage}
\\
\begin{minipage}{200pt}
\AxiomC{\(\tyEnv \vdash \phi_1 : \tyProp \)}
\AxiomC{\(\tyEnv \vdash \phi_2 : \tyProp \)}
\BinaryInfC{\(\tyEnv \vdash \phi_1 \lor \phi_2 : \tyProp \)}
\DisplayProof
\end{minipage}
\vspace{10pt}
\begin{minipage}{200pt}
\AxiomC{\(\tyEnv \vdash \phi_1 : \tyProp\)}
\AxiomC{\(\tyEnv \vdash \phi_2 : \tyProp\)}
\BinaryInfC{\(\tyEnv \vdash \phi_1 \land \phi_2 : \tyProp\)}
\DisplayProof
\end{minipage}
\vspace{10pt}
\begin{minipage}{100pt}
\AxiomC{\(\tyEnv \vdash \phi : \tyProp\)}
\UnaryInfC{\(\tyEnv \vdash \square \phi : \tyProp\)}
\DisplayProof
\end{minipage}
\vspace{10pt}
\begin{minipage}{100pt}
\AxiomC{\(\tyEnv \vdash \phi : \tyProp\)}
\UnaryInfC{\(\tyEnv \vdash \lozenge \phi : \tyProp\)}
\DisplayProof
\end{minipage}
\begin{minipage}{100pt}
\AxiomC{\(\tyEnv \vdash \phi : \Quant\)}
\UnaryInfC{\(\tyEnv \vdash \Avg \phi : \Quant\)}
\DisplayProof
\end{minipage}
\begin{minipage}{100pt}
\AxiomC{\(\tyEnv, X : \ty \vdash \phi : \ty\)}
\UnaryInfC{\(\tyEnv \vdash \mu X. \phi : \ty\)}
\DisplayProof
\end{minipage}
\begin{minipage}{100pt}
\AxiomC{\(\tyEnv, X : \ty \vdash \phi : \ty\)}
\UnaryInfC{\(\tyEnv \vdash \nu X. \phi : \ty\)}
\DisplayProof
\end{minipage}
\vspace{10pt}
\begin{minipage}{100pt}
\AxiomC{\(\tyEnv, X : \ty_1 \vdash \phi : \ty_2\)}
\UnaryInfC{\(\tyEnv \vdash \lambda X. \phi : \ty_1 \to \ty_2\)}
\DisplayProof
\end{minipage}
\begin{minipage}{100pt}
\AxiomC{\(\tyEnv \vdash \phi : \ty_1 \to \ty_2\)}
\AxiomC{\(\tyEnv \vdash \psi : \ty_1\)}
\BinaryInfC{\(\tyEnv \vdash \phi \, \psi : \ty_2\)}
\DisplayProof
\end{minipage}
\caption{Typing rules for PHFL formulas.  }\label{fig:PHFLTypes}
\end{figure}

%\begin{definition}[order]
We define the \emph{order} of a type \(\ty\) by:
\[
%\order(\Qual)=
\order(\Quant)=0\qquad
\order(\ty_1 \rightarrow \ty_2) = \max(\order(\ty_1)+1, \order(\ty_2)).
\]
The order of a formula \(\phi\) such that \(\tyEnv\p \phi:\ty\) is
the largest order of types used in the derivation of \(\tyEnv\p \phi:\ty\).
The \emph{order-\(k\) \PHFL{}} is the fragment of \PHFL{} consisting of formulas of order up to \(k\).
%\end{definition}
%
Order-\(0\) \PHFL{} coincides with the \(\mup\)-calculus~\cite{DBLP:conf/stacs/CastroKP15}.

\subsection{Semantics}
\label{sec:sem-phfl}
%We now give the formal semantics of \PHFL{} formulas.
We first give the semantics of types. 
We write \(\leqR\) for the natural order over the set \(\mathbb{R}\) of real numbers,
and often omit the subscript when there is no danger of confusion.
For a map \(f\), we write \(\dom(f)\) for the domain of \(f\).
\begin{definition}[Semantics of Types]\label{defSemant}
Let \(M\) be a Markov chain.
For each \(\ty\), we define 
  a partially ordered set $\semant{\ty}_M = (D_\ty, \LEQ_\ty)$ 
inductively by:
%\begin{enumerate}
%%\item The semantics of types $\Qual$ and $\Quant$ are as follows.
\[
\begin{array}{l}
%\begin{eqnarray*}
%%D_\Qual = \fspace{S}{\{0, 1\}} \qquad
%%f \LEQ_\Qual g \,\defarrow\, \forall s \in S. f(s) \leq g(s) \\
D_\Quant = \fspace{S_M}{[0, 1]}\qquad
f \LEQ_\Quant g \,\defarrow\, \forall s \in S_M. f(s) \leq g(s)\\
%\end{eqnarray*}
%\item The semantics of the type $\ty \rightarrow \tyTwo$ is defined as follows.
%\begin{eqnarray*}
D_{\ty_1 \rightarrow \ty_2} =
\{f \in D_{\ty_1} \rightarrow D_{\ty_2} \mid \forall x, y \in D_{\ty_1}. x \LEQ_{\ty_1} y \implies f(x) \LEQ_{\ty_2} f(y)\} \\
f \LEQ_{\ty_1 \rightarrow \ty_2} g \,\defarrow\, \forall x \in D_{\ty_1}. f(x) \LEQ_{\ty_2} g(x).
%\end{eqnarray*}
\end{array}\]
%\end{enumerate}
For a type environment \(\tyEnv\), we write 
\(\semant{\tyEnv}_M\) for the set of maps \(f\) such that
\(\dom(f)=\dom(\tyEnv)\) and \(f(x)\in \D_{\tyEnv(x)}\)
for every \(x\in\dom(\tyEnv)\).
\end{definition}

We omit the subscript \(M\) below.
Note that $\semant{\ty}$ forms a complete lattice for each \(\ty\).
We write \(\bot_\ty\) for the least element of \(\semant{\ty}\), and
for a set \(V\subseteq \D_\ty\), we write \(\LUB{\ty} V\)
(\(\GLB{\ty} V\), resp.) for the least 
upper
bound  (greatest lower bound, resp.)
of \(V\) with respect to \(\LEQ_\tau\); we often omit the subscript \(\tau\) if it is clear from the context.
%by induction on the structure of $\ty$.
Note also that for every functional type $\ty_1 \to \ty_2$,
every element of $D_{\ty_1 \to \ty_2}$ is monotonic.
Thus, for every type $\ty$ and every function $f\in D_{\ty \rightarrow \ty}$,
there exist
the least and greatest fixed points of $f$, which we write \(\LFP(f)\)
and \(\GFP(f)\) respectively. They are given by:
\[
\LFP(f) = \GLB{\ty}\set{x\in D_{\ty} \mid f\,x\LEQ_\tau x}
\qquad
\GFP(f) = \LUB{\ty}\set{x\in D_{\ty} \mid x\LEQ_\tau f\,x}.
\]

%%An interpretation of a type environment $\tyEnv$ maps each variable in it to
%%a value of the semantics of the corresponding type.
%%Semantics $\semant{\tyEnv}$ of a type environment $\tyEnv$ is the set of interpretations of $\tyEnv$.

We now define the semantics of formulas. 
Since the meaning of a formula depends on its type environment,
we actually define the semantics \(\semantM{\tyEnv\p \phi:\ty}{M}\)
for each type judgment \(\tyEnv\p \phi:\ty\).
%, by induction on the (unique) derivation of \(\tyEnv\p \phi:\ty\).
Here, the subscript \(M\) denotes the underlying Markov chain, which is often omitted.
%%As the derivation of each type judgment is unique, we
%%define \(\semanti{\tyEnv\p \phi:\ty}\) by induction on the derivation of

\begin{definition}[Semantics of Type Judgement]
Let \(M\) be a Markov chain and assume that $\tyEnv \vdash \phi : \ty$ is derivable.
Then its semantics $\semantM{\tyEnv \vdash \phi : \ty}{M}\in \semant{\tyEnv}\to
\semant{\ty}$ is defined by induction on
the (unique) derivation of ${\tyEnv \vdash \phi : \ty}$ by:
%\begin{align*}
{\allowdisplaybreaks
\begin{eqnarray*}
  \semant{\tyEnv \vdash p : \Qual}_M(\interpret) &=&
  \lambda s\in S_M. \optionValue{s\in\iAPM{M}(p)}{1}{0} \\
\semant{\tyEnv \vdash X : \ty}_M(\interpret) &=& \interpret(X) \\
\semant{\tyEnv \vdash \phi_1 \land \phi_2 : \tyProp}_M(\interpret) &=&
\lambda s \in S_M. \min_{i \in \{1, 2\}} \semant{\tyEnv \vdash \phi_i : \tyProp}_M(\interpret)(s) \\
\semant{\tyEnv \vdash \phi_1 \lor \phi_2 : \tyProp}_M(\interpret) &=&
\lambda s \in S_M. \max_{i \in \{1, 2\}} \semant{\tyEnv \vdash \phi_i : \tyProp}_M(\interpret)(s) \\
\semant{\tyEnv \vdash [\phi]_J: \Qual}_M(\interpret) &=&
\lambda s \in S_M. \optionValue{\semant{\tyEnv \vdash \phi : \Quant}_M(\interpret)(s) \in J}{1}{0} \\
%\semant{\tyEnv \vdash \toQuant{\phi} : \Quant}_M(\interpret) &=&
%\semant{\tyEnv \vdash \phi : \Qual}_M(\interpret) \\
\semant{\tyEnv \vdash \square \phi : \tyProp}_M(\interpret) &=&
\lambda s \in S_M. \min_{s' : P_M(s, s') > 0} \semant{\tyEnv \vdash \phi : \tyProp}_M(\interpret)(s') \\
\semant{\tyEnv \vdash \lozenge \phi : \tyProp}_M(\interpret) &=&
\lambda s \in S_M. \max_{s' : P_M(s, s') > 0} \semant{\tyEnv \vdash \phi : \tyProp}_M(\interpret)(s') \\
\semant{\tyEnv \vdash \Avg \phi : \Quant}_M(\interpret) &=&
\lambda s \in S_M. \sum_{s' \in S_M}
\big(P_M(s, s')\cdot \semant{\tyEnv \vdash \phi : \Quant}_M(\interpret)(s')\big) \\
\semant{\tyEnv \vdash \mu X. \phi : \ty}_M(\interpret) &=&
\LFP(\lambda v\in D_{\ty}.\semant{\tyEnv, X : \ty \vdash \phi : \ty}_M(\interpret[X \mapsto v])) \\
\semant{\tyEnv \vdash \nu X. \phi : \ty}_M(\interpret) &=&
\GFP(\lambda v\in D_{\ty}.\semant{\tyEnv, X : \ty \vdash \phi : \ty}_M(\interpret[X \mapsto v])) \\
\semant{\tyEnv \vdash \lambda X. \phi : \ty_1 \to \ty_2}_M(\interpret) &=&
\lambda v \in D_{\ty_1}. \semant{\tyEnv, X : \ty_1 \vdash \phi : \ty_2}_M(\interpret[X \mapsto v]) \\
\semant{\tyEnv \vdash \phi_1 \, \phi_2:\ty}_M(\interpret) &=&
\semant{\tyEnv \vdash \phi_1:\ty_2\to\ty}_M(\interpret) \,
(\semant{\tyEnv \vdash \phi_2:\ty_2}_M(\interpret))
\end{eqnarray*}
}
%\end{align*}
%%Here \( \tyProp \in \{\, \Qual, \Quant \,\} \).

\end{definition}

%%Let us see that this definition is well-defined.
%%Here are some remarks on the
%%well-definedness of \(\semant{\tyEnv\p\phi}\).
%%\begin{remark}
%%  \item
In the last equality, \(\ty_2\) is uniquely determined from \(\tyEnv\) and \(\phi_2\).
In the  definitions of the semantics of $\square \phi$ and $\lozenge \phi$,
the set $\{s'\in S_M | P(s, s') > 0\}$ is non-empty and finite, because
%%This is because for each $s \in S$,
$\sum_{s' \in S_M} P(s, s') = 1$ and \(S_M\) is finite
by the definition of Markov chains. Thus
the max/min operations are well-defined.
%%is nonempty, and it is finite because it is a subset of the finite set $S$.
%%\item The right hand sides of the definition of semantics of $\mu x. \phi$ and $\nu x. \phi$ are well-defined
%%as stated before.
%%\end{enumerate}
We also note that \( \semant{\tyEnv \vdash \form : \ty} \) is a monotone function from \( \semant{\tyEnv} \) to \( \semant{\ty} \) (where \( \semant{\tyEnv} \) is ordered
by the component-wise ordering; note also Remark~\ref{rem:J} below).
This ensures the well-definedness of the semantics of \(\mu X.\form\),
\(\nu X.\form\), and \(\lambda X.\form\).

\begin{remark}
  \label{rem:J}
  Recall that in a formula \([\phi]_J\), we allow the predicate \(J\) to be 
  ``\(>r\)'' or ``\(\geq r\)'' (where \(r\in[0,1]\)), but neither  ``\(<r\)''
  nor ``\(\leq r\)''.
  Allowing ``\(<r\)'' would break the monotonicity of the semantics of a formula.
  For example, \(\semant{\emptyset\vdash \lambda X.[X]_{<1}:\Quant\to\Qual}
  = \lambda v\in D_{\Quant}.\lambda s\in S.(\mbox{if $v(s)<1$ then $1$ else $0$})\)
  is not monotonic. \qed
%%  This ensures that the semantics of \([\phi]_J\) is monotonic.
  %% That is important for the well-definedness of fixpoints;
  %% see the semantics of formulas below.
\end{remark}

We often omit \(M\), the type of the formula, and the type environment,
and just write 
$\semant{\phi}$ or $\semant{\tyEnv \vdash \phi}$ for $\semantM{\tyEnv \vdash \phi : \ty}{M}$ when
%in the notation of semantics when
there is no danger of confusion.
For a Markov chain
$M=(S, P, \iAP, \inState)$ and 
a closed \PHFL{} formula \(\form\) of type \(\Qual\),
we write  \(M\models \form\) if 
\(\semant{\phi}(\inState)=1\).

\begin{example}
  Recall the \PHFL{} formula \(\phi= \psi\, \toQuant{p}\)  where
  \(\psi=\mu F. \lambda X. X \lor F (\Avg X)\) in Example~\ref{ex:syntax}.
  We have:
{  \allowdisplaybreaks
  \begin{align*}
    \semant{\psi}
    &=
    \LFP\Big(\lambda v\in D_{\Quant\to\Quant}.
    \lambda x\in D_{\Quant}.\lambda s \in S. \\
    &\hphantom{\lambda v \in D_{\Quant\to\Quant}. \lambda x\in D_{\Quant}} \max\Big(x\,s, v\,(\lambda s'\in S.\sum_{s''}P(s',s'')\cdot (x s''))\,s\Big)\Big)\\
    &\ge \Big(\lambda v.
    \lambda x.\lambda s. \max\Big(x\,s, v\,(\lambda s'\in S.\sum_{s''}P(s',s'')\cdot(x s''))\,s\Big)\Big)^{n+1} (\bot_{\Quant\to\Quant})\\
    &= \lambda x.\lambda s. \max_{0\le k\le n}
    \sum_{s_0 s_1 \dots s_k \in S^{k + 1}, s_0=s}
    \big(x(s_k)\cdot \prod_{0 \leq j \leq k - 1} 
    P(s_j, s_{j + 1})\big)
  \end{align*}
  }
for every \(n\ge 0\).
Thus, we have:
  \[
%\begin{array}{l}
  \semant{\psi} \ge  \lambda x.\lambda s\in S.
  \sup_{k\ge 0} \sum_{s_0 s_1 \dots s_k \in S^{k + 1}, s_0=s}
  \big(x(s_k)\cdot \prod_{0 \leq j \leq k - 1} 
  P(s_j, s_{j + 1})\big).
%\end{array}
\]
Actually, the equality holds, because
the righthand side is a fixpoint of
\[\lambda v\in D_{\Quant\to\Quant}.
\lambda x\in D_{\Quant}.\lambda s\in S.\max(x\,s, v(\lambda s\in S.\sum_{s'}P(s,s')\cdot (x s'))).\]
  The semantics of \(\phi\) is, therefore,  given by
  \[
  \semant{\phi} = \lambda s \in S. \sup_{k \geq 0} \sum_{s_0 s_1 \dots s_k \in S^{k + 1}, s_0=s,
  s_k\in \iAP(p)}
  %\big(
  \prod_{0 \leq j \leq k - 1}
  P(s_j, s_{j + 1}). %%\big). %\hfill \qed
  \tag*{\qed}
  \]
\end{example}

\subsection{Expressive Power}
\PHFL{} obviously subsumes the $\mup$-calculus~\cite{DBLP:conf/stacs/CastroKP15}, which coincides with order-\( 0 \) PHFL.
Hence \PHFL{} also subsumes PCTL~\cite{DBLP:journals/fac/HanssonJ94}, since the \( \mup \)-calculus subsumes PCTL~\cite{DBLP:conf/stacs/CastroKP15}.
%%In fact, $\mu^p$-calculus can be viewed as an order-$0$ fragment of PHFL.
%%Since $\mu^p$-calculus can encode PCTL[...], PHFL can also encode PCTL.
%%It is also obvious that PHFL can encode HFL over the embedded Kripke structure.
%%
%%As Viswanathan and Viswanathan showed~\cite{DBLP:conf/concur/ViswanathanV04},
%%HFL is strictly more expressive than modal $\mu$-calculus.

\PHFL{} is \emph{strictly} more expressive than the $\mu^p$-calculus.
\begin{theorem}\label{expressiveThm}
  Order-1 \PHFL{} is strictly more expressive than the $\mup$-calculus, i.e.,
  there exists an order-1 \PHFL{} proposition \(\phi\)
  such that \(\phi\) is not equivalent to any \(\mup\)-formula.
\end{theorem}

\begin{proof}
  Let $\mathcal{M}$ be the set of Markov chains
  $M = (S, P, \iAP, \inState)$ that satisfy the following conditions.
\begin{itemize}
\item $S = \{s_0, s_1, \dots, s_n\}$ for a positive integer $n$,
\item $P(s_i, s_{i + 1}) = 1$ ($0 \leq i \leq n - 1$), $P(s_n, s_n) = 1$ and $P(s_i, s_j) = 0$ otherwise.
\item There are three atomic propositions $a, b, c$ with $\iAP(a) \cup \iAP(b) = \{s_0, s_1, \dots, s_{n - 1}\}$,
$\iAP(a) \cap \iAP(b) = \emptyset$ and $\iAP(c) = \{s_n\}$.
\item The initial state is $\inState = s_0$.
\end{itemize}

Let \(\phi\) be the order-1 \PHFL{} formula of type \(\Qual\):
\[
%%(\mu F. (\lambda X. [(\toQuant{a} \land \lozenge \toQuant{X}) \lor
%%\lozenge F (\lozenge (\toQuant{b} \land \toQuant{X})))]_{\geq 1}) (b \land [\Avg \toQuant{c}]_{\geq 1}).
(\mu F. \lambda X. a \land \lozenge (X \lor
  F (b \land \lozenge X)))(b\land \lozenge c).
\]
Note that, for \(M\in\mathcal{M}\), \(M\models \phi\) holds just if
\(n\) is even, %and if
%%Let $\mathcal{M}' \subseteq \mathcal{M}$ be the set of Markov chains accepted by the PHFL formula
%%Note that the set $\mathcal{M}'$ is the set of Markov chains $M \in \mathcal{M}$ such that $n$ is even and
%\(\iAP\) satisfies
$\iAP(a) = \{s_0, s_1, \dots, s_{\frac{n}{2} - 1}\}$
and $\iAP(b) = \{s_{\frac{n}{2}}, s_{\frac{n}{2} + 1}, \dots, s_{n - 1}\}$.

We show that there is no \(\mup\)-formula equivalent to \(\phi\).
Suppose that a \(\mup\)-formula \(\phi'\) \emph{were} equivalent to \(\phi\),
which would imply that \(M\models \phi\) if and only if \(M\models \phi'\) for
any \(M\in \mathcal{M}\).
For \(M\in\mathcal{M}\), let us write \(K_M\) for the embedded Kripke structure
of \(M\).
%%%%%%Let $\mathcal{K}$ and $\mathcal{K}'$ be the set of embedded Kripke structures of Markov chains in $\mathcal{M}$ and
%%%$\mathcal{M}'$, respectively.
Since all the transitions in \({M}\)  are deterministic,
there exists a modal \(\mu\)-calculus formula \(\phi''\) such that
\(M\models \phi'\) if and only if \(K_M\models \phi''\)
(note that \(\phi''\) is obtained by replacing $\Avg$ with
$\lozenge$, and replacing \([\phi_1]_{J}\) with 
\(\true\) if \(J\) is ``\(\ge 0\)'' and with \(\phi_1\) otherwise).
That would imply that \(K_M\models \phi''\)
for \(M\in \mathcal{M}\),
just if
\(n\) is even and
\(\iAP\) satisfies $\iAP(a) = \{s_0, s_1, \dots, s_{\frac{n}{2} - 1}\}$
and $\iAP(b) = \{s_{\frac{n}{2}}, s_{\frac{n}{2} + 1}, \dots, s_{n - 1}\}$.
But then \(\phi''\) would describe 
the non-regular language \(\set{a^{m}b^{m}\mid m\ge 1}\),
which contradicts the fact that the modal \(\mu\)-calculus can express
only regular properties.
%%
%%
%%
%%We prove there does not exists a $\mu^p$-calculus formula which recognizes $\mathcal{M}'$ from $\mathcal{M}$.
%%Assume such a formula $\phi$ exists.
%%Then we can give a formula $\psi$ in modal $\mu$-calculus recognizes $\mathcal{K}'$ from $\mathcal{K}$, by
%%substituting the operator $\lozenge$ for $\Avg$ and removing the operator $[\cdot]_J$.
%%
%%Since modal $\mu$-calculus and parity tree automata are equi-expressive [...], there exists
%%a parity tree automaton $A$ recognizes $\mathcal{K}'$ from $\mathcal{K}$.
%%This contradicts finiteness of the number of states of the automaton $A$.
\end{proof}

\begin{remark}
  For non-probabilistic logics, 
  HFL was known to be strictly more expressive than the modal
  \(\mu\)-calculus~\cite{DBLP:conf/concur/ViswanathanV04}.
  The above proof can be easily adapted to show that fact.
\end{remark}

%\nk{If we have a time, it would be better to show the
%  strictness of the hierarchies of order-\(k\) \PHFL{},
%  from which the above theorem follows as a special case.}

\section{Undecidability of \PHFL{} Model Checking}\label{section:undecidability}
In this section we prove the undecidability of 
the following problem.

\begin{definition}[PHFL Model Checking]
  The \emph{PHFL model-checking problem} for finite Markov chains
  is the problem of
  deciding
  whether $M\models \phi$, given
  a (finite) Markov chain $M$ and %% = (S, P, \iAP, \inState)$ and
  a closed PHFL formula $\phi$ of type $\Qual$ as input.
%%\begin{itemize}
%%\item[Input] A Markov chain $M = (S, P, \iAP, \inState)$ and a closed PHFL formula $\phi$ of type $\Qual$
%%\item[Output] Whether $\semant{\phi}(\inState) = 1$ or not
%%\end{itemize}
\end{definition}

We prove that the problem is undecidable
even for the order-1 fragment of PHFL without fixpoint alternations,
by a reduction from
the undecidability of the value-1 problem~\cite{DBLP:conf/icalp/GimbertO10} for
 probabilistic automata~\cite{rabin1963probabilistic}.
%%which is known to be undecidable~\cite{rabin1963, paz1971}.
%%to model checking problem of PHFL.
%%In the reduction we encode a probabilistic automaton as a Markov chain, and a set of input alphabets as a PHFL formula.
In contrast to the undecidability of \PHFL{} model checking,
the corresponding model-checking problems are \emph{decidable}
for the full fragments of
the $\mu^p$-calculus~\cite{DBLP:conf/stacs/CastroKP15} and
(non-probabilistic) HFL~\cite{DBLP:conf/concur/ViswanathanV04},
with fixpoint alternations.
Thus, the combination of
probabilities and higher-order predicates introduces a new difficulty.

In Section~\ref{introducePA},
we review the definition of  probabilistic automata and the value-1 problem.
Section~\ref{reduction} shows 
 the reduction from the value-1 problem 
to the \PHFL{} model-checking problem.

\subsection{Probabilistic Automata}\label{introducePA}
We review probabilistic automata~\cite{rabin1963probabilistic} and
the undecidability of the value-1 problem.
Our definition follows \cite{DBLP:journals/siglog/Fijalkow17}.

\begin{definition}[Probabilistic Automata]
  A \emph{probabilistic automaton} $\automaton$ %on the alphabet $\alphabet$
  is a quintuple
  $(Q, \alphabet,q_I, \Delta, F)$ where
\begin{itemize}
\item $Q$ is a finite set of states,
\item $\alphabet$ is a finite set of input symbols,
\item $q_I \in Q$ is an initial state,
\item $\Delta : Q \times \alphabet \to D(Q)$, where $D(Q) := \{\, f \colon Q\to [0,1] \mid \sum_{q \in Q} f(q) = 1 \,\}$
  is the set of probabilistic distributions over the set $Q$,
represents transition probabilities, and
\item $F \subseteq Q$ is a set of accepting states.
\end{itemize}
For a word $w=w_1\cdots w_n \in \alphabet^n$,
the probability that \(w\) is accepted by  \(\automaton= (Q,\alphabet,q_I,\Delta,F)\),
written \(\automaton(w)\), is defined by:
%%\tk{This definition does not seem to work for the empty word.}\nk{The definition seems ok to me (and I think this definition has been taken from the reference. For the empty word, \(n=0\), so that \(\prod_{...}\) is an empty product which yields \(1\), and \(\sum_{...}\) is a sum over a singleton.}
\[
\automaton(w) := \sum_{\substack{q_0, \dots, q_{n-1}\in Q, q_n\in F\\\mathit{ s.t.}\ q_0=q_I}} \, \prod_{1\leq i\leq n} \Delta(q_{i-1},w_i)(q_{i}).
\]
%%%of \(\automaton= (Q,\inState{q},\Delta,F)\) is just
%%%a sequence of states $\rho \in Q^{n + 1}$.
%%%%%with the state $\inState{q}$.
%%%The probability \(\Prob{w}{\rho}\)
%%%that the run $\rho = q_0 q_1 \dots q_n$ is generated,
%%%We write 
%%%\[
%%%\prod_{0 \leq i < n} \Delta(q_i, w_{i + 1})(q_{i + 1})
%%%\]
%%%where $q_0 = \inState{q}$.
%%%
%%%A word $w$ is accepted by a probabilistic automaton $\automaton$
%%%if and only if the word generates a run ends in a state in the set $F$.
\end{definition}

%%We denote the probability that the word $w$ is accepted by the automaton $\automaton$ by $\automaton(w)$.
The \emph{value} of a probabilistic automaton $\automaton$, denoted by
$\val(\automaton)$, is defined by
\[
\val(\automaton) := \sup_{w \in \alphabet^\ast} \automaton(w).
\]

The problem of deciding whether \(\val(\automaton)=1\),
called the \emph{value-1 problem}, is known to be undecidable.
\begin{theorem}[Undecidability of The Value-1 Problem~\cite{DBLP:conf/icalp/GimbertO10}]\label{valueUndecidability}
  Given 
  a probabilistic automaton $\automaton$,
  whether
  $\val(\automaton) = 1$ is undecidable.
\end{theorem}
   
\subsection{The Undecidability Result}\label{reduction}
Let $\automaton = (Q, \alphabet,q_I, %\inState{q},
\Delta, F)$ be
a probabilistic automaton, where \(\alphabet=\set{c_1,\ldots,c_{|\alphabet|}}\)
with \(|\alphabet|>0\).
We shall construct
a Markov chain \(M_\A\) and a PHFL formula \(\form_\A\), so that \(\val(A)=1\) if and only if
\(\M_\A\models \form_\A\). The undecidability of \PHFL{} model checking then
follows immediately from Theorem~\ref{valueUndecidability}.

%%For each $q \in Q$, we denote the automaton $(Q, q, \Delta, F)$ by $\automaton_q$.
%%We encode the value $\val \automaton$ by a PHFL formula.
%%
We first construct the Markov chain \(M_\A\).
The set \( \AP \) of atomic propositions is \( \{\, p_c \mid c \in \alphabet \,\} \uplus \{\, p_F \,\} \).
The Markov chain $M_\automaton = (S, P, \iAP, \inState)$ is defined as follows.
\begin{itemize}
\item The set \(S\) of states is \(Q \uplus (Q \times \alphabet)\).
\item The transition probability $P$ is given by:
\begin{align*}
P((q, c), q') &= \Delta(q, c)(q') & \text{($ c \in \alphabet $ and $q, q' \in Q$)} \\
P(q, (q, c)) &= \frac{1}{|\alphabet|} & \text{($ c \in \alphabet $ and $q \in Q$)} \\
P(s, s') &= 0 & \text{(otherwise)}
\end{align*}
The first transition (from \((q,c)\) to \(q'\)) is used to simulate
the transition of \(\A\) from \(q\) to \(q'\) for the input symbol \(c\).
The second transition (from \(q\) to \((q,c)\)) is used to
choose the next input symbol to be supplied to the automaton; the probability is not important, and replacing \( 1/|\alphabet| \) with any non-zero probability does not affect the arguments below.
\item \(\iAP\) is defined by:
%%  must contain propositions $p_{c_i}$ ($1 \leq i \leq |\alphabet|$) and
%%$p_F$ whose interpretations are
\begin{align*}
\iAP(p_{c}) &= \{\, (q, c) \mid q \in Q \,\} &
\iAP(p_F) &= \set{\,q \mid q \in F\,}. %&
\end{align*}
\item The initial state is $\inState = q_I$.
\end{itemize}
Intuitively, the Markov chain \(M_\A\) simulates the behavior of
\(\A\). %, while probabilistically choosing input symbols for \(\A\).
The atomic proposition \(p_{c}\) means that \(\A\) is currently reading the symbol
\(c\), and \(p_F\) means that \(\A\) is in a final state.

Based on this intuition, we now construct the \PHFL{} formula \(\form_\A\).
For each $c \in \alphabet$, we define a formula \(f_{c}\)
of type \(\Quant\to\Quant\) by:
\[
f_{c} := \lambda X.\lozenge (\toQuant{p_{c}} \land \bigcirc X).
\]
Intuitively $f_{c}(\form)$ denotes the probability that the automaton
transits to a state satisfying \(\form\) given \(c\) as the next input.
%%%puts the next input \( c \) to the automaton (\( \lozenge (\toQuant{p_c} \land\,({-})) \)), performs the transition of the automaton in one-step (\( \bigcirc ({-}) \)) and then refers to the value of \( X \).
%
%%In other words, \( f_c \) simulates the behaviour of the automaton \( \A \) reading \( c \).
Given a word \( w = w_1 w_2 \dots w_n \in \alphabet^* \), we define the
formula \( g_w \) by
\[
g_w := f_{w_1} (f_{w_2} ( \dots (f_{w_n} \toQuant{p_F})\dots)).
\]
We write \(\A_q\) for the automaton obtained from \(\A\) by replacing the initial
state with \(q\). The following lemma states that \(g_w\) represents
the probability that \(w\) is accepted by the automaton from the current state \(q\).
\begin{lemma}\label{wordTransitionLemma}
  $\automaton_q(w) = \semant{g_w}_{M_\A}(q)$ for every \( q \in Q \).
\end{lemma}
\begin{proof}
  Let \(\automaton=(Q, \alphabet,q_I, \Delta, F)\).
  The proof proceeds by induction on the length \(|w|\) of \(w \).
  \begin{itemize}
  \item Case where \(|w|=0\), i.e., \(w=\epsilon\):
    By the definition of \(\automaton_q\),
    \(\automaton_q(\epsilon) = 1\) if \(q\in F\) and \(0\) otherwise.
    We have the required result, as \(g_\epsilon = \toQuant{p_F}\).
  \item Case where \(|w|>0\):
    Let \(w=w_1\cdots w_n=w_1w'\).
    We have:
    %    \[\begin{array}{ll}
    \begin{align*}
    A_q(w) &= 
    \sum_{\substack{q_0, \dots, q_{n-1}\in Q, q_n\in F\; %%\\
        \mathit{ s.t.}\ q_0=q}} \, \prod_{1\leq i\leq n} \Delta(q_{i-1},w_i)(q_{i})\\[1ex]
    &= \sum_{q'\in Q} \Delta(q,w_1)(q')\cdot
    (\sum_{\substack{q_1 \dots, q_{n-1}\in Q, q_n\in F\;%%\\
        \mathit{ s.t.}\ q_1=q'}}
    \, \prod_{2\leq i\leq n} \Delta(q_{i-1},w_i)(q_{i}))\\[1ex]
    &= \sum_{q'\in Q} \Delta(q,w_1)(q')\cdot \automaton_{q'}(w').
    \end{align*}
%    \]
    Since \(g_w = f_{w_1}(g_{w'}) \equiv
    \lozenge (\toQuant{p_{w_1}} \land \bigcirc g_{w'})\), 
    we have:
%    \[
    \begin{align*}
      \semant{g_w}_{M_\A}(q)
      &=
         \max_{c\in \Sigma} \semant{\toQuant{p_{w_1}} \land \bigcirc g_{w'}}(q,c)\\
      &= \semant{\bigcirc g_{w'}}(q,w_1)\\
      &= \sum_{q'\in Q} \Delta(q,w_1)(q')\cdot \semant{g_{w'}}(q')
    \end{align*}
%    \]
    By the induction hypothesis, we have \(\automaton_{q'}(w')=\semant{g_{w'}}(q')\),
    which implies the the required result.
    \qedhere
    \end{itemize}
\end{proof}

Using Lemma~\ref{wordTransitionLemma}, we obtain
\(
\val(\A)
=
\sup_{n \in\omega} \semant{\bigvee_{w \in \alphabet^{\le n}} g_w}_{M_\A}(q_I)
%%\quad=\quad
%%\big(\LUB{\Quant} \set{ \semant{\bigvee_{w \in \alphabet^{\le n}} g_w}_{M_\A}\mid n\in\omega}
%%  \big)(q_0)
\),
where \( \alphabet^{\le n} \) is the set of words of length up to \(n \).
This can be expressed by using the least fixpoint operator.
%By rewriting the right hand side using a fixpoint operator, we get the following theorem.
\begin{theorem}\label{PHFLRepresentation}
Let \(\theta_\A\) be the formula of type $\Quant \rightarrow \Quant$ defined by:
\[
\theta_\A := \mu F. \big( \lambda X. X \lor \bigvee_{c \in \alphabet} F\,(f_{c}\,X) \big).
\]
Then $\val(\automaton) = \semant{\theta_\A\,\toQuant{p_F}}_{M_\A}(q_I)$.
Therefore \(M_\A \models \phi_A \) if and only if \( \val(\automaton) = 1 \),
for \(\phi_A:=[\theta_\A\,\toQuant{p_F}]_{\ge 1}\).
\end{theorem}
\begin{proof}
  Let
  \begin{equation*}
    \xi := \lambda F. \lambda X. X \vee \bigvee_{c \in \alphabet} F\,(f_c\,X).
  \end{equation*}
  Then, %%for every Markov chain \( M \) and natural number \( n \),
  we have %%it is easy to verify:\todo{elaborate}
  \begin{equation*}
    \semant{\theta_A}_{M} \quad=\quad
    \semant{\mu F. \xi\,F}_{M}
    %% \semant{\theta_A\,\toQuant{p_F}}_{M} \quad=\quad
    %% \semant{(\mu F. \xi\,F)\,\toQuant{p_F}}_{M}
    \quad=\quad
    \LUB{\Quant\to\Quant}\set{ \semant{\xi^n(\bot)}\mid n\in\omega}
%%    \semant{\xi^n(\mu F. \xi\,F)\,\toQuant{p_F}}_M
%%    \quad\ge\quad \semant{\xi^n(\bot)\,\toQuant{p_F}}_M,
  \end{equation*}
  where \( \bot := \lambda Z. \mu U. U \) is the formula of type \( \Quant \to \Quant \),
%%  , whose interpretation is the identically zero function,
  and \(\xi^n(x)\) denotes \(n\)-times applications of \(\xi\) to \(x\).
  In fact, \(\LUB{\Quant\to\Quant}\set{ \semant{\xi^n(\bot)}\mid n\in\omega}\)
  is a fixpoint of \(\semant{\xi}\), because:
  \begin{align*}
&    \semant{\xi}(\LUB{\Quant\to\Quant}\set{ \semant{\xi^n(\bot)}\mid n\in\omega})\\
    &    = \lambda x\in D_\Quant.
    x\LUB{\Quant}
    \big(  \LUB{}_{c\in \alphabet} (\LUB{\Quant\to\Quant}\set{ \semant{\xi^n(\bot)}\mid n\in\omega})(\semant{f_c}\; x)\big)\\
    &    = \lambda x\in D_\Quant.
    x\LUB{\Quant}
    \big(  \LUB{}_{c\in \alphabet} (\LUB{\Quant}\set{ \semant{\xi^n(\bot)}(\semant{f_c}\; x)\mid n\in\omega})\big)\\
    &    = \lambda x\in D_\Quant.
    \LUB{\Quant}\set{
    x\LUB{\Quant}
    \big(  \LUB{}_{c\in \alphabet} \semant{\xi^n(\bot)}(\semant{f_c}\; x)\big)\mid n\in\omega}\\
    &    = 
    \LUB{\Quant\to\Quant}\set{ \semant{\xi^{n+1}(\bot)}\mid n\in\omega}.\\
  \end{align*}
  Since \(\semant{\xi}\) is monotonic and
  \(\semant{\bot}\) is the least element, we also have:
%%\(\LUB{\Quant}\set{ \semant{\xi^{n+1}(\bot)}\mid n\in\omega}\)
  \[  \semant{\mu F. \xi\,F}_{M} = \semant{\xi^n(\mu F. \xi\,F)} \GEQ \semant{\xi^n(\bot)}\] for any \(n\in\omega\)
  hence also \[
\semant{\mu F. \xi\,F}_{M}
    %% \semant{\theta_A\,\toQuant{p_F}}_{M} \quad=\quad
    %% \semant{(\mu F. \xi\,F)\,\toQuant{p_F}}_{M}
    \GEQ
    \LUB{\Quant\to\Quant}\set{ \semant{\xi^n(\bot)}\mid n\in\omega}.\]
  Thus, we have the equality.
%%%  Furthermore we have
%%%  \begin{equation*}
%%%    \semant{\theta_A}_{M} \quad=\quad
%%%    \LUB{\Quant\to\Quant}\set{ \semant{\xi^n(\bot)}\mid n\in\omega}
%%%  \end{equation*}
  %%%  because \( \xi \) is ``continuous'' provided that the
%%%  argument \( F \) ranges over continuous functions.
  
  By a straightforward induction on \(n\), we also have:
  \( % \begin{equation*}
    \semant{\xi^{n+1}(\bot)\,\toQuant{p_F}}_{M}
    =
    \semant{\bigvee_{w \in \alphabet^{\le n}} g_w}_{M}
\).
%  \end{equation*}
  Therefore, by using also Lemma~\ref{wordTransitionLemma}, we obtain:
  \begin{equation*}
    \val(\A)
    \quad=\quad
    \sup_n (\semant{\bigvee_{w \in \alphabet^{\le n}} g_w}_{M_\A}(q_I))
    \quad=\quad
%    (\LUB{\Quant}\set{ \semant{\xi^n(\bot)\toQuant{p_F}}_{M_\A}\mid n\in\omega}(q_I)
    \sup_n (\semant{\xi^{n+1}(\bot)\toQuant{p_F}}(q_I))
    \quad=\quad
    \semant{\theta_A\,\toQuant{p_F}}_{M_\A}(q_I),
  \end{equation*}
  which implies the required result.
\end{proof}

The following is an immediate corollary of
Theorems~\ref{valueUndecidability} and \ref{PHFLRepresentation}.
\begin{corollary}[Undecidability of PHFL Model-Checking Problem]
  There is no algorithm that, given a Markov chain $M$ and a closed order-1
  formula $\phi$ of type $\Qual$, decides whether $M\models \phi$.
\end{corollary}

We close this section with some remarks.\footnote{We would like to thank an
     anonymous reviewer of our FSCD 2020 submission for pointing them out.}
\begin{remark}
  Note that the value \(\val(\A)\) of a probabilistic automaton
  cannot even be approximately computed~\cite{DBLP:journals/siglog/Fijalkow17}:
  there is no algorithm that outputs ``Yes'' if \(\val(\A)=1\) and
  ``No'' if \(\val(\A)\le\frac{1}{2}\).
  Thus, the proof of Theorem~\ref{PHFLRepresentation} (in particular,
  the result \(\val(\A)=\semant{\theta_A\,\toQuant{p_F}}_{M_\A}(q_I)\))
  also implies that for a qualitative formula of PHFL \(\psi\),
  \(\semant{\psi}\) is not approximately computable in general.
\end{remark}

\begin{remark}
  It would be interesting to study a converse encoding, i.e., to find
  an encoding of some
  fragment of the PHFL model checking problem
   into the value-1 problem.
   Such an encoding may help us find a decidable
   class of the PHFL model checking problem, based on decidable subclasses
   for the value-1 problem, such as the one studied
   in \cite{DBLP:journals/corr/FijalkowGKO15}.
%%   \nk{Since PCP can be reduced to the value-1 problem,
%%     we can indirectly encode the PHFL
%%     model checking problem for a formula of the form\([\form]_J\)
%%     where \(\form\) is an order-1, \(\mu\)-only, and \([\cdot]_J\)-free formula?}
\end{remark}

\section{Hardness of the PHFL Model-Checking Problem}\label{section:hardness}
In the previous section, we have seen that
PHFL model checking is undecidable
even for the fragment of PHFL without fixpoint alternations.
In this section, 
we give a lower bound of the hardness of the PHFL model-checking problem
in the presence of fixpoint alternations.
The following theorem states the main result of this section.
\begin{theorem}\label{thm:analytical}
The order-1 PHFL model-checking problem is \(\Pi^1_1\)-hard and \(\Sigma^1_1\)-hard.
\end{theorem}
Note that \( \Pi^1_1 \) and \( \Sigma^1_1 \), defined in terms of the second-order arithmetic, contain very hard problems.
For example, those classes contain
the problem of deciding whether a given first-order Peano arithmetic formula is true.

We prove this theorem by reducing
the validity checking problem of the \(\mu\)-arithmetic~\cite{DBLP:conf/lics/Lubarsky89}
to the PHFL model-checking problem.
Even the validity checking problem of a higher-order extension of the \(\mu\)-arithmetic
can be reduced to the PHFL model-checking problem. %%, as we shall see.
The key in the proof
is a representation of natural numbers as quantitative propositions such that all
the operations on natural numbers in the \( \mu \)-arithmetic are expressible in PHFL.
%% In the reduction, we encode the natural number \(n\) in \(\mu\)-arithmetic
%% by \(v \in \semant{\Quant}\) satisfying
%% \(v(s_0) = \frac{1}{2^n}\) and \(v(s'_0) = 1 - \frac{1}{2^n}\)
%% for some states \(s_0\) and \(s'_0\).
%% A \(\mu\)-arithmetic formula of form \(s \leq t\),
%% where \(\muSemant{s} = n\) and \(\muSemant{t} = m\),
%% is translated to a PHFL formula of form \([\phi]_{\geq 1 / 2}\)
%% where \(\semant{\phi} = \frac{1}{2} + \frac{1}{2^{n + 1}} - \frac{1}{2^{m + 1}}\).
%% Here we used the fact that
%% the inequality \(n \leq m\) is equivalent to
%% the inequality \(\frac{1}{2} + \frac{1}{2^{n + 1}} - \frac{1}{2^{m + 1}} \geq \frac{1}{2}\).

This section is structured as follows.
Section~\ref{subsec:muArith} reviews
 the basic notions of the \(\mu\)-arithmetic.
Section~\ref{subsec:encodeMu} describes the reduction
and proves the theorem above.

\subsection{Higher-Order Fixpoint Arithmetic}\label{subsec:muArith}
The \(\mu\)-arithmetic~\cite{DBLP:conf/lics/Lubarsky89} is a first-order arithmetic with fixpoint operators.
This section briefly reviews its higher-order extension, studied by Kobayashi et al.~\cite{DBLP:conf/esop/0001TW18}.

As in PHFL, we first define the types of \(\mu\)-arithmetic formulas.
The set of \emph{types}, ranged over by \(A\), is given by:
%\subsection{Syntax of Types}
\begin{align*}
A &::= N \bnf T &
T &::= \muProp \bnf A \to T.
\end{align*}
The type \(N\) is for natural numbers,
\(\muProp\) for (qualitative) propositions,
and \(A \to T\) for functions.
We do not allow functions to return values of type \(N\).
We define the order of types of the \(\mu\)-arithmetic similarly to the PHFL types,
by: \(\order(N)=\order(\muProp)=0\)
and \(\order(A\to T)=\max(\order(A)+1, \order(T))\).

%\subsection{Syntax of Formulas}
%We now give the syntax of \(\mu\)-arithmetic formulas.
Assume a countably infinite set \( \Var \) of variables ranged over by \(X\).
The set of formulas, ranged over by \(\muform\), is given by the following grammar.
\begin{align*}
  s &::= X \bnf Z \bnf S s
  &
  \muform &::=
  s \bnf %X \bnf 
  %s < s \bnf
  s_1 \leq s_2 \bnf
  \muform_1 \land \muform_2 \bnf
  \muform_1 \lor \muform_2 \bnf
  \lambda X. \muform \bnf
  \muform_1 \, \muform_2 \bnf
  \mu X. \muform \bnf
  \nu X. \muform.
\end{align*}
Here, \(Z\) and \(S\) respectively denote the constant \(0\)
and the successor function on natural numbers.
\begin{figure}
\begin{minipage}{100pt}
\AxiomC{$\mathstrut$}
\UnaryInfC{$\tyEnv,X\COL A \muVdash X : A$}
\DisplayProof
\end{minipage}
\begin{minipage}{75pt}
\AxiomC{$\mathstrut$}
\UnaryInfC{$\tyEnv \muVdash Z : N$}
\DisplayProof
\end{minipage}
\vspace{10pt}
\begin{minipage}{75pt}
\AxiomC{$\tyEnv \muVdash s : N$}
\UnaryInfC{$\tyEnv \muVdash S \, s : N$}
\DisplayProof
\end{minipage}
\begin{minipage}{100pt}
\AxiomC{\(\tyEnv \muVdash s : N\)}
\AxiomC{\(\tyEnv \muVdash t : N\)}
\BinaryInfC{\(\tyEnv \muVdash s \leq t : \muProp\)}
\DisplayProof
\end{minipage}
\vspace{10pt}
\begin{minipage}{150pt}
\AxiomC{\(\tyEnv \muVdash \muform_1 : \muProp\)}
\AxiomC{\(\tyEnv \muVdash \muform_2 : \muProp\)}
\BinaryInfC{\(\tyEnv \muVdash \muform_1 \land \muform_2 : \muProp\)}
\DisplayProof
\end{minipage}
\begin{minipage}{150pt}
\AxiomC{\(\tyEnv \muVdash \muform_1 : \muProp\)}
\AxiomC{\(\tyEnv \muVdash \muform_2 : \muProp\)}
\BinaryInfC{\(\tyEnv \muVdash \muform_1 \lor \muform_2 : \muProp\)}
\DisplayProof
\end{minipage}
\begin{minipage}{100pt}
\AxiomC{$\Gamma, X : A \muVdash \muform : T$}
\UnaryInfC{$\Gamma \muVdash \lambda X. \muform : A \rightarrow T$}
\DisplayProof
\end{minipage}
\begin{minipage}{180pt}
\AxiomC{$\Gamma \muVdash \muform_1 : A \rightarrow T$}
\AxiomC{$\Gamma \muVdash \muform_2 : A$}
\BinaryInfC{$\Gamma \muVdash \muform \, \muform_2 : T$}
\DisplayProof
\end{minipage}
\vspace{10pt}
\begin{minipage}{100pt}
\AxiomC{$\Gamma, X : T \muVdash \muform : T$}
\UnaryInfC{$\Gamma \muVdash \mu X. \muform : T$}
\DisplayProof
\end{minipage}
\begin{minipage}{100pt}
\AxiomC{$\Gamma, X : T \muVdash \muform : T$}
\UnaryInfC{$\Gamma \muVdash \nu x. \muform : T$}
\DisplayProof
\end{minipage}
\caption{Typing Rules for the Higher-order Fixpoint Arithmetic.}\label{fig:muArithTypes}
\end{figure}

The typing rules are shown in Fig.~\ref{fig:muArithTypes}; they are just standard
typing rules for the simply-typed \(\lambda\)-calculus, with
several constructors such as \(Z\COL N\), \(S\COL N\to N\), and \(\land: \muProp\to\muProp\).
We shall consider only well-typed formulas. We define the \emph{order} of a formula as
the largest order of the types of its subformulas.
%% Note that taking fixpoint operators of type \(N\) is not allowed.
%% since
%% the semantics \(\muSemant{N}\) of type \(N\) is not a complete lattice, as explained below.

%% We can show that for any type environment \(\tyEnv\), formula \(\phi\) and type \(A\),
%% if \(\tyEnv \muVdash \phi : A\) is derivable then its derivation is unique, by induction on the structure of the formula \(\phi\).

%We now give the semantics of \(\mu\)-arithmetic.
\begin{definition}[Semantics of Types]
%The semantics of a type \(A\) is a partially ordered set $\semant{A} = (D_A, \leq_A)$ defined inductively
The semantics of a type \(\muTy\) is a partially ordered set \(\muSemant{A} = (D_\muTy, \sqleq_\muTy)\) defined inductively
on the structure of \(\muTy\) as follows.
\begin{enumerate}
\item The semantics of \(N\) and \(\muProp\): % are defined as follows.
\begin{align*}
D_N &= \mathbb{N} &
n \sqleq_N m &\defarrow n = m \\
D_\muProp &= \{0, 1\} &
p \sqleq_\muProp q &\defarrow p \leq q
\end{align*}
\item The semantics of \(A \to T\): % is defined as follows.
\begin{align*}
D_{A \to T} &= \{\, f : D_A \to D_T \mid 
\forall u, v \in D_{A}. u \sqleq_{A} v \implies f(u) \sqleq_{T} f(v) \,\} \\
f \sqleq_{A \to T} g &\defarrow \forall v \in D_{A}. f(v) \sqleq_{T} g(v)
\end{align*}
\end{enumerate}
\end{definition}
The semantics \( \muSemant{T} \) of a type \( T \) forms a complete lattice
(while \(\muSemant{N}\) is not);
we write \(\LUB{T}\) (resp. \(\GLB{T}\)) for the least upper bound (resp. greatest lower bound) operation,
and \(\bot_T\) for the least element.

The interpretation \( \muSemant{\tyEnv} \) of a type environment \( \tyEnv \) is the set of functions \( \muInterpret \) such that \(\dom(\muInterpret) = \dom(\tyEnv)\) and that \(\muInterpret(X) \in \muSemant{\tyEnv(X)}\) for every \(X \in \dom(\tyEnv)\).
It is ordered by the point-wise ordering.

\begin{definition}[Semantics of Formulas]
  The semantics of a formula \(\muform\) with judgment
   \(\tyEnv \muVdash \muform : \muTy\) is a monotone map
   from \(\muSemant{\tyEnv}\) to \(\muSemant{\muTy}\), defined as follows.
   \rm
\begin{align*}
\muSemant{\tyEnv \muVdash X : A}(\muInterpret) &:= \muInterpret(X) \\
\muSemant{\tyEnv \muVdash Z : N}(\muInterpret) &:= 0 \\
\muSemant{\tyEnv \muVdash S s : N}(\muInterpret) &:= \muSemant{\tyEnv \muVdash s : N}(\muInterpret) + 1 \\
%\muSemant{\tyEnv \muVdash X : A}(\muInterpret) &:= \muInterpret(X) \\
\muSemant{\tyEnv \muVdash s \leq t : \muProp}(\muInterpret) &:=
\begin{cases}
  1 & \mbox{(if \( \muSemant{\tyEnv \muVdash s : N}(\muInterpret) \leq \muSemant{\tyEnv \muVdash t : N}(\muInterpret) \))} \\
  0 & \mbox{(if \( \muSemant{\tyEnv \muVdash s : N}(\muInterpret) > \muSemant{\tyEnv \muVdash t : N}(\muInterpret) \))}
\end{cases}
\\
\muSemant{\tyEnv \muVdash \muform_1 \land \muform_2 : \muProp}(\muInterpret) &:=
\muSemant{\tyEnv \muVdash \muform_1 : \muProp}(\muInterpret) \LUB{\muProp}
\muSemant{\tyEnv \muVdash \muform_2 : \muProp}(\muInterpret) \\
\muSemant{\tyEnv \muVdash \muform_1 \lor \muform_2 : \muProp}(\muInterpret) &:=
\muSemant{\tyEnv \muVdash \muform_1 : \muProp}(\muInterpret) \GLB{\muProp}
\muSemant{\tyEnv \muVdash \muform_2 : \muProp}(\muInterpret) \\
\muSemant{\tyEnv \muVdash \lambda X. \muform : A \to T}(\muInterpret) &:=
\lambda v \in \muSemant{A}. \muSemant{\tyEnv, X : A \muVdash \muform : T}(\muInterpret[X \mapsto v]) \\
\muSemant{\tyEnv \muVdash \muform_1 \, \muform_2 : T}(\muInterpret) &:=
\muSemant{\tyEnv \muVdash \muform_1 : A \to T}(\muInterpret) \, (\muSemant{\tyEnv \muVdash \muform_2 : A}(\muInterpret)) \\
\muSemant{\tyEnv \muVdash \mu X. \muform : T}(\muInterpret) &:=
\LFP(\lambda v \in D_T.\semant{\tyEnv, X : T \muVdash \muform : T}(\muInterpret[X \mapsto v])) \\
%\bigwedge \{v \in D_{A} | \semant{\tyEnv, x : A \muVdash \muform : A}(\muInterpret[x \mapsto v]) \leq v \} \\
\muSemant{\tyEnv \muVdash \nu X. \muform : T}(\muInterpret) &:=
%\bigvee \{v \in D_{A} | \semant{\tyEnv, x : A \muVdash \muform : A}(\muInterpret[x \mapsto v]) \geq v \}
\GFP(\lambda v \in D_T. \semant{\tyEnv, X : T \muVdash \muform : T}(\muInterpret[X \mapsto v]))
\end{align*}
\end{definition}
As in the case of PHFL, we write \(\muSemant{\muform}(\muInterpret)\) for
\(\muSemant{\tyEnv \muVdash \muform : A}(\muInterpret)\)
and
just \(\muSemant{\muform}\) for \(\muSemant{\muform}(\emptyset)\) 
when there is no confusion.

% Although we excluded operators \(<\), \(=\), \(\neq\) and \(\lnot\) from the arithmetic,
% we can encode these operators (with standard semantics) using the following facts.
% \begin{itemize}
% \item The formula \(s < t\) is equivalent to \(S s \leq t\)
% \item The formula \(s = t\) is equivalent to \((s \leq t) \land (t \leq s)\)
% \item The formula \(s \neq t\) is equivalent to \((s < t) \lor (t < s)\)
% \item The formula \(\lnot \muform\), where \(\muform\) is typed as \(\muProp\) under some type environment, can be transformed
% to a formula without \(\lnot\) operators using De Morgan dual repeatedly.
% \end{itemize}

\begin{example}
Let \(\muform = \mu F. \lambda X. (X = 100 \lor F (S (S \, X)))\) where
\(100\) is an abbreviation of the term
\(\underbrace{S(S(\dots S}_{100} \, Z) \dots )\).
The semantics \(\muSemant{\muform}\) is a function \(f : \mathbb{N} \to \{0, 1\}\)
where \(f(n) = 1\) just if
\(n\) is an even number no greater than \(100\).
\end{example}

The \emph{validity checking problem} of the higher-order fixpoint arithmetic is the problem
of,
given a closed formula \(\muform\) of type \(\muProp\),
deciding whether \(\muSemant{\muform} = 1\).
The following result is probably folklore,
which follows from the well-known fact that
the \emph{fair termination problem for programs} is \( \Pi^1_1 \)-complete (see, e.g., Harel~\cite{DBLP:journals/jacm/Harel86}),
and the fact that the fair termination of a program can be reduced to the validity of
a first-order fixpoint arithmetic formula (see, e.g.,
~\cite{DBLP:conf/esop/0001TW18} for the reduction).
%%but we include a proof for the sake of completeness.
%the result actually holds already for the order-1 fragment (i.e., Lubarsky's \(\mu\)-arithmetic).
\begin{theorem}\label{thm:muArithHardness}
The validity checking problem of the first-order fixpoint arithmetic is \(\Pi^1_1\)-hard
and \(\Sigma^1_1\)-hard.
\end{theorem}
%%%\begin{proof}
%%%  It is well-known that the \emph{fair termination problem for programs} is \( \Pi^1_1 \)-complete (see, e.g., Harel~\cite{DBLP:journals/jacm/Harel86}).
%%%  Kobayashi et al.~\cite{DBLP:conf/esop/0001TW18} have given a reduction from temporal verification of programs, including fair termination as an instance, to the validity problem of higher-order fixpoint arithmetic.
%%%  Hence the validity problem of higher-order fixpoint arithmetic is \( \Pi^1_1 \)-hard.
%%%  It is also \( \Sigma^1_1 \)-hard since formulas are closed under negation.\footnote{Although negation is not a primitive logical connective of the logic, it is a definable operation in a certain sense.  See also the remark in Section~\ref{sec:syntax-of-phfl}.}
%%%  
%%%  The above argument is almost independent of the choice of the programming language;\footnote{The requirements are (1) that the language should be sufficiently expressive, namely it is Turing-complete and (2) that the language is executable by computers, i.e.~it has no incomputable instructions.}
%%%  let us choose a language of the lowest order, namely while-language.
%%%  For programs of while-language, the reduction by Kobayashi et al.~\cite{DBLP:conf/esop/0001TW18} yields first-order formulas.
%%%\end{proof}
\begin{remark}
  As for an upper bound,
  Lubarsky~\cite{DBLP:conf/lics/Lubarsky89} has shown that predicates on natural numbers definable by \(\mu\)-arithmetic formulas
  belong to \(\Delta^1_2\).
  One can prove that the validity problem for the \( \mu \)-arithmetic is \(\Delta^1_2\) as well. %, although this result dose not immediately follow from Lubarsky's result.
\end{remark}

\subsection{Hardness of PHFL Model Checking}\label{subsec:encodeMu}
We give a reduction from the validity checking problem of the higher-order fixpoint arithmetic
to the PHFL model-checking problem.
The main theorem of this section (Theorem~\ref{thm:analytical}) is an immediate consequence of this reduction and Theorem~\ref{thm:muArithHardness}.
%We show that any formula of \(\mu\)-arithmetic can be encoded as a formula of PHFL.
%In the translation, we encode the value \(n \in \mathbb{N}\) in \(\mu\)-arithmetic
%into an element in \(\semant{\Prop}\) whose value in the initial state is \(1 / 2^n\).
%To encode
%A \(\mu\)-arithmetic formula of the form \(s \leq t\)
%is encoded as
%we use 
%a PHFL formula of form \([\muform_{s, t}]_{\geq 1 / 2}\).

Given a formula \( \muform \) of the higher-order fixpoint arithmetic, we need to effectively construct a pair
\( (\form, M) \) of a formula
of \PHFL{} and a Markov chain such that \( \muform \) is true if and only if \( M \models \form \).
The Markov chain \(M\) is independent of the formula \( \muform \).
We first define the Markov chain and then explain the intuition of the translation of formulas.

The Markov chain \(M = (S, P, \iAP, s_{\text{in}})\) is shown in Figure~\ref{fig:encodeMu}.
It is defined as follows.
\begin{itemize}
\item The set of states is \(S = \{s_0, s'_0, s_1, s_1'\}\).
\item The transition probability satisfies \(P(s_0, s_1) = P(s_0, s'_0) = %1 / 2\),\(
  P(s'_0, s_0) = P(s'_0, s_1') = \frac{1}{2}\),
\(P(s_1, s_0) = P(s_1', s'_0) = 1\)
and \(P(s_i, s_j) = 0\) for all other pairs of states.
\item There are four atomic propositions %\( \AP = \{ p_0, p_0',p_1,p_1'\} \)
  \(p_0, p_0',p_1\), and \(p_1'\),
  representing each state (e.g.~\( \iAP(p_0) = \{ s_0 \} \)).
\item The initial state \(\inState\) is \(s_0\).
\end{itemize}
\begin{figure}
  \centering
\setlength\unitlength{10pt}
\begin{picture}(12, 12)(0, 0)
\put(3, 3){\circle{2}}
\put(2.5, 2.75){\(s'_0\)}
\put(3, 8){\circle{2}}
\put(2.5, 7.75){\(s_0\)}
\put(9, 2){\circle{2}}
\put(8.5, 1.75){\(s_1'\)}
\put(9, 9){\circle{2}}
\put(8.5, 8.75){\(s_1\)}
\qbezier(2.4, 3.8)(2, 5.5)(2.4, 7.2)
\put(2.4, 7.2){\vector(1, 4){0}}
\put(0.6, 5.5){\(0.5\)} % move label a bit away from arrow
\qbezier(3.6, 7.2)(4, 5.5)(3.6, 3.8)
\put(3.6, 3.8){\vector(-1, -4){0}}
\put(4.2, 5.5){\(0.5\)}
\qbezier(3.8, 8.6)(5.5, 9)(8, 9)
\put(8, 9){\vector(1, 0){0}}
\put(5.5, 9.2){\(0.5\)}
\qbezier(8.2, 8.4)(6, 8)(4, 8)
\put(4, 8){\vector(-1, 0){0}}
\put(6, 7){\(1\)}
\qbezier(3.8, 2.4)(5.5, 2)(8, 2)
\put(8, 2){\vector(1, 0){0}}
\put(5.5, 1){\(0.5\)}
\qbezier(8.2, 2.6)(6, 3)(4, 3)
\put(4, 3){\vector(-1, 0){0}}
\put(6, 3.2){\(1\)}
\end{picture}
\caption{The Markov Chain for Reduction from Higher-order Fixpoint Arithmetic to PHFL.}
\label{fig:encodeMu}
\end{figure}

For notational convenience, we write \( v \in \semant{\Quant}_M\) as a tuple \((v(s_0), v(s'_0), v(s_1), v(s_1'))\).
%We also denote the value \(1 / 2^n\) by \(p_n\), for each natural number \(n\).

As mentioned at the beginning of this section, the key of the reduction is the representation of natural numbers, as well as operations on natural numbers.
We encode a propositional formula \(\muform\)
into a quantitative propositional formula \(\form\) such that
\(\semant{\form}_M = (\muSemant{\muform},\_,\_,\_)\),
and encode  a natural number \( n \) into a quantitative propositional formula
\( \psi \) such that
\(\semant{\form}_M = (\frac{1}{2^n},1 - \frac{1}{2^n}, \_,\_)\).
%%\begin{equation*}
%%  \semant{\psi}_{M}(s_0) = \frac{1}{2^n}
%%  \qquad
%%  \mbox{and}
%%  \qquad
%%  \semant{\psi}_{M}(s_0') = 1 - \frac{1}{2^n}.
%%\end{equation*}
Here, \(\_\) denotes a ``don't care'' value.
%%The values of \( \psi \) on \( s_1 \) and \( s_1' \) are arbitrary; we shall write \( \_ \) for this kind of unimportant values.
We implement primitives on natural numbers \( Z \), \( S \) and \( \le \), as follows.

The constant \( Z \) can be represented by \( \toQuant{p_0} \).
Indeed, \( \semant{\toQuant{p_0}}_M = (1, 0, 0, 0)= (1/2^0, 1-(1/2^0), 0, 0) \) as expected.
%then \( \semant{\toQuant{p_0'}}_M = (0, 1, 0, 0) \) as expected.

Assuming that \( \form \) represents
\( n \) (i.e.~\( \semant{\form}_M = (1/2^n, 1-(1/2^n), \_, \_) \)), the successor \( n+1 \) can be represented by
\begin{equation*}
  \form' \quad:=\quad \bigcirc((\bigcirc \form \land (p_1 \lor p_1')) \lor p_0).
\end{equation*}
Indeed, we have:
\begin{align*}
  \semant{\bigcirc \form}_M &= (\_, \_, \frac{1}{2^n}, 1 - \frac{1}{2^n}) \\
  \semant{\bigcirc \form \land (p_1 \lor p_1')}_M &= (0, 0, \frac{1}{2^n}, 1 - \frac{1}{2^n}) \\
  \semant{(\bigcirc \form \land (p_1 \lor p_1')) \lor p_0}_M &= (1, 0, \frac{1}{2^n}, 1- \frac{1}{2^n}) \\
  \semant{\bigcirc((\bigcirc \form \land (p_1 \lor p_1')) \lor p_0)}_M &= (\frac{1}{2} \times \frac{1}{2^n}, \frac{1}{2} + \frac{1}{2} \times (1 - \frac{1}{2^n}), \_, \_)\\
  &= (\frac{1}{2^{n+1}}, 1 - \frac{1}{2^{n+1}}, \_, \_).
\end{align*}

It remains to encode \(\leq\).
We use the fact that, for any natural numbers \( n \) and \( m \), %%we have
\begin{equation*}
  n \le m
  \quad\Leftrightarrow\quad
  \frac{1}{2^n} \ge \frac{1}{2^m}
  \quad\Leftrightarrow\quad
%%  \frac{1}{2^n} - \frac{1}{2^m} \ge 0
%%  \quad\Leftrightarrow\quad
  \frac{1}{2^n} + (1 - \frac{1}{2^m}) \ge 1.
\end{equation*}
The \( s_0' \)-component of the representation of a natural number plays an important role
below. 
Assume that \( \form \) and \( \chi \) represent \( n \) and \( m \) respectively.
Then we have
\begin{equation*}
  \semant{\bigcirc \form \land p_1}_M = (0, 0, \frac{1}{2^n}, 0)
  \qquad
  \semant{\chi \land p_0'}_M = (0, 1-\frac{1}{2^m}, 0, 0)
\end{equation*}
and thus
\begin{equation*}
  \semant{(\bigcirc \form \land p_1) \lor (\chi \land p_0')}_M = (0, 1-\frac{1}{2^m}, \frac{1}{2^n}, 0).
\end{equation*}
Therefore
\begin{equation*}
  \semant{\bigcirc ((\bigcirc \form \land p_1) \lor (\chi \land p_0'))}_M = (\frac{1}{2} \times \big(\frac{1}{2^n} + (1-\frac{1}{2^m}) \big), \_, \_, \_).
\end{equation*}
Thus, \( n \le m \) if and only if the \( s_0 \)-component of the above formula is \( \ge \frac{1}{2} \).
In other words, \(n\le m\) just if:
\begin{equation*}
  \semant{[\bigcirc ((\bigcirc \form \land p_1) \lor (\chi \land p_0'))]_{\frac{1}{2}}}_M =
  (1,\_,\_,\_).
\end{equation*}

Let us formalize the above argument.
We first give the translation of types:
\begin{align*}
\tr(N) &= \Quant &
\tr(\muProp) &= \Qual &
\tr(A \to T) &= \tr(A) \to \tr(T). %& \text{for \(A \neq N\)}
\end{align*}
The translation can be naturally extended to type environments.
%
%% For a type environment \(\tyEnv\) of \(\mu\)-arithmetic, we define a type environment \(\tr(\tyEnv)\) of PHFL as one satisfying
%% \(\dom(\tr(\tyEnv)) = \dom(\tyEnv)\)
%% and
%% \(\tr(\tyEnv)(X) = \tr(\tyEnv(X))\) for all \(X \in \dom(\tyEnv)\).
%
Following the above discussion, the translation of formulas of type \( N \) is given by
\begin{equation*}
  \tr(Z) = \toQuant{p_0}
  \qquad\mbox{and}\qquad
  \tr(S \, s) = \bigcirc((\bigcirc \tr(s) \land (p_1 \lor p_1')) \lor p_0).
\end{equation*}
The comparison operator can be translated as follows:
\begin{equation*}
  \tr(s \leq t) =
     [\bigcirc ((\bigcirc \tr(s) \land p_1) \lor (\tr(t) \land p_0'))]_{\geq \frac{1}{2}}.
\end{equation*}
The translation of other connectives is straightforward:
\begin{gather*}
  \tr(\muform_1 \land \muform_2) = \tr(\muform_1) \land \tr(\muform_2)
  \qquad
  \tr(\muform_1 \lor \muform_2) = \tr(\muform_1) \lor \tr(\muform_2)
  \qquad
  \tr(\lambda X. \muform) = \lambda X. \tr(\muform)
  \\
  \tr(X) = X
  \qquad
  \tr(\muform_1 \, \muform_2) = \tr(\muform_1) \, \tr(\muform_2)
  \qquad
  \tr(\mu X. \muform) = \mu X. \tr(\muform)
  \qquad
  \tr(\nu X. \muform) = \nu X. \tr(\muform).
\end{gather*}

The following lemma states that the translation preserves types.
\begin{lemma}
  If \(\tyEnv \muVdash \muform : A\), then \(\tr(\tyEnv) \vdash \tr(\muform) : \tr(A)\). \end{lemma}
\begin{proof}
  This follows by straightforward induction on the derivation of \(\tyEnv \muVdash \muform : A\).
\end{proof}

%\subsubsection{Semantical Relation}
\label{subsubsec:encodeMuSemant}

We prove the correctness of the translation.
For each type \(\muTy\) of the higher-order fixpoint arithmetic, we define a relation
\(({\semantRel{\muTy}}) \subseteq \muSemant{\muTy} \times \semant{\tr(\muTy)}_M\) by induction on \(\muTy\) as follows:
\begin{align*}
  n \semantRel{N} (r_0, r_0', r_1, r_1')
  &\quad\defarrow\quad
  r_0 = \frac{1}{2^n} \mbox{ and } r_0' = 1 - \frac{1}{2^n}
  \\
  b \semantRel{\muProp} (r_0, r_0', r_1, r_1')
  &\quad\defarrow\quad
  b = r_0
  \\
  f \semantRel{A \to T} g
  &\quad\defarrow\quad
  \forall x \in \muSemant{A}. \forall y \in \semant{\tr(A)}_M.\; x \semantRel{A} y \implies f \, x \semantRel{T} g\, y.
\end{align*}
%% \semantRel{N} &=
%% \left\{\left(n, \left(\frac{1}{2 ^ n}, 1 - \frac{1}{2 ^ n}, 0, 0\right)\right) | n \in \muSemant{N}\right\} \\
%% \semantRel{\muProp} &= \{(p, (p, 0, 0, 0)) | p \in \muSemant{\muProp}\}
%% \end{align*}
%% \begin{enumerate}
%% \item \( n \semantRel{N} (r_0, r_0', r_1, r_1') \) if and only if and \(\semantRel{\muProp}\) are defined as follows.
%% \begin{align*}
%% \semantRel{N} &=
%% \left\{\left(n, \left(\frac{1}{2 ^ n}, 1 - \frac{1}{2 ^ n}, 0, 0\right)\right) | n \in \muSemant{N}\right\} \\
%% \semantRel{\muProp} &= \{(p, (p, 0, 0, 0)) | p \in \muSemant{\muProp}\}
%% \end{align*}
%% \item The relation \(\semantRel{A \to T}\) is defined as follows.
%% \[
%% f \semantRel{A \to T} f' \defarrow
%% \forall g \in \muSemant{A}, g' \in \semant{\tr(A)}. g \semantRel{A} g' \implies f \, g \semantRel{T} f' \, g'
%% \]
%% \end{enumerate}
%%
This relation can be naturally extended to the interpretations of type environments: given a type environment \(\tyEnv\) of the \(\mu\)-arithmetic,
the relation \( ({\semantRel{\tyEnv}}) \subseteq \muSemant{\tyEnv} \times \semant{\tr(\tyEnv)}_M\) is defined by
\[
\muInterpret \semantRel{\tyEnv} \interpret
\quad\defarrow\quad
\forall X \in \dom(\tyEnv). \; \muInterpret(X) \semantRel{\tyEnv(X)} \interpret(X).
\]

The following theorem states the correspondence between
the source and the target of the translation.
A proof is provided in Appendix~\ref{sec:proofs}.
\begin{theorem}\label{thm:semantRel}
Let \( \Gamma \muVdash \muform : \muTy \) be a formula of the higher-order fixpoint arithmetic.
Assume \(\muInterpret \in \muSemant{\tyEnv}\) and \(\interpret \in \semant{\tr(\tyEnv)}\).
If
\(\muInterpret \semantRel{\tyEnv} \interpret\),
then
\(
\muSemant{\tyEnv \muVdash \muform : \muTy}(\muInterpret) \semantRel{\muTy} \semant{\tr(\tyEnv) \vdash \tr(\muform) : \tr(\muTy)}_M(\interpret)
\).
\end{theorem}

%%To prove the theorem above, we first prepare the following lemma.
%%\section{Proof of Theorem~\ref{thm:semantRel}}\label{section:proofAnalytical}

\begin{corollary}\label{cor:semantRel}
  The validity problem of the order-\(k\) fixpoint arithmetic 
%  (where \(k>0\))\tk{I wonder why this assumption is needed.}
is reducible to the order-\(k\)
PHFL model-checking problem.
\end{corollary}
\begin{proof}
  Assume \( \emptyset \muVdash \muform : \muProp \).
  By Theorem~\ref{thm:semantRel}, \( \muSemant{\muform} \semantRel{\muProp} \semant{\tr(\muform)}_M \).
  Therefore, \( \muSemant{\muform} = 1 \) if and only if \( \semant{\tr(\muform)}_M(s_0) = 1 \), i.e.~\( M \models \tr(\muform) \).
  The mapping \( \muform \mapsto (\tr(\muform), M) \) is obviously effective, and preserves the order.
\end{proof}
Theorem~\ref{thm:analytical} is an immediate consequence of Theorem~\ref{thm:muArithHardness} and Corollary~\ref{cor:semantRel}.

\section{Decidable Subclass of Order-1 \PHFL{} Model Checking}\label{section:fragment}
As we have seen in Section~\ref{section:undecidability},
\PHFL{} model checking is undecidable, even for order 1.
In this section, we identify a decidable subclass of the
order-1 \PHFL{} model-checking problems
(i.e., a set of pairs \((\form,\M)\) such that whether \(\M{\models} \form\) is decidable).
We identify the subclass by using a type system: we define a type system \(\TA{\M}\)
for \(\PHFL{}\) formulas, parameterized by \(\M\),  so that if \(\form\) is a 
proposition well-typed in \(\TA{\M}\), then \(\M{\models}\form\) is decidable.

This section is structured as follows.
In Section~\ref{sec:type-based-restriction}, we introduce the type system \(\TA{\M}\),
and prove that the semantics of any order-1 well-typed formula is an affine function.
Section~\ref{sec:decidability} introduces a matrix representation of affine functions
and shows the decidability of \(\M{\models}\form\) by appealing to the decidability
of the first-order theory of reals~\cite{Tarski51}. Section~\ref{sec:expressivity} shows
that the restricted fragment is reasonably expressive, by giving an encoding of the
termination problem for recursive Markov chains into the restricted fragment of PHFL model checking.
% The rest of this section is structured as follows.
% Section~\ref{defineFragmentSection} defines the type system \(\TA{\M}\)
% and state the decidability result.
% %Section~\ref{fixpointEquationSection} reviews 
% %decidability results on polynomial equations, which will be used in our decidability proof.
% Sections~\ref{matRepSection} and \ref{modelCheckingSection} give a
% model-checking algorithm for the decidable class.
% To show that the decidable class is reasonably large, we show
% in Section~\ref{fragmentExpressivePowerSection} that the problem of computing
% termination probabilities of Recursive Markov chains can be encoded into a
% \PHFL{} model-checking problem in the decidable class.
%%%we show that  we introduce probabilistic higher-order recursion scheme (PHORS)
%%%and show that termination probability of an order-$1$ PHORS can be represented by a formula in the PHFL fragment.
%%%
% From Sections \ref{defineFragmentSection} to \ref{modelCheckingSection},
% we fix a Markov Chain $M = (S, P, \iAP, \inState)$ over which formulas are interpreted.
% We also assume \(S = \{s_1, s_2, \dots, s_n\}\).

\subsection{Type-based Restriction of Order-1 \PHFL{}}
\label{sec:type-based-restriction}
We first explain the idea of the restriction imposed by our type system.
By definition, the semantics of a (closed) order-\(1\) PHFL formula \(\phi\)
of type \(\Quant \to \Quant\) with respect to the Markov chain \(M\)
is a map \(f_\phi\) from the set of functions \(S \to [0, 1]\) to the same set,
where \(S\) is the set of states of \(M\).
Thus, if \(S = \{s_1, s_2, \dots, s_n\}\) is fixed, \(f_\phi\)
can be regarded as a function from \([0, 1]^n\) to \([0, 1]^n\).
Now, if the function \(f_\phi\) were affine,
i.e., if
there are functions \(f_1, f_2, \dots, f_n\) such that
\(f_\phi(r_1, r_2, \dots, r_n) = (f_1(r_1, r_2, \dots, r_n), \dots, f_n(r_1, r_2, \dots, r_n))\),
where \(f_i(r_1, r_2, \dots, r_n) = c_{i, 0} + c_{i, 1}r_1 + \dots + c_{i, n}r_n\)
for some real numbers \(c_{i, j}\),
then the function \(f_\phi\) would be representable by a finite number of reals \(c_{i, j}\).
The semantics of a fixpoint formula would then be
given as a solution of a fixpoint equation on the coefficients,
which is solvable by appealing to the decidability of first-order theories of reals~\cite{Tarski51}.

Based on the observation above, we introduce a type system to restrict the formulas
so that the semantics of every well-typed order-\(1\) formula is affine.
The conjunction \(\phi_1 \land \phi_2\) is one of the problematic logical connectives
that may make the semantics of an order-\(1\) formula non-affine:
recall that the \(\min\) operator was used to define the semantics of conjunction.
We require that for every subformula of the form \(\phi_1 \land \phi_2\) and for each state \(s \in S\),
one of the values \(\semant{\phi_1}(s)\) and \(\semant{\phi_2}(s)\) is the constant \(0\) or \(1\).
We can then remove the \(\min\) operator,
since we have \(\min(0, x) = 0\) and \(\min(1, x) = x\)
for every \(x \in [0, 1]\).

We parameterize the type system by the Markov chain \(M\),
since it often depends on \(M\)
whether the semantics of an order-\(1\) formula is affine. For example,
the semantics of \((p\land \phi_1)\lor (q\land \phi_2)\) is affine
if the semantics of \(\phi_1\) and \(\phi_2\) are affine \emph{and if}
\(p\) and \(q\) cannot be simultaneously true (i.e., if \(\iAPM{M}(p)\cap \iAPM{M}(q)=\emptyset\)). Without the parameterization, the resulting type system would
be too conservative.

The discussion above motivates us to refine the type \(\Quant\) of propositions
to \(\PropTU{T}{U}\) where \(T, U \subseteq S\) and \(T \cap U = \emptyset\).
Intuitively, the type \(\PropTU{T}{U}\) describes propositions \(\phi \in \Quant\)
such that \(\semant{\phi}(s) = 0\) for all \(s \in T\) and
\(\semant{\phi}(s) = 1\) for all \(s \in U\);
there is no guarantee on the value of \(\semant{\phi}(s)\)
for \(s\in S\setminus(T\cup U)\).
The syntax of \emph{refined types} is given by:
\begin{align*}
%%\tyQual &::= \decTy \bnf \Qual &
\decTy &::= \PropTU{T}{U} \bnf \PropTU{T}{U} \to \decTy
\end{align*}
where \(T\) and \(U\) range over the set of subsets of \(S\)
satisfying \(T \cap U = \emptyset\).
Note that each type \(\decTy\) %\neq \Qual\)
can be expressed as
\(\PropTU{T_1}{U_1} \to \PropTU{T_2}{U_2} \to \dots \to \PropTU{T_k}{U_k} \to \PropTU{T}{U}\)
where \(k \geq 0\).

We define the translation from the set of refined types %%in \(\decTypeSys_M\)
to the set of types in PHFL by
\begin{align*}
  %\tr(\Qual) &= \Qual &
  \tr(\PropTU{T}{U}) &= \Quant &
\tr(\decTy_1 \to \decTy_2) &= \tr(\decTy_1) \to \tr(\decTy_2) %%&&
\end{align*}
and the translation of type environment \(\decTyEnv\) by \((\tr(\decTyEnv))(x) = \tr(\decTyEnv(x))\).
The semantics of refined types is defined as follows.
As explained above, the values of function types are restricted to affine functions.
\begin{definition}\label{dfn:decidableSemant}
For each refined type \(\decTy\), %%\decTy \neq \Qual\) in the type system \(\decTypeSys_M\),
we define the subset \(\semant{\decTy} \subseteq \semant{\tr(\decTy)} \) as follows.
\[
\begin{array}{l}
  \semant{\PropTU{T}{U}} = \set{v \in D_{\Quant} \mid \forall s \in T. v(s) = 0, \forall s \in U. v(s) = 1}\\
  \semant{\PropTU{T_1}{U_1} \to \dots \PropTU{T_k}{U_k} \to \PropTU{T}{U}} =\\\qquad
  \{
  f \in \semant{\Quant^k \to \Quant} \mid\\\qquad\qquad
%%  D_{\tyProp^k\to\tyProp}\mid\\\qquad\qquad
%%  \semant{\Quant^k\to \Quant} \mid\\\qquad\qquad
f\mbox{ is affine on \( \semant{\PropTU{T_1}{U_1}} \times \dots \times \semant{\PropTU{T_k}{U_k}} \), and}
  \\\qquad\qquad
  \forall v_1\in \semant{\PropTU{T_1}{U_1}},\dots,v_k\in \semant{\PropTU{T_k}{U_k}}. f\,v_1\,\cdots\,v_k\in \semant{\PropTU{T}{U}}\}
% f \LEQ_{\PropTU{T}{U}} g \defarrow \forall s\in S.f(s)\le g(s)\\
% \bot_{\PropTU{T}{U}} = \lambda s\in S.\left\{\begin{array}{ll}
% 1 & \mbox{if $s\in U$}\\
% 0 & \mbox{otherwise}\\
% \end{array}\right.
% \qquad
% \top_{\PropTU{T}{U}} = \lambda s\in S.\left\{\begin{array}{ll}
% 0 & \mbox{if $s\in T$}\\
% 1 & \mbox{otherwise}\\
% \end{array}\right.\\
%  f \LEQ_{\PropTU{T_1}{U_1} \to \dots \to \PropTU{T_k}{U_k} \to \PropTU{T}{U}}g \defarrow\\\qquad
%  \forall
%  v_1\in D_{\PropTU{T_1}{U_1}},\dots,v_k\in D_{\PropTU{T_k}{U_k}}.
%  f\,v_1\,\cdots\,v_k\LEQ_{\PropTU{T}{U}}g\,v_1\,\cdots\,v_k\\
%  \bot_{\PropTU{T_1}{U_1} \to \dots \PropTU{T_k}{U_k} \to \PropTU{T}{U}}  =
%  \lambda v_1\in D_{\PropTU{T_1}{U_1}}.\cdots\lambda v_k\in D_{\PropTU{T_k}{U_k}}.\bot_{\PropTU{T}{U}}\\
%  \top_{\PropTU{T_1}{U_1} \to \dots \PropTU{T_k}{U_k} \to \PropTU{T}{U}}  =
%  \lambda v_1\in D_{\PropTU{T_1}{U_1}}.\cdots\lambda v_k\in D_{\PropTU{T_k}{U_k}}.\top_{\PropTU{T}{U}}\\
\end{array}
\]
%%%\begin{enumerate}
%%%\item For \(\decTy = \PropTU{T}{U}\),
%%%\(D_\decTy\) is the set \(\{v \in \semant{\Quant} \mid \forall s \in T. v(s) = 0, \forall s \in U. v(s) = 1\}\)
%%%and \(f_1 \sqleq_\decTy f_2\) if and only if
%%%\(\forall s \in S. f_1(s) \leq f_2(s)\).
%%%%\(f_1 \sqleq_{\tr(\decTy)} f_2\).
%%%%\begin{align*}
%%%%D_\decTy &= \{v \in \semant{\Quant} | \forall s \in T. v(s) = 0, \forall s \in U. v(s) = 1\} \\
%%%%v_1 \sqleq_\decTy v_2 &\defarrow v_1 \sqleq_\Quant v_2
%%%%\end{align*}
%%%\item For \(\decTy = \PropTU{T_1}{U_1} \to \PropTU{T_2}{U_2} \to \dots \PropTU{T_k}{U_k} \to \PropTU{T}{U}\) (\(k \geq 1\)),
%%%\(D_\decTy\) is the set of affine functions
%%%\(f: ([0, 1]^n)^k \to [0, 1]^n\) which belong to \(\semant{\tr(\decTy)}\)
%%%(with the identification between \([0, 1]^S\) and \([0, 1]^n\)),
%%%and \(f_1 \sqleq_\decTy f_2\) if and only if
%%%for every tuple
%%%\((v_1, v_2, \dots, v_k)\) in
%%%\(\semant{\PropTU{T_1}{U_1}} \times \dots \times \semant{\PropTU{T_k}{U_k}}\),
%%%the relation \(f_1 \, v_1 \, v_2 \, \dots \, v_k \sqleq f_2 \, v_1 \, v_2 \, \dots \, v_k\) holds.
%%%%and \(f_1 \sqleq_\decTy f_2\) if and only if \(f_1 \sqleq_{\tr(\decTy)} f_2\).
%%%\end{enumerate}
\end{definition}
\noindent
In the definition above, by ``\(f\) is affine on \(D_1 \times \dots \times D_k \)'', we mean that, for each state \( s_\ell \in S \),
there exist some coefficients \( c^\ell_0,c^\ell_{1,1},\ldots,c^\ell_{k,n} \) such that, for every \( (v_1,\dots,v_k) \in D_1 \times \dots \times D_k \),
\[f\,v_1\,\cdots\, v_k\,s_\ell =
c^{\ell}_{0}+\sum_{i\in\set{1,\ldots,k},j\in\set{1,\ldots,n}}c^{\ell}_{i,j} v_i(s_j).\]
% \[f((x_{1,1},\ldots,x_{1,n}),\ldots,(x_{k,1},\ldots,x_{k,n})) =
% c_{0}+\sum_{i\in\set{1,\ldots,k},j\in\set{1,\ldots,n}}c_{i,j}x_{i,j}\]

\begin{remark}
  \label{rem:closure-property}
  Note that \( \semant{\decTy} \) is not closed under various operations.
  For example,
  the greatest lower bound of affine functions \(\lambda (x,y).x\) and
  \(\lambda (x,y).y \in [0,1]^2\to [0,1]\) is \( \lambda (x,y). \min(x,y) \), which is not affine.
  This means that the conjunction does not preserve affinity, as mentioned above.
  A similar observation applies to fixpoints: for a monotone function \( h \) on \( \semant{\tr(\decTy)} \), even if \( h\,x \in \semant{\decTy} \) for every \( x \in \semant{\decTy} \), it is not necessarily the case that \( \LFP(h) \in \semant{\decTy} \).
  For example, let \(S=\set{s}\), \(\decTy=\PropTU{\emptyset}{\emptyset}\to \PropTU{\emptyset}{\set{s}}\),
  and \(h(f) = \lambda v.\lambda s.\max(f(v)(s), v(s)^2)\). For any \(f\in \semant{\decTy}\) and \(v\in\PropTU{\emptyset}{\emptyset}\),
  \(h(f)(v)(s) = \max(f(v)(s), v(s)^2)=\max(1,v(s)^2)=1\), hence \(h(f)\in\semant{\decTy}\).
  However, \(\LFP(h)=\lambda v.\lambda s.v(s)^2\not\in \semant{\decTy}\).
%% cf. Lemma~{lem:decidableAffineLFP}
\end{remark}

We restrict PHFL formulas by a type system parameterized by a Markov chain \(M\).
%%\begin{align*}
%%\psi &::= [\phi]_J &
%%\phi &::=
%%\toQuant{p} \bnf X \bnf \phi_1 \land \phi_2 \bnf \phi_1 \lor \phi_2
%%\bnf \bigcirc \phi \bnf
%%\mu X. \phi \bnf \lambda X. \phi \bnf \phi_1 \, \phi_2
%%\end{align*}
%%and further restrict them by using
We consider a type judgment of the form:
\(\decTyEnv; \lamTE \pM \form:\decTy\).
Here, \(\decTyEnv\) is a type environment of the form
\(X_1\COL\decTy_1,\ldots,X_k\COL\decTy_k\); it is for fixpoint variables, i.e.,
those bound by \(\mu\) or \(\nu\). The other type environment \(\lamTE\)
is of the form \(Y_1\COL\PropTU{T_1}{U_1},\ldots,Y_m\COL\PropTU{T_\ell}{U_\ell}\);
it is for variables bound by \(\lambda\). We require that the domains of
\(\decTyEnv\) and \(\lamTE\) are disjoint.
The intended meaning of the judgment
\(\decTyEnv; \lamTE \pM \form:\decTy\),
where \(\lamTE=Y_1\COL\PropTU{T_1}{U_1},\ldots,Y_\ell\COL\PropTU{T_\ell}{U_\ell}\)
and \(\decTy=\PropTU{T_{\ell+1}}{U_{\ell+1}}\to\cdots \to \PropTU{T_{\ell+m}}{U_{\ell+m}}\to\PropTU{T}{U}\)
is as follows. Assume: (i) each fixpoint variable \(X\) is bound
to an affine function as described by \(\decTyEnv(X)\),
(ii) each \(Y_i\;(1\leq i\leq \ell+m)\) is bound to a value \((x_{i,1},\ldots,x_{i,n})\)
described by \(\PropTU{T_i}{U_i}\). Then 
the value of \(\form\,Y_{\ell+1}\,\cdots\,Y_{\ell+m}\) is an affine function
on \(x_{i,j}\). Note that the value of \(\form\,Y_{\ell+1}\,\cdots\,Y_{\ell+m}\)
need not be affine on the values of fixpoint variables.
Below \(\lamTE\) is treated as a \emph{sequence} of type bindings, while
\(\decTyEnv\) is treated as a \emph{set}.

%%In the figure, the type environment \(\decTyEnv\) maps each variable to
%%a type in the set ranged over by \(\decTy\). %%\(\decTypeSys_M\).
%%The operator \(\bracketOp\) has been restricted to the top-level,
%%and the operators \(\lozenge, \square\) and \(\nu\) have been removed.
%%Note that \(\psi\) is a qualitative formula
%%and \(\phi\) is a quantitative formula.

\begin{figure}
  %\begin{center}
  %  \begin{multicols}{2}
  \begin{minipage}[t]{0.42\textwidth}
    \typicallabel{T-Weak}
\infrule[T-WeakTU]{\decTyEnv;\lamTE \vdash_M \phi : \PropTU{T}{U}\\
  T' \subseteq T\andalso U' \subseteq U}
        {\decTyEnv;\lamTE \vdash_M \phi : \PropTU{T'}{U'}}
        \vspace*{1ex}
        \infrule[T-Weak]{\decTyEnv;\lamTE \vdash_M \phi : \decTy
          %\andalso\lamTE\subseteq \lamTE'
        }
{\decTyEnv;X\COL\PropTU{T}{U},\lamTE \vdash_M \phi : \decTy}
        \vspace*{1ex}
%%        \infrule[T-Ex]{\decTyEnv;\lamTE,Y_1\COL\PropTU{T_1}{U_1},Y_2\COL\PropTU{T_2}{U_2},
%%          \lamTE' \vdash_M \phi : \decTy
%%        }
%%{\decTyEnv;\lamTE,Y_2\COL\PropTU{T_2}{U_2},Y_1\COL\PropTU{T_1}{U_1},
%%          \lamTE' \vdash_M \phi : \decTy}
%%        \vspace*{1ex}
    \infrule[T-AP]{}{\decTyEnv;\emptyset  \pM p : \PropTU{\cmpl{\iAP(p)}}{\iAP(p)}}
        \vspace*{1ex}
    \infrule[T-FVar]{}{\decTyEnv,X\COL\decTy;\emptyset\pM X:\decTy}
        \vspace*{1ex}
%    \infrule[T-Var]{}{\decTyEnv;\lamTE,X\COL\PropTU{T}{U} \pM X:\PropTU{T}{U}}
    \infrule[T-Var]{}{\decTyEnv;X\COL\PropTU{T}{U} \pM X:\PropTU{T}{U}}
        \vspace*{1ex}
        \infrule[T-Mu]{\decTyEnv, X\COL\decTy;\emptyset \vdash_M \phi : \decTy\\
        \mbox{$\decTy$ is of the form $\cdots\to \PropTU{T}{\emptyset}$}}
        {\decTyEnv;\emptyset \vdash_M \mu X.\phi : \decTy}
        \vspace*{1ex}
        \infrule[T-Nu]{\decTyEnv, X\COL\decTy;\emptyset \vdash_M \phi : \decTy\\
        \mbox{$\decTy$ is of the form $\cdots\to \PropTU{\emptyset}{U}$}
        }
        {\decTyEnv;\emptyset \vdash_M \nu X.\phi : \decTy}
        \vspace*{1ex}
\infrule[T-Abs]{\decTyEnv;\lamTE, X\COL\PropTU{T}{U} \vdash_M \phi : \decTy}
        {\decTyEnv;\lamTE \vdash_M \lambda X.\phi : \PropTU{T}{U}\to\decTy}
        \vspace*{1ex}
\infrule[T-App]{\decTyEnv;\lamTE \vdash_M \phi_1 : \PropTU{T}{U}\to\decTy\\
          \decTyEnv;\lamTE \vdash_M \phi_2 : \PropTU{T}{U}}
        {          \decTyEnv;\lamTE \vdash_M \phi_1\phi_2 : \decTy}
  \end{minipage}
  \begin{minipage}[t]{0.57\textwidth}
    \typicallabel{T-And}
        \infrule[T-Ex]{\decTyEnv;\lamTE_1,X\COL\PropTU{T}{U},Y\COL\PropTU{T'}{U'},\lamTE_2 \vdash_M \phi : \decTy
          %\andalso\lamTE\subseteq \lamTE'
        }
{\decTyEnv;\lamTE_1,Y\COL\PropTU{T'}{U'},X\COL\PropTU{T}{U},\lamTE_2 \vdash_M \phi : \decTy}
        \vspace*{1ex}
\infrule[T-Conj]{\decTyEnv; \lamTE \vdash_M \phi_1 : \PropTU{T_1}{U_1}\andalso
 \decTyEnv; \lamTE \vdash_M \phi_2 : \PropTU{T_2}{U_2}\\
 \lamTE\neq \emptyset\imply T_1 \cup U_1 \cup T_2 \cup U_2 = S_M}
{\decTyEnv;\lamTE \vdash_M \phi_1 \land \phi_2 : \PropTU{(T_1 \cup T_2)}{(U_1 \cap U_2)}}
        \vspace*{1ex}
\infrule[T-Disj]{\decTyEnv; \lamTE \vdash_M \phi_1 : \PropTU{T_1}{U_1}\andalso
 \decTyEnv; \lamTE \vdash_M \phi_2 : \PropTU{T_2}{U_2}\\
 \lamTE\neq \emptyset\imply T_1 \cup U_1 \cup T_2 \cup U_2 = S_M}
{\decTyEnv;\lamTE \vdash_M \phi_1 \lor \phi_2 : \PropTU{(T_1 \cap T_2)}{(U_1 \cup U_2)}}

        \vspace*{1ex}
\infrule[T-J]{\decTyEnv; \emptyset\pM \phi : \PropTU{T}{U}}
        {\decTyEnv;\emptyset \pM [\phi]_J : \PropTU{T}{U}}

        \vspace*{1ex}
        \infrule[T-Min]{\decTyEnv;\emptyset \pM \phi : \PropTU{T}{U}\\
        T' = \set{s\in S_M \mid \exists s'\in T.P_M(s,s')>0}\\
        U' = \set{s\in S_M \mid \forall s'\in S_M.P_M(s,s')>0\imply s'\in U}}
   {\decTyEnv; \emptyset \pM \All \phi : \PropTU{T'}{U'}}

        \vspace*{1ex}
\infrule[T-Max]{\decTyEnv; \emptyset \pM  \phi : \PropTU{T}{U}\\
        T' = \set{s\in S_M \mid \forall s'\in S_M.P_M(s,s')>0\imply s'\in T}\\
        U' = \set{s\in S_M \mid \exists s'\in U.P_M(s,s')>0}}
{\decTyEnv; \emptyset \pM \Some \phi : \PropTU{T'}{U'}}
        \vspace*{1ex}
\infrule[T-Avg]{\decTyEnv;\lamTE \pM  \phi : \PropTU{T}{U}\\
        T' = \set{s\in S_M \mid \forall s'\in S_M.P_M(s,s')>0\imply s'\in T}\\
        U' = \set{s\in S_M \mid \forall s'\in S_M.P_M(s,s')>0\imply s'\in U}}
{\decTyEnv; \lamTE\pM \Avg \phi : \PropTU{T'}{U'}}
\end{minipage}
%\end{multicols}
%\end{center}
\caption{Typing Rules for the Decidable Fragment}
\label{fig:decSubclass}
\end{figure}

The typing rules are given in Figure~\ref{fig:decSubclass}. We explain
key rules below.
The rule \rn{T-WeakTU} is for weakening the information represented by
\(T\) and \(U\); this rule is required, for example, for adjusting the types
between a function and its argument.
The rule \rn{T-Weak} is a usual weakening rule for adding
type bindings to \(\lamTE\).
The rule \rn{T-AP} is for atomic propositions; recall that
\(\iAP(p)\) denotes the set of states where \(p\) holds with probability \(1\).
The rule \rn{T-Mu} is for least fixpoint formulas. The second premise means
that \(\decTy\) is of the form \(\PropTU{T_1}{U_1} \to \dots \to \PropTU{T_k}{U_k} \to \PropTU{T}{U}\), where \(U=\emptyset\).\footnote{The condition \(U=\emptyset\)
  is sometimes too restrictive. For example, consider the formula
  \(\mu X.\true\). We can only assign \(\PropTU{\emptyset}{\emptyset}\),
  although \(\PropTU{\emptyset}{S_M}\) can be assigned to
  the equivalent formula \(\true\). To remedy this problem, it suffices to add
  the rule for unfolding:
  \infrule{\decTyEnv;\lamTE \vdash_M [\mu X.\phi/X]\phi : \decTy}{\decTyEnv;\lamTE \vdash_M \mu X.\phi : \decTy}
For the sake of simplicity, we do not consider this rule.}
Without this restriction, the value of \(\form\,\form_1\,\cdots\,\form_k\) at a
state in \(U\) may be wrongly estimated to be \(1\). For example,
consider the case where \(\phi=X\) and the simple type of \(X\) is \(\Quant\).
Then, the value of \(\mu X.X\)
should be the map \(f\) such that \(f(s)=0\) for every state. Without the restriction,
however, we could wrongly derive \(\mu X.X: \PropTU{\emptyset}{S_M}\).
Note also that \(\lamTE\) is empty in \rn{T-Mu}; this is just for technical
convenience, and is not a fundamental restriction. Indeed,
if \(\mu X.\phi\) contains a free variable \(Y\) of type \(\PropTU{T}{U}\), then
we can replace it with
\((\mu X'.\lambda Y.[X'\,Y/X]\phi)Y\), without changing the semantics.
Analogous conditions are imposed in the rule \rn{T-Nu} for greatest fixpoint formulas.
In the rule \rn{T-Conj} for conjunctions, the first two premises imply that the value
of \(\form_i\) at a state in \(T_i\) is \(0\); therefore, the value
of \(\form_1\land \form_2\) at a state in \(T_1\cup T_2\) is \(0\), which
explains \(T_1\cup T_2\) in the conclusion. Similarly for \(U_1\cap U_2\).
The third premise (on the second line) ensures that the value of
\(\phi_1\land \phi_2\) is an affine function on the value of the variables in
\(\Delta\). That is guaranteed if \(\Delta=\emptyset\). Otherwise, we require
\(T_1\cup U_1\cup T_2\cup U_2=S_M\); recall the earlier discussion on a
sufficient condition for the semantics of an order-1 formula to be affine.
The rule \rn{T-Disj} for disjunctions is analogous.
In the rules \rn{T-J}, \rn{T-Min}, and \rn{T-Max}, we require that
the type environment \(\lamTE\) for \(\lambda\)-bound variables be empty,
since the operators \([\cdot]_J, \All\), and \(\Some\) break the affinity.
The sets \(T'\) and \(U'\) in the conclusions of those rules are conservatively
approximated. In \rn{T-J}, recall that we have excluded out trivial
bounds such as \(>1\) and \(\ge 0\); thus, the value of \([\form]_J\) is \(0\)
(\(1\), resp.)
if the value of \(\form\) is \(0\) (\(1\), resp.).
In the rule \rn{T-Avg}, we need not require \(\lamTE\) to be empty,
as the average of affine functions is again affine.

%%%A key rule in Figure~\ref{fig:decSubclass} is for conjunctions. Note that 
%%%\(\semant{\phi_1 \land \phi_2}(s)=0\) if
%%%either \(\semant{\phi_1}(s) = 0\) or \(\semant{\phi_2}(s) = 0\) holds;
%%%hence \(s\in T_1\cup T_2\) implies \(\semant{\phi_1 \land \phi_2}(s)=0\).
%%%Note also that \(\semant{\phi_1 \land \phi_2}(s)=1\) if both \(\semant{\phi_1}(s) = 1\)
%%%\emph{and} \(\semant{\phi_2}(s) = 1\) hold. Thus, 
%%%\(s\in U_1\cap U_2\) implies \(\semant{\phi_1 \land \phi_2}(s)=1\).
%%%This is why \(\phi_1 \land \phi_2\) has type \(\PropTU{T_1 \cup T_2}{U_1 \cap U_2}\).
%%%The extra condition \(T_1\cup U_1\cup T_2\cup U_2=S\) requires that,
%%%for each state \(s\), either 
%%%\(\semant{\phi_1}(s)\) or  \(\semant{\phi_2}(s)\) is 
%%%the constant \(0\) or \(1\); recall the earlier discussion on a
%%%sufficient condition for the semantics of an order-1 formula to be affine.
%%%%%%for ,
%%%%%%\(1\), one of \(\phi_i(s)\)'s. This condition ensures that the value of
%%%%%%a fixpoint-free formula \(\phi\) at state \(s\)
%%%%%%can be expressed by an affine expression on the values of
%%%%%%the subformulas of \(\phi\), which will be exploited by the matrix representation of
%%%%%%the semantics of a formula given in Appendix~\ref{section:decidable}.
%%%The rule for disjunctions is analogous.
%%%In the rule for applications, the formula in the conclusion is in a fully applied form;
%%%this does not lose generality, as we consider the order-1 fragment of PHFL here.
%%%The last rule is for conservatively estimating the sets \(T\) and \(U\).

\renewcommand\toQuant[1]{#1}

\begin{example}
  \label{ex:affine-formula}
  Let \(M\) be an element of $\mathcal{M}$ in the proof of Theorem~\ref{expressiveThm},
  i.e., a Markov chain \((S, P, \iAP, \inState)\) that satisfies the following conditions.
  \begin{itemize}
\item $S = \{s_0, s_1, \dots, s_n\}$ for a positive integer $n$,
\item $P(s_i, s_{i + 1}) = 1$ ($0 \leq i \leq n - 1$), $P(s_n, s_n) = 1$ and $P(s_i, s_j) = 0$ otherwise.
\item There are three atomic propositions $a, b, c$ with $\iAP(a) \cup \iAP(b) = \{s_0, s_1, \dots, s_{n - 1}\}$,
$\iAP(a) \cap \iAP(b) = \emptyset$ and $\iAP(c) = \{s_n\}$.
\item The initial state is $\inState = s_0$.
\end{itemize}
  Let \(\phi_1\) be the formula:
  \[
%%(\mu F. (\lambda X. [(\toQuant{a} \land \lozenge \toQuant{X}) \lor
%%\lozenge F (\lozenge (\toQuant{b} \land \toQuant{X})))]_{\geq 1}) (b \land [\bigcirc \toQuant{c}]_{\geq 1}).
(\mu F. \lambda X. a \land \Avg (X \lor
  F (b \land \Avg X)))(b\land \lozenge c),
  \]
  which is a variation of the formula \(\phi\) considered in the proof, obtained by replacing two occurrences of
   \(\Some\) with
   \(\Avg\). Since \(M\) has only deterministic transitions, \(\phi_1\) has
   the same value as \(\phi\).
   Let \(\decTyEnv = F\COL \PropTU{\iAP(a)}{\emptyset}\to \PropTU{\iAP(b)\cup\iAP(c)}{\emptyset}\)
   and \(\lamTE=X\COL \PropTU{\iAP(a)}{\emptyset}\).
   Then,
   we have:
   \[
   \footnotesize
   \infers{\decTyEnv;\lamTE\pM
     a \land \Avg (X \lor
     F (b \land \Avg X)): \PropTU{\iAP(b)\cup\iAP(c)}{\emptyset}}
   {\decTyEnv;\lamTE\pM  a:\PropTU{\iAP(b)\cup\iAP(c)}{\iAP(a)}
     & \hspace*{-5em}
     \infers{
     \decTyEnv;\lamTE\pM
     \Avg (X \lor
     F (b \land \Avg X)): \PropTU{\emptyset}{\emptyset}}
            {
    \infers{\decTyEnv;\lamTE\pM
     X \lor
     F (b \land \Avg X): \PropTU{\emptyset}{\emptyset}}
           {\infers{\decTyEnv;\lamTE\pM X:\PropTU{\iAP(a)}{\emptyset}
               }{}
           &
           \infers{\decTyEnv;\lamTE\pM F (b \land \Avg X): \PropTU{\iAP(b)\cup\iAP(c)}{\emptyset}}{\cdots}
              }}
   }
   \]
   Here, \(\decTyEnv;\lamTE\pM F (b \land \Avg X): \PropTU{\iAP(b)\cup\iAP(c)}{\emptyset}\) is derived as follows.
   \[\footnotesize
   \infers{\decTyEnv;\lamTE\pM F (b \land \Avg X): \PropTU{\iAP(b)\cup\iAP(c)}{\emptyset}}
          {\decTyEnv;\lamTE\pM F:
            \PropTU{\iAP(a)}{\emptyset}\to \PropTU{\iAP(b)\cup\iAP(c)}{\emptyset}
            &\hspace*{-4em}
            \infers{ \decTyEnv;\lamTE\pM b \land \Avg X:\PropTU{\iAP(a)}{\emptyset}}{
            \infers{ \decTyEnv;\lamTE\pM b \land \Avg X:\PropTU{\iAP(a)\cup\iAP(c)}{\emptyset}}
                   {\decTyEnv;\lamTE\pM b:\PropTU{\iAP(a)\cup\iAP(c)}{\iAP(b)}
                     & \infers{\decTyEnv;\lamTE\pM \Avg X:\PropTU{\emptyset}{\emptyset}}
                     {\decTyEnv;\lamTE\pM X:\PropTU{\iAP(a)}{\emptyset}}
            }}
          }
          \]
          We can thus obtain
          \[\emptyset;\emptyset\pM
          (\mu F. \lambda X. a \land \Avg (X \lor
          F (b \land \Avg X)))(b\land \lozenge c): \PropTU{\iAP(b)\cup\iAP(c)}{\emptyset}.\]
          Note that, by the same argument as the proof of Theorem~\ref{expressiveThm},
          there exists no \(\mup\)-calculus formula equivalent to \(\phi_1\).
          \qed
\end{example}

The following lemma states that a formula that is well-typed in \(\decTypeSys_M\) is also well-typed
in the original PHFL type system.
\begin{lemma}\label{lem:decidableWellTyped}
  Let \(\phi\) be a PHFL formula such that \(\decTyEnv;\lamTE \vdash_M \phi: \decTy\).
Then we have \(
\tr(\decTyEnv,\lamTE) \vdash \phi : \tr(\tyQual)
\).
\end{lemma}
\begin{proof}
  This follows by a straightforward induction on the derivation of
  \(\decTyEnv;\lamTE \vdash_M \phi: \decTy\).
\end{proof}

The following lemma states that the refined type system does not impose any
restriction on the order-0 fragment of PHFL.
Thus, together with the observation in Example~\ref{ex:affine-formula},
the lemma implies that our
decidable fragment is strictly more expressive than the \(\mup\)-calculus.
\begin{lemma}
  \label{lem:mup-subsumed}
  If \(\tyEnv \vdash \phi: \ty\), and \(\phi\) is an order-0 formula,
  then \(\decTyEnv;\emptyset\pM \phi: \PropTU{\emptyset}{\emptyset}\)
  where \(\decTyEnv\) is the type environment such that \(\dom(\decTyEnv)=\dom(\tyEnv)\)
  and
  \(\decTyEnv(X)=\PropTU{\emptyset}{\emptyset}\) for every \(X\in \dom(\decTyEnv)\).
\end{lemma}
\begin{proof}
  This follows by a straightforward induction on the derivation of
  \(\tyEnv \vdash \phi: \ty\). Note that since \(\lamTE\) is always empty,
  the condition \(T_1\cup U_1\cup T_2\cup U_2\) in \rn{T-Conj} and \rn{T-Disj}
  is irrelevant.
\end{proof}

%%We study semantic properties of the type system.
%%The first property is soundness, i.e.~the interpretation of a typable formula is affine in a certain sense.

In the rest of this subsection, we prove the following properties.
\begin{enumerate}
\item The type system is sound in the sense that
the semantics of any formula \(\form\) of type \(\decTy\) indeed belongs to \(\semant{\decTy}\);
see Theorem~\ref{thm:decidableTypeSoundness} for the precise statement.
\item The calculation of the semantics of a well-typed formula (especially, the least/greatest fixpoint
  computation) can be
  performed up to the equivalence relation \(\sim_{\decTy}\), where
  \(f\sim_{\PropTU{T_1}{U_1} \to \dots \to \PropTU{T_m}{U_m} \to \PropTU{T}{U}}g \) just if
  \(f\) and \(g\) are equivalent on the intended domain, i.e., if
  \(f\,v_1\,\cdots\,v_m = g\,v_1\,\cdots\,v_m\) for any \(v_1\in\semant{\PropTU{T_1}{U_1}},
  \ldots,v_m\in\semant{\PropTU{T_m}{U_m}}\); see Lemmas~\ref{lem:decidableModuloSim} and
  \ref{cor:decidableFixpointRepresentation}.
\end{enumerate}

The reason why the type system ensures affinity has been intuitively explained already,
except for the fixpoints.
Here we show (in Lemma~\ref{lem:decidableAffineLFP})
that the fixpoint of a typable fixpoint operator is indeed affine.
The key observation is that \( \semant{\decTy} \subseteq \semant{\tr(\decTy)}\) is closed under the limit of chains, as stated in the following lemma.
\begin{lemma}\label{lem:decidableChainComplete}
  Let \( \decTy \) be a refined type and \( \gamma \) be an ordinal number.
  For every increasing chain \( (f_\alpha)_{\alpha < \gamma} \) of elements in \( \semant{\decTy} \), the limit \( \bigsqcup_{\alpha < \gamma} f_\alpha \) in \( \semant{\tr(\decTy)} \) belongs to \( \semant{\decTy} \).
  Similarly, for every decreasing chain \( (f_\alpha)_{\alpha < \gamma} \), the limit \( \bigsqcap_{\alpha < \gamma} x_\alpha \) is in \( \semant{\decTy} \).
\end{lemma}
\begin{proof}
  We prove the former.
  Assume \( \kappa = \PropTU{T_1}{U_1} \to \dots \to \PropTU{T_k}{U_k} \to \PropTU{T}{U} \).

  We first give an alternative characterization of affinity.
  For each \( i \le k \), given \( v_i \in \semant{\PropTU{T_i}{U_i}} \) and \( r \in [0,1] \), we define \( r \cdot v_i \) by \( (r \cdot v_i)(s) = r (v_i(s)) \).
  Note that \( r \cdot v_i \) may not be a member of \( \semant{\PropTU{T_i}{U_i}} \), but \( r \cdot v_i + (1-r) \cdot v_i' \in \semant{\PropTU{T_i}{U_i}} \) for every \( v_i,v_i' \in \semant{\PropTU{T_i}{U_i}} \) and \( r \in [0,1] \) (here the sum is the point-wise sum on reals).
  % \begin{equation*}
  %   (r \cdot v_i)(s) =
  %   \begin{cases}
  %     1 & \mbox{if \( s \in U \)} \\
  %     r (v_i(s)) & \mbox{if \( s \notin U \).}
  %   \end{cases}
  % \end{equation*}
  % Then \( r \cdot v_i \in \semant{\PropTU{T_i}{U_i}} \).
  % We define \( \bot_i \in \semant{\PropTU{T_i}{U_i}} \) by \( \bot_i(s) = 1 \) if \( s \in U_i \) and \( \bot_i(s) = 0 \) otherwise.
  Then \( f \in \semant{\Prop^k \to \Prop} \) is affine on \( \semant{\PropTU{T_1}{U_1}} \times \dots \times \semant{\PropTU{T_k}{U_k}} \) if and only if, for every \( (v_1, \dots, v_k), (v'_1,\dots,v'_k) \in \semant{\PropTU{T_1}{U_1}} \times \dots \times \semant{\PropTU{T_k}{U_k}} \) and \( r \in [0,1] \),
  \begin{equation*}
    f (r \cdot v_1 + (1-r) \cdot v_1') \dots (r \cdot v_k + (1-r) \cdot v_k')
    =
    r (f\,v_1\,\dots\,v_k) + (1-r) (f\,v_1'\,\dots\,v_k').
    % =
    % f\,(r \cdot v_1) \dots (r \cdot v_k) - f\,\bot_1\,\dots\,\bot_k
  \end{equation*}

  Let \( (f_\alpha)_{\alpha < \gamma} \) be an increasing chain and \( f = \bigsqcup_{\alpha<\gamma} f_\alpha \).
  Then \( f \) can be characterized in terms of the limits in real numbers as \( f\,y_1\,\dots\,y_k = \lim_{\alpha < \gamma} (f_\alpha\,y_1\,\dots\,y_k) \) for every \( y_1,\dots,y_k \in \semant{\Prop} \).
  Since \( \lim_{\alpha < \gamma} \) commutes with linear operations,
  for every \( (v_1, \dots, v_k)\),
   \((v'_1,\dots,v'_k) \in \semant{\PropTU{T_1}{U_1}} \times \dots \times \semant{\PropTU{T_k}{U_k}} \) and \( r \in [0,1] \), we have:
  \begin{align*}
    & f\,(r \cdot v_1 + (1-r) \cdot v'_1) \dots (r \cdot v_k + (1-r) \cdot v'_k)
    \\
    &=
    \lim_{\alpha<\gamma}(f_\alpha\,(r \cdot v_1 + (1-r) \cdot v'_1) \dots (r \cdot v_k + (1-r) \cdot v'_k)
    \\
    &=
    \lim_{\alpha<\gamma} \big(r(f_\alpha\,v_1\dots v_k) + (1-r) (f_\alpha\,v'_1\,\dots\,v'_k)\big)
    \\
    &=
    r \big(\lim_{\alpha<\gamma} (f_\alpha\,v_1\,\dots v_k)\big) +
    (1-r) \big( \lim_{\alpha<\gamma} (f_\alpha\,v'_1\,\dots\,v'_k) \big)
    \\
    &=
    r (f\,v_1\,\dots\,v_k) + (1-r) (f\,v_1'\,\dots\,v_k').
  \end{align*}

%  The latter can be proved similarly.
  The latter is the dual of the former, and can be proved in the same manner,
  by just replacing \(\bigsqcup\) with \(\bigsqcap\).
\end{proof}

\begin{lemma}\label{lem:decidableAffineLFP}
  Let \( \decTy = \PropTU{T_1}{U_1} \to \dots \to \PropTU{T_k}{U_k} \to \PropTU{T}{U} \) be a refined type and \( h \) be a monotone function on \( \tr(\decTy) \) such that \( x \in \semant{\decTy} \) implies \( h\,x \in \semant{\decTy} \).
  \begin{itemize}
    \item If \( U = \emptyset \), then \( \LFP(h) \in \semant{\decTy} \).
    \item If \( T = \emptyset \), then \( \GFP(h) \in \semant{\decTy} \).
  \end{itemize}
  Here \( \LFP \) and \( \GFP \) are taken in \( \semant{\tr(\decTy)} \).
\end{lemma}
\begin{proof}
  We prove the former; the latter can be proved similarly.
  We define \( h^{\alpha}(\bot) \) for ordinals \( \alpha \) as follows:
  \begin{align*}
    h^{\alpha}(\bot) =
    \begin{cases}
      \bot_{\tr(\decTy)} & \mbox{if \( \alpha = 0 \)} \\
      h(h^{\beta}(\bot)) & \mbox{if \( \alpha = \beta + 1 \)} \\
      \bigsqcup_{\beta < \alpha} h^{\beta}(\bot) & \mbox{if \( \alpha \) is a limit ordinal}.
    \end{cases}
  \end{align*}
  Then \( \LFP(h) = h^{\alpha}(\bot) \) for some sufficiently large \( \alpha \).

  It suffices to show that \( h^{\alpha}(\bot) \in \semant{\decTy} \) for every \( \alpha  \).
  We prove this claim by transfinite induction.
  For \( \alpha = 0 \), \( \bot_{\tr(\decTy)} \in \semant{\decTy} \) since \( U = \emptyset \).
  For \( \alpha = \beta + 1 \), we have \( h(h^{\beta}(\bot)) \in \semant{\decTy} \) since \( h^{\beta}(\bot) \in \semant{\decTy} \).
  If \( \alpha \) is a limit cardinal, we appeal to Lemma~\ref{lem:decidableChainComplete}; note that \( (f^{\beta}(\bot))_{\beta < \alpha} \) is an increasing chain.
\end{proof}
Note that in the above lemma,
\(\LFP(h)\in\semant{\decTy}\) would not hold if we drop
the condition \(U=\emptyset\); recall Remark~\ref{rem:closure-property}.
That justifies the corresponding conditions in rule \rn{T-Mu} and \rn{T-Nu}.

% \begin{remark}
%   One may wonder whether the condition on \( h \) can be weakened to \( h\,x\in \semant{\decTy} \) for every \( x \in \semant{\decTy} \).
%   Unfortunately the weaker version has a counterexample.
%   For example, let \( \decTy \) be \( \PropTU{\emptyset}{S} \) and \( h \) the identity function on \( \semant{\Prop} \).
%   Obviously \( h\,x \in \semant{\decTy} \) if \( x \in \semant{\decTy} \).
%   However \( \LFP(h) = \bot_\Prop \notin \semant{\decTy} = \{ \top_{\Prop} \} \). 
%   %
%   \qed
% \end{remark}

The following theorem is soundness of the type system.
Let \( \semant{\decTyEnv} \) be the subset of \( \semant{\tr(\decTyEnv)} \) consisting of interpretations \( \interpret \) such that \( \interpret(X) \in \semant{\decTyEnv(X)} \) for every \( X \in \dom(\decTyEnv) \).
For \( \lamTE = (Y_1 : \PropTU{T_1}{U_1}, \dots, Y_m : \PropTU{T_m}{U_m}) \), we write \( \lambda \lamTE.\varphi \) for \( \lambda Y_1. \dots \lambda Y_m. \varphi \) and \( \lamTE \to \decTy \) for \( \PropTU{T_1}{U_1} \to \dots \to \PropTU{T_m}{U_m} \to \decTy \).
\begin{theorem}\label{thm:decidableTypeSoundness}
  Let \(\phi\) be a PHFL formula such that \(\decTyEnv;\lamTE \pM \phi: \decTy\).
  Then, for every \( \interpret \in \semant{\decTyEnv} \), we have \(\semant{\tr(\decTyEnv) \vdash \lambda \lamTE.\phi: \tr(\lamTE \to \decTy)}(\interpret) \in \semant{\lamTE \to \decTy}\).
\end{theorem}
\begin{proof}
  By induction on the structure of derivation \( \decTyEnv;\lamTE \pM \phi : \decTy \), with case
  analysis on the last rule used.
  We use Lemma~\ref{lem:decidableAffineLFP} for the cases of \rn{T-Mu} and \rn{T-Nu}. % of fixpoints.
  For the cases of \rn{T-Conj} and \rn{T-Disj}, %% where \(\Delta\neq \emptyset\),
  the condition \(T_1 \cup U_1 \cup T_2 \cup U_2 = S_M\) plays the key role.
  Note that if \(s\) belongs to \(T_i\cup U_i\), then the semantics of \(\lambda \lamTE.\form_0\land\form_1\)
  at \(s\) coincides with either that of \(\lambda \lamTE.\form_{1-i}\) or
  a constant function \(\lambda \tilde{v}.0\).
%%%  , we have \(\form=\form_1\land\form_2\) and
%%%  \(\decTy=\PropTU{T_1\cup T_1}{U_1\cap U_2}\), with
%%%   \(\semant{\tr(\decTyEnv) \vdash \lambda \lamTE.\phi_1:
%%%    \tr(\lamTE \to \decTy)}(\interpret) \in \semant{\lamTE \to \PropTU{T_1}{U_1}}\),
%%%   \(\semant{\tr(\decTyEnv) \vdash \lambda \lamTE.\phi_2:
%%%     \tr(\lamTE \to \decTy)}(\interpret) \in \semant{\lamTE \to \PropTU{T_2}{U_2}}\),
%%%   and \(\lamTE\neq \emptyset\imply T_1 \cup U_1 \cup T_2 \cup U_2 = S_M\).
%%%   The result follows immediately \(\lamTE=\emptyset\). Suppose
%%%   \(\lamTE\neq\emptyset\).
%%%   Then for every \(s\in S_M\), \(s\) belongs to at least one of \(T_1,U_1,T_2\), and \(U_2\).
%%%   If \(s\in T_i\), then 
  Other cases are easy.
\end{proof}

Let us discuss another important property of the type system.
%%In the type system, arguments of a function of type \( \PropTU{T}{U} \to \decTy \) must be of type \( \PropTU{T}{U} \).
%%Hence two functions \( f, g \) of type \( \PropTU{T}{U} \to \decTy \) is indistinguishable if they coincide on \( \semant{\PropTU{T}{U}} \).
For \( \decTy = \PropTU{T_1}{U_1} \to \dots \to \PropTU{T_k}{U_k} \to \PropTU{T}{U} \) and \( f,g \in \semant{\tr(\decTy)} \),
we write \( f \precsim_{\decTy} g \) if \( f\,v_1\,\dots\,v_k \LEQ_{\Quant} g\,v_1\,\dots\,v_k \) for every \((v_1,\dots,v_k) \in \semant{\PropTU{T_1}{U_1}} \times \dots \times \semant{\PropTU{T_k}{U_k}} \).
For interpretations \( \interpret_0, \interpret_1 \in \semant{\decTyEnv} \), we write \( \interpret_0 \precsim_{\decTyEnv} \interpret_1 \) if \( \interpret_0(X) \precsim_{\decTyEnv(X)} \interpret_1(X) \) for every \( X \in \dom(\decTyEnv) \).
We write \(f\sim_{\decTy}g\) for \(f\preceq_{\decTy}g\land g\preceq_{\decTy}f\), and analogously for interpretations. The results in the rest of this subsection
allows us to compute the semantics of a formula of type
\(\decTy\) up to \(\sim_{\decTy}\).
\begin{lemma}\label{lem:decidableSimFixpoint}
  Let \( \decTy = \PropTU{T_1}{U_1} \to \dots \to \PropTU{T_k}{U_k} \to \PropTU{T}{U} \) be a refined type and \( h_0 \) and \( h_1 \) be monotone functions on \( \tr(\decTy) \) such that \( x \in \semant{\decTy} \) implies \( h_0\,x, h_1\,x \in \semant{\decTy} \) and that \( x_0 \precsim_{\decTy} x_1 \) implies \( h_0\,x_0 \precsim_{\decTy} h_1\,x_1 \).
  \begin{itemize}
    \item If \( U = \emptyset \), then \( \LFP(h_0) \precsim_{\decTy} \LFP(h_1) \).
    \item If \( T = \emptyset \), then \( \GFP(h_0) \precsim_{\decTy} \GFP(h_1) \).
  \end{itemize}
  Here \( \LFP \) and \( \GFP \) are taken in \( \semant{\tr(\decTy)} \).
\end{lemma}
\begin{proof}
  We prove the former; the proof of the latter is analogous.
  As discussed in the proof of Lemma~\ref{lem:decidableAffineLFP},
  \( \LFP(h_0) = h_0^{\alpha}(\bot) \) and \( \LFP(h_1) =  h_1^{\alpha}(\bot) \) for some sufficiently large ordinal \( \alpha \).
  We prove \( h_0^{\alpha}(\bot) \precsim_{\decTy} h_1^{\alpha}(\bot) \) by induction on \( \alpha \).
  Trivially \( h_0^{\alpha}(\bot) = \bot \precsim_{\decTy} \bot = h_1^{\alpha}(\bot) \).
  For \( \alpha = \beta+1 \), since \( h_0^{\beta}(\bot) \precsim_{\decTy} h_1^{\beta}(\bot) \) by the induction hypothesis, \( h_0(h_0^{\beta}(\bot)) \precsim_{\decTy} h_1(h_1^{\beta}(\bot)) \) follows from the assumption.
  If \( \alpha \) is a limit cardinal, then \( h_i^{\alpha} = \bigsqcup_{\beta<\alpha} h_i^{\beta}(\bot) \).
  For every \( (v_1,\dots,v_k) \in \semant{\PropTU{T_1}{U_1}} \times \dots \times \semant{\PropTU{T_k}{U_k}} \), we have \( h_i^{\alpha}(\bot)(v_1)\dots(v_k) = \bigsqcup_{\beta<\alpha} (h_i^{\beta}(\bot)(v_1)\dots(v_k)) \) by definition.
  Since \( h_0^{\beta}(\bot)(v_1)\dots(v_k) \precsim_{\decTy} h_1^{\beta}(\bot)(v_1)\dots(v_k) \) holds for every \( \beta<\alpha \) by the induction hypothesis, we have \( h_0^{\alpha}(\bot)(v_1)\dots(v_k) \precsim_{\decTy} h_1^{\alpha}(\bot)(v_1)\dots(v_k) \).
\end{proof}

\begin{lemma}\label{lem:decidableModuloSim}
  Let \(\phi\) be a PHFL formula such that \(\decTyEnv;\lamTE \pM \phi: \decTy\).
  Then, for every \(\interpret,\interpret' \in \semant{\decTyEnv} \) such that \( \interpret \precsim_{\decTyEnv} \interpret' \), we have \( \semant{\lambda \lamTE.\phi}(\interpret) \precsim_{\lamTE \to \decTy} \semant{\lambda \lamTE.\phi}(\interpret') \).
\end{lemma}
\begin{proof}
  By induction on the structure of derivation \( \decTyEnv;\lamTE \pM \phi: \decTy \).
  The most cases are easy.
  For fixpoints, we appeal to Lemma~\ref{lem:decidableSimFixpoint}.
  The details are given in Appendix~\ref{sec:proofs}.
\end{proof}

The above lemma shows that the semantic domain for \( \decTy \) can be regarded as \( \semant{\decTy} \) modulo \( \sim_{\decTy} \).
Furthermore the least and greatest fixpoints can be characterized in terms of the preorder \(\precsim_\decTy\).
\begin{lemma}\label{cor:decidableFixpointRepresentation}
%%  Let \( \decTy \) be a refined type and \( h \) be a function in the previous lemma.
    Let \( \decTy \) be a refined type and \( h \) be a monotone function on \( \tr(\decTy) \) such that \( x \in \semant{\decTy} \) implies \( h\,x \in \semant{\decTy} \) and that \( x \precsim_{\decTy} x' \) implies \( h\,x \precsim_{\decTy} h\,x' \).
  If \( \LFP(h) \in \semant{\decTy} \), then \( \LFP(h) \) is a least element in \( \{ x \in \semant{\decTy} \mid h\,x \precsim_{\decTy} x \} \) with respect to \( \precsim_{\decTy} \).
  Similarly, if \( \GFP(h) \in \semant{\decTy} \), then \( \GFP(h) \) is a greatest postfixpoint
 of \( h \) w.r.t~\( \precsim_{\decTy} \). 
\end{lemma}
%%\end{corollary}
\begin{proof}
  Assume \( \decTy = \PropTU{T_1}{U_1} \to \dots \to \PropTU{T_k}{U_k} \to \PropTU{T}{U} \).
  Let \(\bot_\decTy\) be the least element of \(\semant{\decTy}\), which is defined by:
  for every \( s \in S_M \),
  %%  \begin{equation*}
  \[
  \begin{array}{l}
    \bot_{\PropTU{T}{U}}(s) = 
    \begin{cases}
      1 & \mbox{if \(s\in U \)} \\
      0 & \mbox{otherwise.}
    \end{cases}\\
    \bot_\decTy\,v_1\,\dots\,v_k =
    \begin{cases}
      \bot_{\PropTU{T}{U}} & \mbox{if \(v_i\GEQ\bot_{\PropTU{T_i}{U_i}}\) for every \(i\)}\\
%%    1&\mbpx{if \( (v_1,\dots,v_k) \in
%%\semant{\PropTU{T_1}{U_1}} \times \dots \times \semant{\PropTU{T_k}{U_k}}, s\in U \)} \\
      \lambda s.0 & \mbox{otherwise.}
    \end{cases}
  \end{array}
  \]
%%  \end{equation*}
  For an ordinal number \( \alpha \), we define \( x_{\alpha}\) by \( x_0 = \bot_\decTy\),
  \( x_{\beta+1} = h(x_{\beta}) \) and \( x_{\alpha} = \bigsqcup_{\beta<\alpha} x_\beta \) if \( \alpha \) is a limit ordinal. By the monotonicity of \(h\) and Lemma~\ref{lem:decidableChainComplete},
  the sequence \(\set{x_\alpha}_\alpha\) is well-defined and forms an increasing chain.
  Thus, there exists an ordinal \(\alpha\) such that \(x_{\alpha}\) is a fixpoint of \(h\),
  and in particular, an element of
  \( \{ x \in \semant{\decTy} \mid h\,x \precsim_{\decTy} x \} \).
  It follows by straightforward induction on \(\alpha\) that \(x_{\alpha}\) is also a least element
  of \( \{ x \in \semant{\decTy} \mid h\,x \precsim_{\decTy} x \} \).
  Now, by the assumption that \(\LFP(h)\in\semant{\decTy}\),
  \(\LFP(h)\) is also an element of
  \( \{ x \in \semant{\decTy} \mid h\,x \precsim_{\decTy} x \} \),
  and since \(\LFP(h)\LEQ{} x_\alpha\), 
\(\LFP(h)\) is also a least element of   \( \{ x \in \semant{\decTy} \mid h\,x \precsim_{\decTy} x \} \).
%%  Obviously, \(\LFP(h)\) is an element of
%%  \( \{ x \in \semant{\decTy} \mid h\,x \precsim_{\decTy} x \} \).
%%  To show it is a least one, suppose \(h\,x \precsim_{\decTy} x\). 
%%  By Lemma~\ref{lem:fixpoint}, there exists a fixpoint \(y\) of \(h\)
%%  such that  \(y\sim_{\decTy}x\). Since \(\LFP(h)\) is the least fixpoint,
%%  we have \(\LFP(h)\LEQ{} y \sim_{\decTy} x\), which implies
%%  we have \(\LFP(h)\precsim_{\decTy}\) as required.
The proof for \(\GFP(h)\) is analogous.
\end{proof}
\noindent
If \( \decTyEnv;\emptyset \pM \mu X.\phi: \decTy \) and \( \interpret \in \semant{\decTyEnv} \), then \( h = \semant{\lambda X.\phi}(\interpret) \) satisfies the condition of the lemma above
(recall Lemma~\ref{lem:decidableModuloSim}).
Hence \( \semant{\mu X.\phi}(\interpret) \) is \( \sim_{\decTy} \)-equivalent to every least
%pre
fixpoint of \( \semant{\lambda X.\phi}(\interpret) \) with respect to \( \precsim_{\decTy} \);
in other words, to compute \(\semant{\mu X.\phi}(\interpret) \) up to 
\(\sim_{\decTy}\), it suffices to compute a least fixpoint of
\( \semant{\lambda X.\phi}(\interpret) \) in the quotient set \(\semant{\decTy}/\sim_{\decTy}\).
A similar statement holds for \( \nu X.\phi \).

\subsection{Decidability}
\label{sec:decidability}
This subsection gives a description of the interpretation of a PHFL formula using the first-order theory of reals,
and obtains the decidability of the restricted fragment of PHFL model checking.
To this end, we need to represent an element in \( \semant{\decTy} \) as a tuple of reals.
Each element of \( \semant{\PropTU{T}{U}} \) can be naturally written as an \( n \)-tuple of \( [0,1] \) where \( n \) is the number of states in \( M \).
What remains is a representation of functions \( \semant{\PropTU{T_1}{U_1} \to \dots \to \PropTU{T_m}{U_m} \to \PropTU{T}{U}} \).

The key result is Lemma~\ref{lem:decidableModuloSim}, which allows us to identify \( x,y \in \semant{\decTy} \) if \( x \sim_{\decTy} y \).
Let \( \decTy = \PropTU{T_1}{U_1} \to \dots \to \PropTU{T_m}{U_m} \to \PropTU{T}{U} \).)
Since \( x \in \semant{\decTy} \) is affine on \( \semant{\PropTU{T_1}{U_1}} \times \dots \times \semant{\PropTU{T_m}{U_m}} \), the \( \sim_{\decTy} \)-equivalence class of \( x \) contains an affine function on \( [0,1]^m \), i.e.,
\( x \sim_{\decTy} f \) holds for some
\[
f\,v_1\,\dots\,v_k\,s = c_s + \sum_{i =1}^{m} \sum_{s' \in S} c_{s, i,s'} v_i(s').
\]
We can use the tuple of coefficients \( (c_s, (c_{s,i,s'})_{i \le k, s' \in S})_{s \in S} \) as a representation of \( x \).

Henceforth, we assume the set \(S_M\) of states of \(M\) is \(\set{1,\ldots,n}\).
We define the \emph{affine semantics} \( \msem{\Prop^m \to \Prop} \) of type \( \Prop^m \to \Prop \) by
\begin{align*}
\msem{\Prop^m \to \Prop}
&=
\{(c_{i,j,k})_{i\in\set{1,\ldots,n},j\in\set{0,\ldots,m},k\in\set{0,\ldots,n}} \in [0,1]^{n(m+1)(n+1)}
\mid\\&\qquad\qquad
\sum_{j=0}^{m} \sum_{k=0}^{n} c_{i,j,k} \le 1 \mbox{ for every \( i = 1,\dots,n \).}
\}
\end{align*}
The tuple \((c_{i,j,k})_{i,j,k}\in\msem{\decTy} \)
%\msem{\PropTU{T_1}{U_1}\to \cdots \PropTU{T_m}{U_m}\to \PropTU{T}{U}}\)
represents
the function \(f\in\semant{\Quant^m\to\Quant}\) such that, for each $i\in S_M=\set{1,\ldots,n}$,
\[f\,g_1\,\cdots\,g_m\,i =
c_{i,0,0}+\Sigma_{j\in\set{1,\ldots,m},k\in\set{1,\ldots,n}}c_{i,j,k}\cdot g_j(k).\]
We write \(\tofun{(c_{i,j,k})_{i,j,k}}\) for the function \(f\) above.
The \emph{affine semantics} \( \msem{\decTy} \) of refined type \( \decTy \) is defined by
\begin{equation*}
\msem{\decTy} = \set{(c_{i,j,k})_{i,j,k} \in \msem{\tr(\decTy)} \mid \tofun{(c_{i,j,k})_{i,j,k}} \in \semant{\decTy} }.
\end{equation*}
% \[
% \begin{array}{l}
% \msem{\PropTU{T_1}{U_1}\to \cdots \PropTU{T_m}{U_m}\to \PropTU{T}{U}}\\
% = \{(c_{i,j,k})_{i\in\set{1,\ldots,n},j\in\set{0,\ldots,m},k\in\set{0,1,\ldots,n}} \in [0,1]^{n(m+1)(n+1)}
% \mid \\\qquad
% %%c_{i,j,k}\in [0,1]\mbox{ for each \(i\in\set{1,\ldots,n},j\in\set{0,\ldots,m},k\in\set{0,1,\ldots,n}\)}\\\qquad
% c_{i,0,k}=0 \mbox{ for each \(i\in\set{1,\ldots,n},k\in\set{1,\ldots,n}\)}\\\qquad
% c_{i,j,0}=0 \mbox{ for each \(i\in\set{1,\ldots,n},j\in\set{1,\ldots,m}\)}\\\qquad
% c_{i,j,k}=0 \mbox{ for each \(i\in T\cup U, j\in\set{1,\ldots,m},k\in\set{1,\ldots,n}\)}\\\qquad
% c_{i,0,0}=0 \mbox{ for each \(i\in T\)}\\\qquad
% c_{i,0,0}=1 \mbox{ for each \(i\in U\)}\\\qquad
% c_{i,j,k}=0 \mbox{ for each \(i\in\set{1,\ldots,n},j\in\set{1,\ldots,m},
%   k\in T_j\cup U_j\)}\\\qquad
% \sum_{j\in\set{0,\ldots,m},k\in\set{0,\ldots,n}} c_{i,j,k} \le 1 \mbox{ for each \(i\in\set{1,\ldots,n}\)}
% \\\quad
% \}.
% \end{array}
% \]
% Let \( \decTy = \PropTU{T_1}{U_1}\to \cdots \PropTU{T_m}{U_m}\to \PropTU{T}{U} \).
% Note that \(c_{i,0,k}\) and \(c_{i,j,0}\) are always \(0\) for \(j>0, k>0\);
% these entries have been added just for a notational convenience.
% The tuplu must satisfies the condition \( c_{i,j,k} = 0 \) if \( i \in T \cup U \) or \( k \in T_i \cup U_i \);
% this condition ensures that, for each \( x \in \semant{\decTy} \), the tuple \( (c_{i,j,k})_{i,j,k} \in \msem{\decTy} \) such that \( \tofun{(c_{i,j,k})_{i,j,k}} \sim_{\decTy} x \) is unique.
\newcommand{\coef}{\mathbf{c}}
\newcommand{\coefb}{\mathbf{d}}
\begin{lemma}
  For every \( x \in \semant{\decTy} \), there exists \( \coef = (c_{i,j,k})_{i,j,k} \in \msem{\decTy} \) such that \( x \sim_{\decTy} \tofun{\coef} \).
\end{lemma}
\begin{proof}
  Let \( \decTy = \PropTU{T_1}{U_1} \to \dots \to \PropTU{T_m}{U_m} \to \PropTU{T}{U} \).
  Since \( x \in \semant{\decTy} \), it is affine on \( \semant{\PropTU{T_1}{U_1}} \times \dots \times \semant{\PropTU{T_m}{U_m}} \), i.e.~for every \( (v_1,\dots,v_m) \in \semant{\PropTU{T_1}{U_1}} \times \dots \times \semant{\PropTU{T_m}{U_m}} \) and \( i \in S_M \),
  \begin{equation*}
    x\,v_1\,\dots\,v_m\,i = c_{i,0,0} + \sum_{j=1}^{m} \sum_{k = 1}^{n} c_{i,j,k} v_j(k).
  \end{equation*}
  We can assume without loss of generality that \( c_{i,j,k} = 0 \) if \( k \in T_j \) since \( v_j(k) = 0 \) in this case.
  Let \( c_{i,0,k} = c_{i,j,0} = 0 \) for \( j,k \ge 0 \).
  Then \( (c_{i,j,k})_{i,j,k} \in \msem{\Prop^m \to \Prop} \) and \( \tofun{(c_{i,j,k})_{i,j,k}} \sim_{\decTy} x \).
  Therefore \( (c_{i,j,k})_{i,j,k} \in \msem{\decTy} \) since the latter implies \( \tofun{(c_{i,j,k})_{i,j,k}} \in \semant{\decTy} \).
\end{proof}
\begin{remark}\label{rem:decidableCanonicalRepl}
  For \( x \in \semant{\decTy} \), \( \coef \in \msem{\decTy} \) such that \( x \sim_{\decTy} \tofun{\coef} \) is not necessarily unique.
  The representation becomes unique if one imposes the following conditions.
  Assume that \( \decTy = \PropTU{T_1}{U_1} \to \dots \to \PropTU{T_m}{U_m} \to \PropTU{T}{U} \).
  \begin{itemize}
    \item \( c_{i,j,0} = c_{i,0,k} = 0 \) for \( j,k > 0 \).  Note that these coefficients are not used in \( \tofun{(c_{i,j,k})_{i,j,k}} \).
    \item \( c_{i,j,k} = 0 \) if \( k \in T_j \).  If \( k \in T_j \), then \( v_j(k) = 0 \) and thus \( c_{i,j,k} \) does not affect the value of \( \tofun{(c_{i,j,k})_{i,j,k}} \).
    So we can assume without loss of generality that \( c_{i,j,k} = 0 \).
    \item \( c_{i,j,k} = 0 \) if \( k \in U_j \).  If \( k \in U_j \), then \( v_j(k) = 1 \).
    So, by adding \( c_{i,j,k} \) to the constant part \( c_{i,0,0} \) if necessarily, we obtain another representation that belongs to the same \( \sim_{\decTy} \)-equivalence class and that \( c_{i,j,k} = 0 \).
  \end{itemize}
  A representation \( \coef \in \msem{\decTy} \) is \emph{canonical} if it satisfies the above conditions.
  It is not difficult to see that each \( \sim_{\decTy} \)-equivalence class has exactly one canonical representation.

  The second and third conditions are convenient for computing the affine semantics of \(\form_1\lor\form_2\) and
  \(\form_1\land \form_2\). Suppose \(\emptyset;\Delta\pM \form_i:\PropTU{T_i}{U_i}\) for \(i\in\set{1,2}\) and
  \(T_1\cup U_1\cup T_2\cup U_2=S_M\).
  If \((c_{i,j,k})_{i,j,k}\) and \((c'_{i,j,k})_{i,j,k}\) are canonical representations for the (affine) semantics
  of \(\lambda \Delta.\form_1\) and \(\lambda \Delta.\form_2\) respectively, then
%%  the semantics of \(\lambda \Delta.\form_1\lor\form_2\) can be obtained by the pointwise maximum:
%%  \((\max(c_{i,j,k},c'_{i,j,k}))_{i,j,k}\),
   the semantics \((c''_{i,j,k})_{i,j,k})\) of \(\lambda \Delta.\form_1\land\form_2\)
   is also obtained pointwise by:
   \[
    \begin{minipage}[b]{.8\textwidth} \centering
    \(
       c''_{i,j,k} = \left\{\begin{array}{ll}
   \min(c_{i,0,0},c'_{i,0,0}) & \mbox{if $j=k=0$}\\
   c'_{i,j,k} & \mbox{if $i\in U_1$}\\
   c_{i,j,k} & \mbox{if $i\in U_2\setminus U_1$}\\
   0 & \mbox{otherwise}
   \end{array}\right.
    \)
    \\[-1em] % a bit of negative space
    ~
    \end{minipage}
    \tag*{\qed}
  \]
%%%  In contrast, the pointwise operation is not valid for non-canonical representations.
%%%  For example, let \(S_M=\set{1}\), and define \((c_{1,j,k})_{j\in\set{0,1},k\in\set{0,1}}\) and
%%%  \((c'_{1,j,k})_{j\in\set{0,1},k\in\set{0,1}}\) by
%%%  \[
%%%  c_{1,j,k}=\left\{\begin{array}{ll}
%%%    1 & \mbox{if \(j=k=0\)}\\
%%%    0 & \mbox{otherwise}\\
%%%  \end{array}\right.
%%%  \qquad
%%%  c'_{1,j,k}=\left\{\begin{array}{ll}
%%%    \frac{1}{2} & \mbox{if \(j=k=0\) or \(j=k=1\)}\\
%%%    0 & \mbox{otherwise}\\
%%%  \end{array}\right.
%%%  \]
%%%  Both are representations for elements of \(\semant{\PropTU{\emptyset}{\set{1}}\to \PropTU{\emptyset}{\set{1}}}\),
%%%  and express the functions \(f(x)=1\) and \(f'(x)=\frac{1}{2}+\frac{1}{2}x\) respectively.
%%%  They represent funare equivalent up to \(\sim_{\PropTU{\emptyset}{\set{1}}\to \PropTU{\emptyset}{\set{1}}}\) (note that
%%%  the intended domain of \(f\) and \(f'\) is the singleton set \(\set{1}\)), hence their least upper bound
%%%  is \(\tofun{(c_{1,j,k})_{j,k}}\). The pointwise maximum, however, yields
%%%  \((c''_{1,j,k})_{j\in\set{0,1},k\in\set{0,1}}\) where
%%%  \[
%%%  c_{1,j,k}=\left\{\begin{array}{ll}
%%%    1 & \mbox{if \(j=k=0\)}\\
%%%    \frac{1}{2} & \mbox{if \(j=k=1\)}\\
%%%    0 & \mbox{otherwise}\\
%%%  \end{array}\right.
%%%  \]
%%%  (or, with the above functional representation, \(f''(x)=1+\frac{1}{2}x\)),
%%%  which is not a valid element of \(\msem{\PropTU{\emptyset}{\set{1}}\to \PropTU{\emptyset}{\set{1}}}\).
\end{remark}

As stated before,
we will describe the affine semantics of a well-typed PHFL formula using the first-order theory of reals.
Before doing so, however, we show an example of directly computing the affine semantics.
\begin{example}
  \label{ex:affine-semantics}
  Let \(M\) be the Markov chain \((\set{1,2}, \Ptr, \iAP, 1)\)
  where \(\iAP(p_i)=\set{i}\) for \(i\in\set{1,2}\), and \(\Ptr(1,1)=\Ptr(1,2)=0.5\) and \(\Ptr(2,2)=1\),
  as depicted below.
  %  in Figure~\ref{fig:markov}
  \begin{center}
    \includegraphics[scale=0.5]{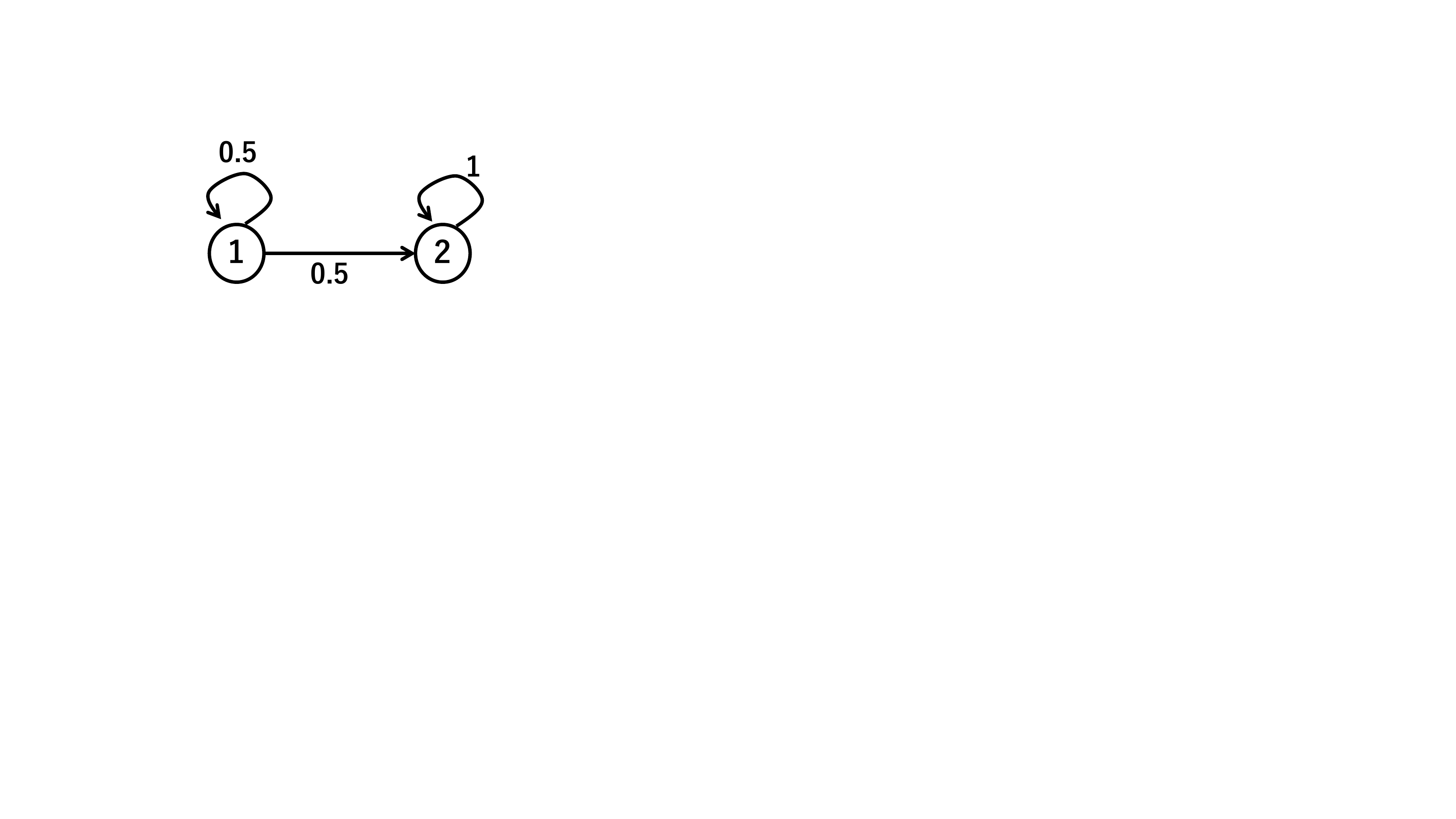}
  \end{center}
  Let \(\form\) be \(\mu X.\lambda Y.(p_2\land Y)\lor \Avg (p_1\land X(Y))\).
For \(\rho\in \msem{\decTyEnv}\),
we write \(\msem{\decTyEnv;\lamTE\pM \form:\decTy}(\rho)\) for
the (canonical) matrix representation of \(\semant{\tr(\decTyEnv)\p \lambda \tr(\lamTE).\form}(\tofun{\rho})\).
 Let us compute \(\msem{\emptyset;\emptyset\pM \form:\decTy}\)
    where \(\decTy=\PropTU{\emptyset}{\emptyset}\to\PropTU{\emptyset}{\emptyset}\).
%%    we omit the derivation \(\pi\) for the judgment since it is not important.
    
  Let \(\rho\) be \(\set{X\mapsto (v_{i,j,k})_{i\in\set{1,2},j\in\set{0,1},k\in\set{0,1,2}}}\).
    We write below an element \((c_{i,j,k})_{i,j,k}\) of
  \(\msem{\decTy}\)
  as a matrix \((M_1\; M_2)\), where \(M_i\) denotes \((c_{i,j,k})_{j,k}\).
  We can compute the affine semantics of subexpressions as follows.
%%    By the definition, \(\msem{X\COL \decTy;Y\COL\PropTU{\emptyset}{\emptyset}\pM Y:\PropTU{\emptyset}{\emptyset}}\) is:
  \begin{align*}
  \msem{X\COL \decTy;Y\COL\PropTU{\emptyset}{\emptyset}\pM Y:\PropTU{\emptyset}{\emptyset}}(\rho)&=
      \left(\begin{array}{lllclll}
         0 & 0 & 0 && 0 & 0 & 0\\ 
         0 & 1 & 0 && 0 & 0 & 1\\ 
      \end{array}\right)\\
  \msem{X\COL \decTy;Y\COL\PropTU{\emptyset}{\emptyset}\pM p_2:\PropTU{\set{1}}{\set{2}}}(\rho)&=
      \left(\begin{array}{lllclll}
         0 & 0 & 0 && 1 & 0 & 0\\ 
         0 & 0 & 0 && 0 & 0 & 0\\ 
      \end{array}\right)\\
      \msem{X\COL \decTy;Y\COL\PropTU{\emptyset}{\emptyset}\pM p_2\land Y:\PropTU{\set{1}}{\emptyset}}(\rho)\hspace*{-.3em}&=
      \left(\begin{array}{lllclll}
         0 & 0 & 0 && 0 & 0 & 0\\ 
         0 & 0 & 0 && 0 & 0 & 1\\ 
      \end{array}\right)\\
&\hspace*{-1cm}\mbox{(Recall the discussion in the latter half of Remark~\ref{rem:decidableCanonicalRepl}.)}\\
  \msem{X\COL \decTy;Y\COL\PropTU{\emptyset}{\emptyset}\pM X(Y):\PropTU{\emptyset}{\emptyset}}(\rho)&=
      \left(\begin{array}{lllclll}
         v_{1,0,0} & 0 & 0 && v_{2,0,0} & 0 & 0\\ 
         0 & v_{1,1,1} & v_{1,1,2} && 0 & v_{2,1,1} & v_{2,1,2}\\ 
      \end{array}\right)
  \end{align*}
  \begin{align*}&
  \msem{X\COL \decTy;Y\COL\PropTU{\emptyset}{\emptyset}\pM \Avg X(Y):\PropTU{\emptyset}{\emptyset}}(\rho)\\&\qquad =
      \left(\begin{array}{lllclll}
         \frac{1}{2}(v_{1,0,0}+v_{2,0,0}) & 0 & 0 && v_{2,0,0} & 0 & 0\\ 
         0 & \frac{1}{2}(v_{1,1,1}+v_{2,1,1}) & \frac{1}{2}(v_{1,1,2}+v_{2,1,2}) && 0 & v_{2,1,1} & v_{2,1,2}\\ 
      \end{array}\right)
\hfill      \tag*{($\Avg$ can be computed pointwise)}\\&
  \msem{X\COL \decTy;Y\COL\PropTU{\emptyset}{\emptyset}\pM p_1:\PropTU{\set{2}}{\set{1}}}(\rho)=
      \left(\begin{array}{lllclll}
         1 & 0 & 0 && 0 & 0 & 0\\ 
         0 & 0 & 0 && 0 & 0 & 0\\ 
      \end{array}\right)\\&
  \msem{X\COL \decTy;Y\COL\PropTU{\emptyset}{\emptyset}\pM p_1\land \Avg X(Y):\PropTU{\set{2}}{\emptyset}}(\rho)\\&\qquad =
      \left(\begin{array}{lllclll}
         \frac{1}{2}(v_{1,0,0}+v_{2,0,0}) & 0 & 0 && 0 & 0 & 0\\ 
         0 & \frac{1}{2}(v_{1,1,1}+v_{2,1,1}) & \frac{1}{2}(v_{1,1,2}+v_{2,1,2}) && 0 & 0 & 0
      \end{array}\right)\\&
\tag{Remark~\ref{rem:decidableCanonicalRepl}.}\\&
    \msem{X\COL \decTy;Y\COL\PropTU{\emptyset}{\emptyset}\pM
      (p_2\land Y)\lor (p_1\land \Avg X(Y)):\PropTU{\emptyset}{\emptyset}}(\rho)\\&\qquad =
      \left(\begin{array}{lllclll}
         \frac{1}{2}(v_{1,0,0}+v_{2,0,0}) & 0 & 0 && 0 & 0 & 0\\ 
         0 & \frac{1}{2}(v_{1,1,1}+v_{2,1,1}) & \frac{1}{2}(v_{1,1,2}+v_{2,1,2}) && 0 & 0 & 1
      \end{array}\right)\\&
\tag{Analogous to the computation of \(\land\) discussed in Remark~\ref{rem:decidableCanonicalRepl}.}\\&
    \msem{X\COL \decTy;\emptyset\pM
     \lambda Y. (p_2\land Y)\lor (p_1\land \Avg X(Y)):\decTy}(\rho)\\&\qquad =
      \left(\begin{array}{lllclll}
         \frac{1}{2}(v_{1,0,0}+v_{2,0,0}) & 0 & 0 && 0 & 0 & 0\\ 
         0 & \frac{1}{2}(v_{1,1,1}+v_{2,1,1}) & \frac{1}{2}(v_{1,1,2}+v_{2,1,2}) && 0 & 0 & 1
      \end{array}\right)
\end{align*}
  Thus, we have:
  \[\small
  \begin{array}{l}
  \msem{\emptyset;\emptyset\pM \form:\decTy}(\emptyset) \\
=  \LFP(\lambda 
      \left(\begin{array}{lllclll}
         v_{1,0,0} & \_ & \_ && v_{2,0,0} & \_ & \_\\ 
         \_ & v_{1,1,1} & v_{1,1,2} && \_ & v_{2,1,1} & v_{2,1,2}
      \end{array}\right). \\\qquad\qquad
      \left(\begin{array}{lllclll}
         \frac{1}{2}(v_{1,0,0}+v_{2,0,0}) & 0 & 0 && 0 & 0 & 0\\ 
         0 & \frac{1}{2}(v_{1,1,1}+v_{2,1,1}) & \frac{1}{2}(v_{1,1,2}+v_{2,1,2}) && 0 & 0 & 1
      \end{array}\right)\\
  = 
      \left(\begin{array}{lllclll}
         0 & 0 & 0 && 0 & 0 & 0\\ 
         0 & 0 & 1 && 0 & 0 & 1\\ 
      \end{array}\right).\\
  \end{array}\\
  \]
  Using the result above, we obtain: %%can compute the semantics of \(\form(p_2)\) as;
  \[
  \msem{\emptyset;\emptyset\pM \form(p):\decTy}(\emptyset)
  = 
      \left(\begin{array}{lllclll}
        1 & 0 & 0 && 1 & 0 & 0\\
        \end{array}
         \right).
         \]
         Thus, \(\form(p_2)\) holds at state \(1\) with probability \(1\).
         \qed
\end{example}

In the example above, we have directly computed the affine semantics.
In general, however, we describe the affine semantics \((c_{i,j,k})_{i,j,k}\) by using logical formulas,
to deal with arbitrary alternations of fixpoint operators.
All the operations and properties on \( \semant{\decTy} \)
required for computing the affine semantics
are definable by using the first-order theory of reals.
Assume \( \decTy = \PropTU{T_1}{U_1} \to \dots \to \PropTU{T_m}{U_m} \to \PropTU{T}{U} \) and
let \( \coef = (c_{i,j,k})_{i,j,k} \in \msem{\decTy} \).
For example, the value \( \tofun{\mathbf{c}}\,v_1\,\dots\,v_m\,i \) can be represented as a term
\begin{equation*}
  c_{i,0,0} + \sum_{j=1}^{m} \sum_{k = 1}^{n} c_{i,j,k} v_j(k).
\end{equation*}
Then, for \( \coef, \coefb \in \msem{\decTy}\), the relation \( \coef \precsim_{\decTy} \coefb \) is written as
\begin{align*}
  &
  \forall v_1,\dots,v_m \in [0,1]^{n}.
  \\
  &
  \Big(\bigwedge_{j=1,\dots,m} v_j \in \msem{\PropTU{T_j}{U_j}}\Big)
  \Rightarrow \bigwedge_{i = 1,\dots,n} \Big[\big(\tofun{\coef}\,v_1\,\dots\,v_k\,i\big) \le \big(\tofun{\coefb}\,v_1\,\dots\,v_k\,i \big)\Big]
  %c_{i,0,0} + \sum_{j=1}^{m} \sum_{k = 1}^{n} c_{i,j,k} v_j(k) \le c'_{i,0,0} + \sum_{j=1}^{m} \sum_{k = 1}^{n} c'_{i,j,k} v_j(k)
\end{align*}
and \( \tofun{\coef} \sim_{\decTy} \tofun{\coefb} \) as
\begin{equation*}
  \tofun{\coef} \precsim_{\decTy} \tofun{\coefb} ~~\wedge~~ \tofun{\coefb} \precsim_{\decTy} \tofun{\coef}.
\end{equation*}
The meets and joins are also describable: for example, \( \tofun{\coef_1} \sqcap \tofun{\coef_2} \sim_{\decTy} \tofun{\coefb} \) is equivalent to
\begin{align*}
  &
  \forall v_1,\dots,v_m \in [0,1]^{n}.
  \Big(\bigwedge_{j=1,\dots,m} v_j \in \msem{\PropTU{T_j}{U_j}}\Big)
  \\
  & \qquad 
  \Rightarrow \bigwedge_{i = 1,\dots,n} \Big[\big(\tofun{\coefb}\,v_1\,\dots\,v_k\,i\big) = \max\Big((\tofun{\coef_1}\,v_1\,\dots\,v_k\,i),~ (\tofun{\coef_2}\,v_1\,\dots\,v_k\,i)\Big) \Big];
  %c_{i,0,0} + \sum_{j=1}^{m} \sum_{k = 1}^{n} c_{i,j,k} v_j(k) \le c'_{i,0,0} + \sum_{j=1}^{m} \sum_{k = 1}^{n} c'_{i,j,k} v_j(k)
\end{align*}
note that \( \max \) as well as \( \min \) is an operation definable in the first-order theory of reals.

\begin{lemma}\label{lem:decidableFormulaConstruction}
  Suppose \( \decTyEnv;\lamTE \pM \phi : \decTy \) and \( \decTyEnv = \{X_1 : \decTy_1, \dots, X_N : \decTy_N\} \).
  Then one can effectively construct a formula \( \Phi \) of the first-order real arithmetic such that, for every \( \coef_\ell \in \msem{\decTy_\ell} \) (\( \ell = 1,\dots,N \)) and \( \coefb \in \msem{\lamTE \to \decTy} \),
  \begin{equation*}
    \Phi(\coef_1,\dots,\coef_N,\coefb)
    ~~\Leftrightarrow~~
    \semant{\lambda \lamTE.\phi}(\{ X_\ell \mapsto \tofun{\coef_\ell} \mid 1 \le \ell \le N \}) \sim_{\lamTE \to \decTy} \tofun{\coefb}.
  \end{equation*}
\end{lemma}
\begin{proof}
  By induction on the structure of derivation \( \decTyEnv;\lamTE \pM \phi : \decTy \).
  The most cases are easy.
  For example, consider the case that the last rule is (\textsc{T-Conj}).
  Then \( \phi = \phi_1 \wedge \phi_2 \) and \(\decTyEnv;\lamTE \pM \phi_i : \PropTU{T_i}{U_i} \) for \( i = 1,2 \).
  By the induction hypothesis, we have predicates \( \Phi_1 \) and \( \Phi_2 \) representing \( \semant{\lambda \lamTE.\phi_1} \) and \( \semant{\lambda \lamTE.\phi_2} \).
  Then \( \Phi(\coef_1,\dots,\coef_N, \coefb) \) is defined as
  \begin{align*}
    &
    \exists \coefb_1 \in \msem{\lamTE\to\PropTU{T_1}{U_1}}.
    ~
    \exists \coefb_2 \in \msem{\lamTE\to\PropTU{T_2}{U_2}}.
    \\ & \quad
    \Phi_1(\coef_1,\dots,\coef_N, \coefb_1)
    ~~\wedge~~
    \Phi_2(\coef_1,\dots,\coef_N, \coefb_2)
    ~~\wedge~~
    \tofun{\coefb} \sim_{\decTy} \tofun{\coefb_1} \sqcap \tofun{\coefb_2}.
  \end{align*}
  Since \( \semant{\lambda \lamTE.\phi_i}(\{ X_{\ell} \mapsto \tofun{\coef_\ell} \mid 1 \le \ell \le N \}) \in \semant{\lamTE \to \PropTU{T_i}{U_i}} \) by Theorem~\ref{thm:decidableTypeSoundness}, there exists a \( \coefb_i \in \msem{\lamTE\to\PropTU{T_i}{U_i}} \) such that \( \Phi_i(\coef_1,\dots,\coef_N,\coefb_i) \).

  The only non-trivial cases are fixpoints.
  Consider the case that the last rule is (\textsc{T-Mu}).
  Then \( \phi = \nu X. \phi_0 \) and \( \decTyEnv, X : \decTy; \lamTE \pM \phi_0 : \decTy \).
  By the induction hypothesis, we have \( \Phi_0 \) representing the semantics of \( \phi_0 \).
  Then \( \Phi(\coef_1,\dots,\coef_N,\coefb) \) is defined as
  \begin{equation*}
    \Phi_0(\coef_1,\dots,\coef_N,\coefb,\coefb)
    ~~\wedge~~
    \forall \coefb' \in \msem{\decTy}.
    \Phi_0(\coef_1,\dots,\coef_N,\coefb',\coefb') \Rightarrow \tofun{\coefb} \precsim_{\decTy} \tofun{\coefb'}.
  \end{equation*}
  The first condition says that \( \tofun{\coefb} \) is a fixpoint of \( \semant{\lambda X.\phi_0} \) modulo \( \sim_{\decTy} \).
  The second condition says that \( \tofun{\coefb} \) is a least element in the set of fixpoints \( \tofun{\coefb'} \) of \( \semant{\lambda X.\phi_0} \) modulo \( \sim_{\decTy} \).
  Correctness of the above formula follows from Corollary~\ref{cor:decidableFixpointRepresentation}.
\end{proof}

\begin{theorem}
  \label{th:decidable-fragment}
  Given a closed formula \(\form\) and a Markov chain \(M\) such that
  \(\emptyset;\emptyset \pM \form:\PropTU{T}{U}\),
  it is decidable whether \(M\models \form\).
\end{theorem}
\begin{proof}
  Assume that \( \inStateM{M} = 1 \).
  By Lemma~\ref{lem:decidableFormulaConstruction}, one can effectively construct a formula \( \Phi \) of the first-order real arithmetic such that \( \Phi(r_1,\dots,r_n) \) is valid if and only if \( \semant{\form}_M(\emptyset)(i) = r_i \) for every \( i = 1,\dots,n\).
  Hence \( M \models \form \) if and only if \( \exists r_2,\dots,r_n. \Phi(1,r_2,\dots,r_n) \).
  The validity of this formula is decidable~\cite{Tarski51}.
\end{proof}

\begin{remark}
  The formula obtained in the proof above
  contains both universal and existential quantifiers in general;
  hence the complexity is doubly exponential time in general~\cite{Davenport88}.
  However, if \(\form\) is of the form \([\form']_{\ge r}\) where
  \(\form'\) contains no occurrences of \(\nu\) and \([\cdot]_J\),
  then we can express \( \Phi \) 
  %the condition \(c_{1,0,0}=1\)
  using only existential quantifiers, and appeal 
  to a decision algorithm for the existential theory of the reals, whose complexity
  is PSPACE~\cite{Canny88}.
  In fact, to deal with \(\mu\)-formulas in the restricted fragment,
  it suffices to consider only inequalities of the form
  \(\LFP(h)\preceq_{\decTy} x\), which can be represented by
  an existential formula \(\exists x'.(x'\precsim_{\decTy} x \land h\,x' \precsim_{\decTy} x')\).
  Note that \( \precsim \) is definable by a quantifier-free formula, as stated in the lemma below.
  Note also that the semantics of conjunctions and disjunctions can be expressed without using
  quantifiers, as discussed in Remark~\ref{rem:decidableCanonicalRepl}.
\end{remark}

% For \(\decTy=\PropTU{T_1}{U_1}\to \cdots \PropTU{T_m}{U_m}\to \PropTU{T}{U}\),
% we define the order \(\LEQ_{\msem{\decTy}}\) by:
% \[
% \begin{array}{l}
% (c_{i,j,k})_{i,j,k}\LEQ_{\msem{\decTy}}(c'_{i,j,k})_{i,j,k}
%   \defarrow
% \tofun{(c_{i,j,k})_{i,j,k}}\LEQ_{{\decTy}}\tofun{(c'_{i,j,k})_{i,j,k}}
% %%%\forall i\in\set{1,\ldots,n}.\forall J\subseteq \set{1,\ldots,m}\times\set{1,\ldots,n}.
% %%%c_{i,0,0}+\sum_{(j,k)\in J} c_{i,j,k} \le  c'_{i,0,0}+\sum_{(j,k)\in J} c'_{i,j,k}.
% \end{array}
% \]
% (where \(\tofun{(c_{i,j,k})_{i,j,k}}\) above is identified with its restriction to \(D_{\decTy}\)).

% Note that \(\LEQ_{\msem{\decTy}}\) is different from the element-wise order on reals.
% For example, for \(n=m=1\), we have \((c_{i,j,k})_{i,j,k}\LEQ_{\msem{\decTy}}
% \top_{\msem{\decTy}}\) for \(c_{1,0,0}=c_{1,1,1}=\frac{1}{2}\) and
% \(c_{1,0,1}=c_{0,1,0}=0\), but
% \(c_{1,1,1}=\frac{1}{2} \not\le 0=(\top_{\msem{\decTy}})_{1,1,1}\))
% Based on the definition, the inequality \(\LEQ_{\msem{\decTy}}\) above is definable in
% the first-order arithmetic on reals by using universal quantifiers. The following lemma
% implies that  \(\LEQ_{\msem{\decTy}}\) is actually definable without using quantifiers on reals.
\begin{lemma}
  \label{lem:quantifier-free-representation-of-leq}
  Let \( \decTy = \PropTU{T_1}{U_1} \to \dots \to \PropTU{T_m}{U_m} \to \PropTU{T}{U} \).
  For \( J \subseteq \set{1,\dots,m} \times \set{1,\dots,n} \), we write \( J \in \mathcal{J}(\decTy) \) if the characteristic function belongs to \( \prod_{i=1}^{m} \semant{\PropTU{T_i}{U_i}} \), i.e.,
  \begin{itemize}
    \item if \( k \in T_j \), then \( (k,j) \notin J \), and
    \item if \( k \in U_j \), then \( (k,j) \in J \).
  \end{itemize}
  Let \((c_{i,j,k})_{i,j,k}, (c'_{i,j,k})_{i,j,k} \in \msem{\decTy}\).
  Then, 
  \(\tofun{(c_{i,j,k})_{i,j,k}} \preceq_{\decTy} \tofun{(c'_{i,j,k})_{i,j,k}}\) if and only if:
  \[
  \bigwedge_{i\in\set{1,\ldots,n}} \bigwedge_{J \in \mathcal{J}(\decTy)}
  \left(
  c_{i,0,0}+\sum_{(j,k)\in J} c_{i,j,k} \le  c'_{i,0,0}+\sum_{(j,k)\in J} c'_{i,j,k}
  \right).
  \]
\end{lemma}
\begin{proof}
  \begin{itemize}
    \item[] \textbf{If:} 
Suppose \(g_1\in \semant{\PropTU{T_1}{U_1}},\ldots,g_m\in \semant{\PropTU{T_m}{U_m}}\).
We regard \( (g_1,\dots,g_m) = (g_j(k))_{j\in\set{1,\dots,m}, k \in \set{1,\dots,n}} \) as an \( mn \)-dimensional vector.
We also identify \( J \in \mathcal{J}(\decTy) \) with \( (h_j(k))_{j,k} \) where \( h_j(k) = 1 \) if \( (j,k) \in J \) and \( h_j(k) = 0 \) if \( (j,k) \notin J \).

Let \( 0 = s_0 < s_1 < s_2 < \dots < s_N = 1 \) be the sequence obtained by sorting \( \{ g_j(k) \mid 1 \le j \le m, 1 \le k \le n \} \cup \{ 0, 1 \} \).
For each \( \ell > 0 \), we define \( J_{\ell} \) and \( r_\ell \) by
\begin{equation*}
  J_{\ell} = \{ (j,k) \mid s_\ell \le h_j(k) \}
  \qquad\mbox{and}\qquad
  r_\ell = s_\ell - s_{\ell-1}.
\end{equation*}
Then
\begin{equation*}
  (g_1,\dots,g_m) = \sum_{\ell} r_\ell J_\ell.
\end{equation*}
Hence
\begin{align*}
  \tofun{(c_{i,j,k})_{i,j,k}}\,g_1\,\dots\,g_m\,i
  &=
  \sum_{\ell} r_\ell (c_{i,0,0} + \sum_{(j,k) \in J_{\ell}} c_{i,j,k})
  \\ &\le
  \sum_{\ell} r_\ell (c'_{i,0,0} + \sum_{(j,k) \in J_{\ell}} c'_{i,j,k})
  \\ &=
  \tofun{(c_{i,j,k})_{i,j,k}}\,g_1\,\dots\,g_m\,i.
\end{align*}

% We need to show
% \[\tofun{(c_{i,j,k})_{i,j,k}}g_1\,\cdots\,g_m\,i \le
% \tofun{(c'_{i,j,k})_{i,j,k}}g_1\,\cdots\,g_m\,i\]
% Let
% \[g_{j_1}(k_1) \le g_{j_2}(k_2) \le \cdots \le g_{j_{mn}}(k_{mn})\]
% be an increasing sequence obtained by sorting 
% \(g_{j}(k)\) (\(j\in \set{1,\ldots,m},k\in\set{1,\ldots,n}\)).
% Then, we have
% \[
% \begin{array}{l}
% \tofun{(c_{i,j,k})_{i,j,k}}g_1\,\cdots\,g_m\,i\\
% = (c_{i,0,0}+\sum_{1\le p\le mn}c_{i,j_p,k_p})\cdot g_{j_1}(k_1)
% + (c_{i,0,0}+\sum_{2\le p\le mn}c_{i,j_p,k_p})\cdot (g_{j_2}(k_2) - g_{j_1}(k_1))\\\quad
% + \cdots
% + (c_{i,0,0}+c_{i,j_{mn},k_{mn}})\cdot (g_{j_{mn}}(k_{mn}) - g_{j_{mn-1}}(k_{mn-1})
% + c_{i,0,0}\cdot (1-g_{j_{mn}}(k_{mn}))\\
% \le (c'_{i,0,0}+\sum_{1\le p\le mn}c'_{i,j_p,k_p})\cdot g_{j_1}(k_1)
% + (c'_{i,0,0}+\sum_{2\le p\le mn}c'_{i,j_p,k_p})\cdot (g_{j_2}(k_2) - g_{j_1}(k_1))\\\quad
% + \cdots
% + (c'_{i,0,0}+c'_{i,j_{mn},k_{mn}})\cdot (g_{j_{mn}}(k_{mn}) - g_{j_{mn-1}}(k_{mn-1})
% + c'_{i,0,0}\cdot (1-g_{j_{mn}}(k_{mn}))\\
% = \tofun{(c'_{i,j,k})_{i,j,k}}g_1\,\cdots\,g_m\,i,
% \end{array}
% \]
% as required.

  \item[] \textbf{Only if:}
    % Let \(\decTy=\PropTU{T_1}{U_1}\to \cdots \PropTU{T_m}{U_m}\to \PropTU{T}{U}\)
    % and \(J\subseteq \set{1,\ldots,m}\times \set{1,\ldots,n}\).
%%  Suppose \(\tofun{(c_{i,j,k})_{i,j,k}}\LEQ_{\tr(\decTy)} \tofun{(c'_{i,j,k})_{i,j,k}}\),
%%  and \(J\subseteq \set{1,\ldots,m}\times\set{1,\ldots,n}\).
  For each \( J \in \mathcal{J}(\decTy) \), let %%\(g_j\in \semant{\PropTU{T_j}{U_j}}\)
  \(g_j\) be the element of \(\semant{\Quant}\)
  such that \(g_j(k)=1\) if \((j,k)\in J\)
and \(g_j(k)=0\) if \((j,k)\not\in J\).
Then \( g_j \in \semant{\PropTU{T_j}{U_j}} \) for every \( j = 1,\dots,m\).
We have
\[c_{i,0,0}+\sum_{(j,k)\in J} c_{i,j,k} =
\tofun{(c_{i,j,k})_{i,j,k}}\,g_1\,\cdots\,g_m\,i \le
\tofun{(c'_{i,j,k})_{i,j,k}}\,g_1\,\cdots\,g_m\,i =
c'_{i,0,0}+\sum_{(j,k)\in J} c'_{i,j,k},
\]
as required.
\end{itemize}  
\end{proof}

\subsection{Expressivity}
\label{sec:expressivity}
We have already seen in Example~\ref{ex:affine-formula} and Lemma~\ref{lem:mup-subsumed}
that the decidable fragment of the order-1 \PHFL{} model checking problem
strictly subsumes $\mup{}$-calculus model checking.
To provide a further evidence of the expressivity of the decidable fragment,
in this subsection, we show that the termination problem for
recursive Markov chains (or,
probabilistic pushdown systems)~\cite{DBLP:journals/jacm/EtessamiY09,DBLP:journals/fmsd/BrazdilEKK13} can be encoded into the decidable fragment of
order-1 \PHFL{} model checking.
Since recursive Markov chains are known to be equivalent to order-1 
probabilistic higher-order recursion schemes (\PHORS{})~\cite{PHORS},
we encode below the termination problem
for the latter into the \PHFL{} model checking problem.

We first recall the definition of \PHORS{}, specialized for order 1.
\begin{definition}
  An order-1 \PHORS{} \(\GRAM\) is a triple \((\NONTERMS, \RULES, t)\),
  where:
  \begin{itemize}
    \item \(\NONTERMS\) is a finite map from (order-1) variables to their arities;
    \item \(\RULES\) is a finite set of rules of the form:
      \[ X\, y_1\,\cdots, y_{\NONTERMS(X)}\to t_L \C{p}t_R,\]
      where \(t_L, t_R\) range over the set of \((\NONTERMS, \set{y_1,\ldots,y_{\NONTERMS(X)}})\)-terms.
      Here, the set of \((\NONTERMS, V)\)-terms, ranged over by \(t\), is
      defined by the grammar:
      \[ t ::= \Te \mid y \mid X\,t_1\,\cdots,t_{\NONTERMS(X)},\]
      where \(y\) and \(X\) range over \(V\) and \(\dom(\NONTERMS)\) respectively,
      and \(\Te\) is a special symbol denoting termination.
    \item \(t\) is a \((\NONTERMS,\emptyset)\)-term.
  \end{itemize}
  For \(\GRAM=(\NONTERMS_\GRAM,\RULES_\GRAM,t_\GRAM)\),
  the reduction relation \(t\redswith{\pi,p}{\GRAM} t'\) on terms (where
  \(\pi\in\set{L,R}^*\) and \(p\in [0,1]\)) is defined by:
  \infrule{}{t\redswith{\epsilon, 1}{\GRAM} t}
  \infrule{t\redswith{\pi, q}{\GRAM} X\,t_1\,\cdots\,t_k \andalso
    X\,y_1\,\cdots\,y_k \to t_L\C{p}t_R\in \RULES_\GRAM}
          {t\redswith{\pi L, q\cdot p}{\GRAM}[t_1/y_1,\ldots,t_k/y_k]t_L}
  \infrule{t\redswith{\pi, q}{\GRAM} X\,t_1\,\cdots\,t_k \andalso
    X\,y_1\,\cdots\,y_k \to t_L\C{p}t_R\in \RULES_\GRAM}
          {t\redswith{\pi R, q\cdot (1-p)}{\GRAM}[t_1/y_1,\ldots,t_k/y_k]t_R}
          Note that the reduction is deterministic:
          for \(\pi\in\set{L,R}^*\), there exists at most one \((p,t')\) such that
           \(t\redswith{\pi, p}{\GRAM} t'\).
          The \emph{termination probability} of \(\GRAM=(\NONTERMS_\GRAM,\RULES_\GRAM,t_\GRAM)\),
          written \(\TP(\GRAM)\), is defined by:
          \[
          \begin{array}{l}
          \TP(\GRAM,t,\pi) = \left\{\begin{array}{ll}
          p & \mbox{if $t\redswith{\pi,p}{\GRAM} \Te$}\\
           0 & \mbox{otherwise}\end{array}\right.\\
          \TP(\GRAM) = \sum_{\pi\in \set{L,R}^*} \TP(\GRAM,t_\GRAM,\pi).
          \end{array}
          \]
\end{definition}
The following example has been taken from \cite{PHORS}.
\begin{example}
\label{ex:random-walk}
Let \(\GRAM\) be the order-1 \PHORS{} \((\NONTERMS, \RULES, X_0)\),
where:
\[
\begin{array}{l}
\NONTERMS=\set{X_0\mapsto 0, F\mapsto 1}\\
\RULES = \set{X_0\ =\ X_1\,\Te\C{1}X_0,\quad
X_1\,y\ =\ y \C{p} X_1(X_1\,y)}.
\end{array}
\]
Then, we have \(X_0 \redswith{L,1}{\GRAM} X_1\,\Te \redswith{L,p}{\GRAM}\Te\) and
\(X_0\redswith{L,1}{\GRAM} X_1\,\Te \redswith{R,1-p}{\GRAM} X_1(X_1\,\Te)
\redswith{L,p}{\GRAM} X_1\,\Te\redswith{L,p}{\GRAM} \Te\).
The termination probability
\(\TP(\GRAM)\) is \(1\) if \(p\ge \frac{1}{2}\) and \(\frac{p}{1-p}\) if \(p<\frac{1}{2}\). \qed
\end{example}
We now encode an order-1 \PHORS{}
\(\GRAM\) into a Markov chain \(M_\GRAM\) and a closed order-1 PHFL formula \(\form_\GRAM\)
such that \(\emptyset;\emptyset\pM \form_\GRAM:\PropTU{\emptyset}{\emptyset}\)
and \(\TP(\GRAM)\) coincides with the probability
that \(\form_\GRAM\) holds at the initial state of \(M_\GRAM\) (i.e.,
\(\TP(\GRAM=\semant{\form_\GRAM}(\emptyset)(\inState)\)).
We first give a construction of \(M\). Let \(\set{p_1,\ldots,p_m}\)
be the set \(\set{p, 1-p \mid \C{p} \mbox{ occurs in }\GRAM}\). We assume \(p_1<
p_2 < \cdots < p_m\).
We define the Markov chain \(M_\GRAM = (S, P, \iAP, \inState)\) as follows.
\begin{itemize}
\item \(S = \{s_0, s_1, \dots, s_{m + 1}\}\),
\item \(P\) satisfies \(P(s_0, s_1) = p_1\), \(P(s_0, s_i) = p_i - p_{i - 1}\) for \(2 \leq i \leq m\),
\(P(s_0, s_{m + 1}) = 1 - p_m\), \(P(s_i, s_0) = 1\) for \(1 \leq i \leq m + 1\) and \(P(s_i, s_j) = 0\) otherwise,
\item \(\iAP(P_i) = \{s_i\}\) for each \(0 \leq i \leq m + 1\), and
\item \(\inState = s_0\).
\end{itemize}
Note that \(\semant{\Avg (P_1\lor \cdots \lor P_i)}(\emptyset)(s_0) = p_i\)
and \(\semant{\Avg (P_{i+1}\lor \cdots \lor P_{m+1})}(\emptyset)(s_0) = 1-p_i\).
Now for each term of \(t\) of PHORS \(\GRAM\), we construct
a formula \(\toPHFL{t}\), so that the termination probability
of \(t\) coincides with the probability that \(\form_t\) holds at \(s_0\).
The translation is given by:
\[
\begin{array}{l}
  \toPHFL{\Te} = \true\qquad  \toPHFL{y} = y\qquad
  \toPHFL{X_i\,t_1\,\cdots\,t_k} = X_i\,\toPHFL{t_1}\,\cdots\,\toPHFL{t_k}.
\end{array}
\]
For
each rule
\(X_i\,y_i\,\cdots\,y_{\NONTERMS(X_i)}=t_{L}\C{p_{i}} t_{R}\) of \(\RULES\),
%%\(\RULES_\GRAM=\set{X_1\,y_1\,\cdots\,y_{\NONTERMS(X_1)}=t_{1,L}\C{p_{i_1}} t_{1,R},
%%  \ldots, X_\ell\,y_1\,\cdots\,y_{\NONTERMS(X_\ell)}=t_{\ell,L}\C{p_{i_\ell}} t_{\ell,R}}\),
we construct the following equality on PHFL formulas:
\[
%%X_1 = \lambda y_1,\ldots,y_{\NONTERMS(X_1)}.
%%\Avg(((P_1\lor \cdots \lor P_{i_1})\land \Avg \toPHFL{t_L})
%%\lor ((P_{i+1}\lor \cdots \lor P_{m+1})\land \Avg \toPHFL{t_R})). 
X_i = \lambda y_1,\ldots,y_{\NONTERMS(X_i)}.
%%P_0\land (
\Avg(((P_1\lor \cdots \lor P_i)\land \Avg \toPHFL{t_L})
\lor ((P_{i+1}\lor \cdots \lor P_{m+1})\land \Avg \toPHFL{t_R})). %).
\]
We thus obtain a system of mutually recursive equations
\(\set{X_1=\form_1,\cdots,X_k=\form_k}\), whose least solution \(\theta\) (which
maps each \(X_i\) to a formula \(\psi_i\) that satisfies \(\psi_i = \theta\form_i\))
can be represented by using the least fixpoint operators in an obvious manner.
We then let \(\form_\GRAM\) be \(\theta\toPHFL{t_\GRAM}\).

\begin{example}
  \label{ex:encoding}
  Recall \(\GRAM\) in Example~\ref{ex:random-walk}. Assuming \(p<\frac{1}{2}\),
  \(M_\GRAM\) is:
  \begin{itemize}
\item \(S = \set{s_0, s_1, s_2, s_3,s_4}\).
\item \(P(s_0, s_1) = p\), \(P(s_0, s_2) = 1-2p\), \(P(s_0, s_3) = p\), \(P(s_0,s_4)=0\),
\(P(s_i, s_0) = 1\) for \(i\in\set{1,2,3,4}\) and \(P(s_i, s_j) = 0\) for \(i,j>0\).
\item \(\iAP(P_i) = \{s_i\}\) for each \(i\in\set{0,\dots,4}\), and
\item \(\inState = s_0\).
\end{itemize}
  The rules of \(\GRAM\) are translated to the following equations:
  \[
  \begin{array}{l}
    X_0 = \Avg (((P_1\lor P_2\lor P_3)\land \Avg(X_1\,\Te))\lor (P_4\land \Avg X_0 ))\\
    X_1 = \lambda y.
    \Avg ((P_1\land \Avg y)\lor ((P_2\lor P_3\lor P_4)\land  \Avg(X_1(X_1\,y))))
    \end{array}
  \]
  Thus, \(\form_\GRAM\) is given as:
  \[
  \mu X_0.\Avg (((P_1\lor P_2\lor P_3)\land \Avg(\psi_1\,\Te))\lor (P_4\land \Avg X_0 )),
  \]
  where \(\psi_1\) is:
  \[
  \mu X_1.\lambda y.
  \Avg ((P_1\land \Avg y)\lor ((P_2\lor P_3\lor P_4)\land  \Avg(X_1(X_1\,y)))).
  \]
Actually, \(\form_\GRAM\) can be simplified to \(\psi_1\,\Te\) in this case.
  \qed
\end{example}

To see the correctness of the above encoding, 
recall that
\(\semant{\Avg (P_1\lor \cdots \lor P_i)}(\emptyset)(s_0) = p_i\)
and \(\semant{\Avg (P_{i+1}\lor \cdots \lor P_{m+1})}(\emptyset)(s_0) = 1-p_i\).
Thus, by the equality on \(X_i\) above,
the probability that \(\theta(X_i\,\toPHFL{t_1}\,\cdots,\toPHFL{t_{\NONTERMS(X_i)}})\) holds
at \(s_0\) is equivalent to
\(p_i\cdot q_L+(1-p_i)\cdot q_R\), where \(q_L\) and \(q_R\)
are respectively the probabilities that
\(\theta(\toPHFL{[t_1/y_1,\ldots,t_{\NONTERMS(X_i)}/y_{\NONTERMS(X_i)}]t_L})\)
and 
\(\theta(\toPHFL{[t_1/y_1,\ldots,t_{\NONTERMS(X_i)}/y_{\NONTERMS(X_i)}]t_R})\)
hold at \(s_0\).
Thus,
the formula \(\theta(X_i\,\toPHFL{t_1}\,\cdots,\toPHFL{t_{\NONTERMS(X_i)}})\) 
mimics the termination probability of \(X_i\,\toPHFL{t_1}\,\cdots,\toPHFL{t_{\NONTERMS(X_i)}}\),
which is equivalent to \(p_i\cdot q'_L+(1-p_i)\cdot q'_R\), where \(q'_L\) and \(q'_R\)
are respectively the termination probabilities of 
\([t_1/y_1,\ldots,t_{\NONTERMS(X_i)}/y_{\NONTERMS(X_i)}]t_L\)
and 
\([t_1/y_1,\ldots,t_{\NONTERMS(X_i)}/y_{\NONTERMS(X_i)}]t_R\).
We omit a formal proof of correctness of the encoding.

We now check that \(\form_\GRAM\) belongs to the restricted fragment.
To this end, it suffices to check that, for each rule
\(X\,y_1\,\cdots\,y_k\to t_L\C{p}t_R\in \RULES\), the type judgment:
\[
\decTyEnv_\NONTERMS; \lamTE \p_{M_\GRAM}
\Avg(((P_1\lor \cdots \lor P_i)\land \Avg \toPHFL{t_L})
\lor ((P_{i+1}\lor \cdots \lor P_{m+1})\land \Avg \toPHFL{t_R})):
\PropTU{\emptyset}{\emptyset}
\]
holds for \(\decTyEnv_\NONTERMS = \set{X\COL (\PropTU{\emptyset}{\emptyset})^{\NONTERMS(X)}\to \PropTU{\emptyset}{\emptyset}
  \mid X\in\dom(\NONTERMS)}\)
and \(\lamTE = y_1\COL \PropTU{\emptyset}{\emptyset},\ldots,
y_k\COL\PropTU{\emptyset}{\emptyset}\).

By a straightforward induction on the structure of \(t_L\), we have:
\(\decTyEnv_\NONTERMS; \lamTE \p_{M_\GRAM} \toPHFL{t_L}:\PropTU{\emptyset}{\emptyset}\),
which implies
\(\decTyEnv_\NONTERMS; \lamTE \p_{M_\GRAM} \Avg \toPHFL{t_L}:\PropTU{\emptyset}{\emptyset}\).
We also have:
\[\decTyEnv_\NONTERMS; \lamTE \p_{M_\GRAM} P_1\lor \cdots \lor P_i:
\PropTU{\set{P_0,P_{i+1},\ldots,P_{m+1}}}{\set{P_1,\ldots,P_i}}.\]
Thus, by using \rn{T-And}, we have:
\[\decTyEnv_\NONTERMS; \lamTE \p_{M_\GRAM} (P_1\lor \cdots \lor P_i)\land \Avg \toPHFL{t_L}:
\PropTU{\set{P_0,P_{i+1},\ldots,P_{m+1}}}{\emptyset}.\]
Similarly, we have:
\[\decTyEnv_\NONTERMS; \lamTE \p_{M_\GRAM} (P_{i+1}\lor \cdots \lor P_{m+1})\land \Avg \toPHFL{t_R}:
\PropTU{\set{P_0,P_1,\ldots,P_i}}{\emptyset}.\]
Thus, by using \rn{T-Or} and \rn{T-Avg}, we obtain
\[
\decTyEnv_\NONTERMS; \lamTE \p_{M_\GRAM}
\Avg(((P_1\lor \cdots \lor P_i)\land \Avg \toPHFL{t_L})
\lor ((P_{i+1}\lor \cdots \lor P_{m+1})\land \Avg \toPHFL{t_R})):
\PropTU{\emptyset}{\emptyset}
\]
as required.
\section{Related Work}\label{section:related}
As mentioned in Section~\ref{section:intro}, PHFL can be regarded as a probabilistic extension of the higher-order fixpoint logic,
and as a higher-order extension of the \(\mu^p\)-calculus.
We thus compare our work with previous studies on
(non-probabilistic) higher-order fixpoint logic
and those on (non-higher-order) probabilistic logics.
As already mentioned, for (non-probabilistic) HFL, model checking of
finite-state systems is known to be
decidable~\cite{DBLP:conf/concur/ViswanathanV04}, and \(k\)-EXPTIME
complete~\cite{DBLP:journals/lmcs/AxelssonLS07}.  This is in a sharp
contrast with our result that \PHFL{} model checking is highly
undecidable (both \(\Pi^1_1\)-hard and \(\Sigma^1_1\)-hard) even at
order 1.
%%Kobayashi et al.~\cite{DBLP:conf/lics/KobayashiLG19}.
%%In a non-probabilistic setting, Kobayashi et al.~\cite{DBLP:conf/popl/KobayashiLB17}
%%have revealed a close relationship between model checking of
%%higher-order systems (called HORS) and HFL model checking of
%%finite-state systems.

As for studies on probabilistic logics,
besides the \(\mup\)-calculus, there are other probabilistic extensions of
the modal \(\mu\)-calculus~\cite{morgan1997probabilistic, DBLP:conf/lics/HuthK97, DBLP:journals/corr/MioS13}.
{\L}ukasiewicz $\mu$-calculus introduced by Mio and Simpson~\cite{DBLP:journals/corr/MioS13}
is among the most expressive ones, which has, in addition to
\(\land\) and \(\lor\),
another kind of conjunction (\(\odot\)) and disjunction (\(\oplus\)),
called {\L}ukasiewicz operations.
To our knowledge, ours is
the first \emph{higher-order} and probabilistic extension of
the modal \(\mu\)-calculus.
The decidable fragment of \PHFL{} studied in
Section~\ref{section:fragment} is strictly more expressive than the \(\mup\)-calculus.
We conjecture that the expressive power of PHFL is
incomparable to that of {\L}ukasiewicz $\mu$-calculus.
On the one hand, the property defined by the PHFL formula in the proof of
Theorem~\ref{expressiveThm} cannot be expressed in
{\L}ukasiewicz $\mu$-calculus, hence PHFL is not subsumed by
{\L}ukasiewicz $\mu$-calculus. On the other hand,
{\L}ukasiewicz operations do not seem expressible in PHFL.
It would be interesting to study a higher-order extension of
{\L}ukasiewicz $\mu$-calculus (in other words, an extension of PHFL with 
  {\L}ukasiewicz operations).
  
%%%Some probabilistic logics define semantics of the conjunction operator
%%%in a different way from ours.
%%%Indeed, when \(\semant{\phi_1}(s) = x\) and \(\semant{\phi_2}(s) = y\),
%%%the value \(\semant{\phi_1 \land \phi_2}(s)\) is
%%%\(\min(x, y)\) in the logics in~\cite{DBLP:journals/jcss/AlfaroM04, morgan1997probabilistic}
%%%along with PHFL and \(\mup\)-calculus, and
%%%\(\max(0, x + y - 1)\) in the logics in~\cite{DBLP:conf/lics/HuthK97, DBLP:journals/corr/MioS13}.
%%%When we require one of the values \(x\) and \(y\) is the constant \(0\) or \(1\)
%%%for every formulas of form \(\phi_1 \land \phi_2\),
%%%as we did in Section~\ref{section:fragment},
%%%behaviours of the conjuntion operator
%%%in these logics agree.

Recently, Kobayashi et al.~\cite{DBLP:conf/lics/KobayashiLG19} introduced
PHORS, a probabilistic extension of higher-order recursion schemes (HORS),
which can also be viewed as a higher-order extension of recursive Markov chains
(or probabilistic pushdown systems), and proved that the almost sure termination 
problem is undecidable. Although the problem setting is quite different (in our
work, the \emph{logic} is higher-order whereas the \emph{system} to be verified is higher-order
in their work), our encoding of the \(\mu\)-arithmetic in Section~\ref{section:hardness}
has been partially inspired by 
their undecidability proof; they also 
represented a natural number \(n\) as the probability \(\frac{1}{2^n}\).

In Section~\ref{section:undecidability}, we have used the undecidability
of the value-1 problem for probabilistic automata to prove the undecidability
of \PHFL{} model checking.
Fijalkow et al.~\cite{DBLP:journals/corr/FijalkowGKO15}
studied a decidable subclass of probabilistic automata called leaktight automata.
The idea of their restriction appears to be quite different from our type-based
restriction of the \PHFL{} model checking problem.
%%%a probabilistic extension of HORS (called PHORS) has been proposed
%%%by Kobayashi et al.~\cite{DBLP:conf/lics/KobayashiLG19}.
%%%We gave a translation from the termination problem of order-\(1\) PHORS
%%%to a fragment of PHFL model-checking problem
%%%in Section~\ref{section:fragment},
%%%which can be extended to general orders.
%%%Unlike the non-probabilistic setting mentioned above,
%%%the translation of the reverse direction is impossible,
%%%since the PHFL model-checking problem is \(\Pi^1_1\)-hard
%%%while the termination problem for PHORS is in \(\Sigma^0_2\)~\cite{DBLP:conf/lics/KobayashiLG19}.

\section{Conclusion}\label{section:conclusion}
We have introduced PHFL, a probabilistic logic which can be regarded as both a probabilistic extension of HFL and
a higher-order extension of the probabilistic logic \(\mu^p\)-calculus.
We have shown that the model-checking problem for PHFL for a finite Markov chain is undecidable
for the \(\mu\)-only and order-\(1\) fragment.
We have also shown that the model-checking problem for
the full order-\(1\) fragment of
PHFL is \(\Pi^1_1\)-hard and \(\Sigma^1_1\)-hard.
As positive results,
we have introduced a decidable subclass of the PHFL model-checking problem,
and showed that the termination problem of Recursive Markov Chains can be encoded
in the subclass.

Finding an upper bound of the hardness of the \PHFL{} model-checking problem is
left for future work. It is also left for future work to find a larger decidable class of
\PHFL{} model-checking problems.
%%prove (or disprove) the conjecture that expressive power of
%%the PHFL fragment is not subsumed by that of the \(\mu^p\)-calculus.
%%Our encode of the termination probabilities of Recursive Markov Chains supports this conjecture,
%%but we do not have a proof yet.
% Another direction of future work is
% to investigate connections between PHFL and PHORS~\cite{phors},
% like those between (non-probabilistic) HFL and HORS~\cite{horsAndHFL}.

\subsection*{Acknowledgements}
We would like to thank anonymous referees for useful comments.
This work was supported by JSPS KAKENHI Grant Number JP15H05706 and
JP20H00577, and JP20H05703.

\bibliographystyle{alpha}
\bibliography{paper}

\appendix
\section{Proofs}
\label{sec:proofs}
This section proofs omitted in the main text.

\subsection{Proof of Theorem~\ref{thm:semantRel}}

The following lemma states that the relation \(\semantRel{}\) is preserved by
the least upper bound operation.
\begin{lemma}\label{lem:appx:logical-relation-limit}
  Let \(I\) be a set and \(T\) be a type of \(\mu\)-arithmetic.
  Assume families \( \{ v_i \}_{i \in I} \) and \( \{ u_i \}_{i \in I} \) of elements of \( \muSemant{T} \) and \( \semant{\tr(T)} \), respectively, and
  suppose that \( v_i \semantRel{T} u_i \) for every \( i \in I \).
Then \( \bigJoin_{i \in I} v_i \semantRel{T} \bigJoin_{i \in I} u_i \).
\end{lemma}
\begin{proof}
By induction on the structure of \(T\).
The base case \(T = \muProp\) is obvious.
Assume that \(T = A \to T'\).

Let \(x \in \muSemant{A}\) and \(y \in \semant{\tr(A)}\) and assume that \(x \semantRel{A} y\).
For each \( i \in I \), since \( v_i \semantRel{T} u_i \), we have \( v_i\,x \semantRel{T'} u_i\,y \).
By the induction hypothesis,
\begin{equation*}
  \bigJoin_{i \in I} (v_i\,x) \semantRel{T'} \bigJoin_{i \in I} (u_i\,y).
\end{equation*}
Since the order on functions are component-wise, we have
\begin{equation*}
  \left(\bigJoin_{i \in I} v_i\right)\,x = \bigJoin_{i \in I} (v_i\,x)
  \quad\mbox{and}\quad
  \left(\bigJoin_{i \in I} u_i\right)\,y = \bigJoin_{i \in I} (u_i\,y).
\end{equation*}
So
\begin{equation*}
  \left(\bigJoin_{i \in I} v_i\right)\,x \semantRel{T'} \left(\bigJoin_{i \in I} u_i\right)\,y.
\end{equation*}
Since \( x \semantRel{A} y \) is arbitrary, we have: %this implies that
\begin{equation*}
  \left(\bigJoin_{i \in I} v_i\right) \semantRel{A \to T'} \left(\bigJoin_{i \in I} u_i\right). 
  \tag*{\qedhere}
\end{equation*}
\end{proof}

We are now ready to prove Theorem~\ref{thm:semantRel}.

\begin{proof}[Proof of Theorem~\ref{thm:semantRel}]
We prove the theorem by induction on the structure of \(\muform\).
In this proof, we omit the subscript \( M \) of \( \semant{-}_M \) for simplicity.
%%We discuss only the main cases; see \cite{MitaniFull} for more details.
\begin{itemize}
\item Case \(\muform = X\).
  In this case, \(\tr(\muform) = X\) and \(\tyEnv(X)=A\).
  Thus, we have:
\[
\muSemant{\tyEnv \muVdash \muform : \muTy}(\muInterpret) = \muInterpret(X)
\semantRel{\tyEnv(X)} \interpret(X)
= \semant{\tr(\tyEnv) \vdash \tr(\muform) : \tr(\muTy)}(\interpret),
\]
as required.
\item Case \(\muform = Z\).
In this case, \(\tr(\muform) = \toQuant{p_0'}\) and \(\muTy = N\).
We have
\[
\muSemant{\tyEnv \muVdash \muform : \muTy}(\muInterpret)
= 0
\semantRel{N} (1, 0, 0, 0)
= \semant{\tr(\tyEnv) \vdash \tr(\muform) : \tr(\muTy)}(\interpret).
\]

\item Case \(\muform = S \, t\).
In this case, \(\muTy=N\).
Let \(n = \muSemant{\tyEnv\muVdash t: N}(\muInterpret)\).
By the induction hypothesis, we have
\begin{equation*}
  \semant{\tr(\tyEnv) \vdash \tr(t) : \tr(N)}(\interpret) =
  \left(\frac{1}{2^{n}},\; 1 - \frac{1}{2^{n}},\; \_,\;\_ \right).
\end{equation*}
%We have \(\tr(\muform) = ...\) and thus
By the definition of \(\tr(\muform)\) and calculation, we have
\begin{equation*}
  \semant{\tr(\tyEnv) \vdash \tr(\muform):\tr(\muTy)}(\interpret) =
  \left(\frac{1}{2^{n + 1}},\; 1 - \frac{1}{2^{n + 1}},\; \_,\; \_\right),
\end{equation*}
which implies
\(\muSemant{\tyEnv\muVdash S\,t:\muTy}(\muInterpret)=n+1 \semantRel{N}
\semant{\tr(\tyEnv) \vdash\tr(\muform):\tr(\muTy)}(\interpret)\).

\item Case \(\muform = (s \leq t)\).
In this case, \(\muTy=\muProp\).
Let \(n = \muSemant{\tyEnv\muVdash s:N}(\muInterpret)\) and
\(m = \muSemant{\tyEnv\muVdash t:N}(\muInterpret)\).
By the induction hypothesis, we have
\begin{align*}
\semant{\tr(\tyEnv) \vdash\tr(s):\tr(N)}(\interpret) &= (\frac{1}{2^n},\;, 1 - \frac{1}{2^n},\;\_,\;\_) \\
\semant{\tr(\tyEnv) \vdash\tr(t):\tr(N)}(\interpret) &= (\frac{1}{2^m},\;, 1 - \frac{1}{2^m},\;\_,\;\_).
\end{align*}
By the definition of \( \tr(s \leq t) \) and calculation, we have
\begin{equation*}
  \semant{\tr(\tyEnv)\vdash \tr(s \le t):\tr(A)}(\interpret) =
  \begin{cases}
    (1,\_,\_,\_) &
    \mbox{(if \( \frac{1}{2} \times \big(\frac{1}{2^n} + (1-\frac{1}{2^m}) \big) \ge \frac{1}{2} \), i.e. if \(n\le m\))} \\
    (0,\_,\_,\_) & \mbox{(if \( \frac{1}{2} \times \big(\frac{1}{2^n} + (1-\frac{1}{2^m}) \big) < \frac{1}{2} \), i.e., if \(n>m\)).} \\
  \end{cases}
\end{equation*}
%%%This means that
%%%\begin{equation*}
%%%  \semant{\tr(s \le t)}(\interpret) =
%%%  \begin{cases}
%%%    (1,\_,\_,\_) & \mbox{(if \( n \le m \))} \\
%%%    (0,\_,\_,\_) & \mbox{(if \( n > m \)).} \\
%%%  \end{cases}
%%%\end{equation*}
Thus, we have 
\( \muSemant{\tyEnv\muVdash s \le t:\muTy}(\muInterpret) \semantRel{\muTy}
\semant{\tr(\tyEnv)\vdash \tr(s \le t):\tr(\muTy)}(\interpret) \) as required.

\iffull
\item Case \(\muform = \muform_1 \land \muform_2\).

In this case, we have \(\muTy = \muProp\). %% and \(\tr(\muTy) = \Qual\).

By the induction hypothesis, we have
\[
\muSemant{\tyEnv \muVdash \muform_i : \muTy}(\muInterpret) \semantRel{\muTy} \semant{\tr(\tyEnv) \vdash \tr(\muform_i) : \tr(\muTy)}(\interpret)
\]
for each \(i = 1, 2\).

By the definition of \(\semantRel{\muProp}\), we have
\[
\semant{\tr(\Gamma) \vdash \tr(\muform_i) : \tr(\muProp)}(\interpret)
= (\muSemant{\tyEnv \muVdash \muform_i : \muProp}(\muInterpret), \; \_,\; \_,\; \_)
\]
for each \(i = 1, 2\).
Therefore we have
\begin{align*}
\muSemant{\tyEnv \muVdash \muform_1 \land \muform_2 : \muProp}(\muInterpret)
&= \muSemant{\tyEnv \muVdash\muform_1:\muProp}(\muInterpret) \GLB{\muProp} \muSemant{\tyEnv \muVdash\muform_2:\muProp}(\muInterpret) \\
&\semantRel{\muProp}
(\muSemant{\Gamma \vdash\muform_1:\muProp}(\muInterpret)
\GLB{\muProp} \muSemant{\Gamma \vdash\muform_2:\muProp}(\muInterpret),\; \_,\; \_,\; \_) \\
&= \semant{\tr(\tyEnv)\vdash \muform_1:\Qual}(\interpret) \GLB{}
\semant{\tr(\tyEnv)\vdash\muform_2:\Qual}(\interpret) \\
&= \semant{\tr(\tyEnv)\vdash\muform_1 \land \muform_2:\Qual}(\interpret)
\end{align*}
as desired.

\item Case \(\muform = \muform_1 \lor \muform_2\).
Similar to the above case.
\fi
\item Case \(\muform = \lambda X. \muform'\).
  In this case, \(\muTy\)  is of the form \(B \to T\), with
  \(\tyEnv, X : B \muVdash \muform' : T\).
For any
  \(v \in \semant{B}\) and \(u \in \semant{\tr(B)}\) such that \(v \semantRel{B} u\),
  we have:
\begin{align*}
  \muSemant{\tyEnv \muVdash \muform : \muTy}(\muInterpret)(v)
&  = 
\muSemant{\tyEnv, X : B \muVdash \muform' : T}(\muInterpret[X \mapsto v])\\
&\semantRel{T}
\semant{\tr(\tyEnv, X : B) \vdash \tr(\muform') : \tr(T)}(\interpret[X \mapsto u])\\
& \qquad\qquad\mbox{ (by the induction hypothesis)}\\
&=
\semant{\tr(\tyEnv) \vdash \tr(\muform) : \tr(\muTy)x}(\muInterpret)(u)
\end{align*}
%%By the induction hypothesis, \(\muform\) satisfies
%%\[
%%\muSemant{\tyEnv, X : B \muVdash \muform : T}(\muInterpret[X \mapsto v])
%%\semantRel{T}
%%\semant{\tr(\tyEnv, X : B) \vdash \tr(\muform) : \tr(T)}(\interpret[X \mapsto u]).
%%\]
%%We have thus
Therefore, %%by the definition of \(\semantRel{B \to T}\),
we have
\[
\muSemant{\tyEnv \muVdash \muform : B \to T}(\muInterpret)
\semantRel{B \to T}
\semant{\tr(\tyEnv) \vdash \tr(\muform) : \tr(B \to T)}(\interpret)
\]
as required.

\item Case \(\muform = \muform_1 \, \muform_2\).
We have \(\muTy = T\), with \(\tyEnv \muVdash \muform_1 : B \to T\) and \(\tyEnv \muVdash \muform_2 : B\).

By the induction hypothesis, we have
\(\muSemant{\tyEnv \muVdash\muform_1:B\to T}(\muInterpret) \semantRel{B \to T} \semant{\tr(\tyEnv) \vdash\tr(\muform_1):\tr(B\to T)}(\interpret)\)
and \(\muSemant{\tyEnv \muVdash\muform_2:B}(\muInterpret) \semantRel{B} \semant{\tr(\tyEnv) \vdash\tr(\muform_2):\tr(B)}(\interpret)\).
Therefore by the definition of \(\semantRel{B \to T}\), we have
\begin{align*}
\muSemant{\tyEnv \muVdash\muform_1 \, \muform_2:\muTy}(\muInterpret) &= \muSemant{\tyEnv \muVdash\muform_1:B\to T}(\muInterpret) \, (\muSemant{\tyEnv \muVdash\muform_2:B}(\muInterpret)) \\
&\semantRel{\muTy} \semant{\tr(\tyEnv) \vdash\tr(\muform_1):\tr(B\to T)}(\interpret) \, (\semant{\tr(\tyEnv) \vdash\tr(\muform_2):\tr(B)}(\interpret)) \\
&= \semant{\tr(\tyEnv) \vdash\tr(\muform_1 \, \muform_2):\tr(\muTy)}(\interpret)
\end{align*}
as required.
\item Case \(\muform = \mu X. \muform'\).

  In this case, \(A=T\), with \(\tyEnv, X\COL T\muVdash \muform':T\).
By the induction hypothesis, for any
\(v \in \muSemant{T}\) and
\(u \in \semant{\tr(T)}\)
such that \(v \semantRel{T} u\),
we have
\[
\muSemant{\tyEnv,X\COL T\muVdash \muform':T}(\muInterpret[X \mapsto v])
\semantRel{T} \semant{\tr(\tyEnv),X\COL\tr(T)\vdash \tr(\muform'):\tr(T)}(\interpret[X \mapsto u]).
\]

Since \(\tr(\mu X. \muform') = \mu X. \tr(\muform')\), it suffices to show:
\[
\muSemant{\tyEnv\muVdash\mu X. \muform':T}(\muInterpret)
\semantRel{T} \semant{\tr(\tyEnv)\vdash \mu X. \tr(\muform'):\tr(T)}(\interpret).
\]

Let \( \Func : \muSemant{T} \to \muSemant{T} \) and \( \Funcb : \semant{\tr(\muTyT)} \to \semant{\tr(\muTyT)} \) be the functions defined by:
\begin{align*}
 & \Func(v) := \muSemant{\tyEnv,X\COL T\muVdash\muform':T}(\muInterpret[X \mapsto v])\\
 & \Funcb(u) := \semant{\tr(\tyEnv),X\COL\tr(T)\vdash\tr(\muform'):\tr(T)}(\interpret[X \mapsto u]).
\end{align*}
By the reasoning above, we have \(\Func \semantRel{T \to T} \Funcb\).
By the definitions of the semantics, we have
\(\muSemant{\tyEnv\muVdash\mu X. \muform':T}(\muInterpret) = \LFP(\Func)\) and
\(\semant{\tr(\tyEnv)\vdash\mu X. \tr(\muform'):\tr(T)}(\interpret) = \LFP(\Funcb)\).
Then there exists an ordinal \( \alpha \) such that
\begin{equation*}
  \LFP(\Func) = \Func^\alpha(\bot_{\muTyT})
  \qquad\mbox{and}\qquad
  \LFP(\Funcb) = \Funcb^{\alpha}(\bot_{\tr(\muTyT)}),
\end{equation*}
where \(f^\beta(x)\) is defined by
\(f^0(x)=x\), \(f^{\beta+1} = f(f^\beta(x))\), and \(f^{\beta}=\LUB{}_{\gamma<\beta} f^\gamma(x)\)
if \(\beta\) is a limit ordinal.
We shall prove by (transfinite) induction on \( \beta \) that \( \Func^{\beta}(\bot_{{\muTyT}}) \semantRel{\muTyT} \Funcb^{\beta}(\bot_{{\tr(\muTyT)}}) \),
which would imply
\begin{equation*}
  \LFP(\Func) = \Func^\alpha(\bot_{{\muTyT}})
  \semantRel{\muTyT}
  \Funcb^{\alpha}(\bot_{{\tr(\muTyT)}}) = \LFP(\Funcb)
\end{equation*}
as required.

The base case \(
\Func^0(\bot_{{\muTyT}})=\bot_{{\muTyT}}
  \semantRel{\muTyT}
  \bot_{{\tr(\muTyT)}}=  \Funcb^{0}(\bot_{{\tr(\muTyT)}})\)
  follows by a straightforward induction on the structure of \(\muTyT\).
  The case where \(\beta\) is a successor ordinal follows immediately
  from the induction hypothesis and \(\Func \semantRel{T \to T} \Funcb\).
  If \( \beta \) is a limit ordinal, then
\begin{equation*}
  \Func^\beta(\bot_{{\muTyT}}) = \LUB{}_{\gamma < \beta} \Func^\gamma(\bot_{{\muTyT}})
  \qquad\mbox{and}\qquad
  \Funcb^\beta(\bot_{{\muTyT}}) = \LUB{}_{\gamma < \beta} \Funcb^\gamma(\bot_{{\muTyT}}).
\end{equation*}
By the induction hypothesis (of the transfinite induction),
\begin{equation*}
  \Func^\gamma(\bot_{{\muTyT}}) \semantRel{\muTyT} \Funcb^\gamma(\bot_{{\muTyT}})
\end{equation*}
for every \( \gamma < \beta \).
%Hence, 
\iffull
By Lemma~\ref{lem:appx:logical-relation-limit} below,
\else
Since \(\semantRel{\muTyT}\) is preserved by the least upper-bound operation (which can be
proved by an easy induction on \(T\)),
\fi
we have
\begin{equation*}
  \Func^\beta(\bot_{{\muTyT}}) \semantRel{\muTyT} \Funcb^\beta(\bot_{{\muTyT}})
\end{equation*}
as required.

%% \begin{itemize}
%% \item
%% Case \(\beta = \gamma + 1\)

%% The result is immediate from the induction hypothesis of Theorem~\ref{thm:semantRel}.
%% \item
%% Case that \(\beta\) is a limit ordinal

%% We use the following lemma.

%% Applying this lemma for \(I = \{\alpha | \alpha < \beta\}\),
%% \(t_\alpha = \Func^\alpha(\bot_{\muSemant{\muTy}})\) and \(t'_\alpha = \Func'^\alpha(\bot_{\semant{\tr(\muTy)}})\),
%% we get the desired result.
%% \end{itemize}
\item Case \(\muform = \nu X. \muform'\).
Similar to the case for \(\muform=\mu X.\muform'\) above.
\qedhere
\end{itemize}
\end{proof}

\subsection{Proof of Lemma~\ref{lem:decidableModuloSim}}

The proof proceeds by induction on the derivation for
\( \decTyEnv;\lamTE \pM \phi: \decTy \), with case analysis on the last rule used.
\begin{itemize}

\item Cases for \rn{T-WeakTU} and \rn{T-Weak}: Trivial by the induction hypothesis.
\item Cases for \rn{T-AP} and \rn{T-Var}: Since
  \(\semant{\lambda \lamTE.\phi}(\interpret)\) does not depend on \(\interpret\),
  we have \( \semant{\lambda \lamTE.\phi}(\interpret) = \semant{\lambda \lamTE.\phi}(\interpret') \).
\item Case for \rn{T-FVar}: In this case, \(\lamTE=\emptyset\) and \(\phi=X\).
  We have
  \[
  \semant{\lambda \lamTE.\phi}(\interpret)= \interpret(X)  \precsim_{\lamTE \to \decTy}
  \interpret'(X)=\semant{\lambda \lamTE.\phi}(\interpret'),\]
  as required.
\item Case for \rn{T-Mu}: In this case, \(\lamTE=\emptyset\) and \(\phi=\mu X.\phi'\),
  with \(\decTyEnv, X\COL\decTy;\emptyset \vdash_M \phi' : \decTy\), where
  $\decTy$ is of the form $\cdots\to \PropTU{T}{\emptyset}$.
  Let \(h_0 = \lambda x.\semant{\phi}{\rho\set{X\mapsto x}}\) and
  \(h_1 = \lambda x.\semant{\phi'}{\rho'\set{X\mapsto x}}\).
  By the induction hypothesis, \(x_0\preceq_{\decTy} x_1\) implies
  \(h_0\,x_0\preceq_{\decTy}h_1\,x_1\).
  Thus, by Lemma~\ref{lem:decidableSimFixpoint}, we have
  \[
  \semant{\lambda \lamTE.\phi}(\interpret)
  = \LFP(h_0)\precsim_{\decTy} \LFP(h_1) = 
  \semant{\lambda \lamTE.\phi}(\interpret'),
  \]
  as required.
\item Case for \rn{T-Nu}: Similar to the case for \rn{T-Mu} above.
\item Case for \rn{T-Abs}: In this case, \(\phi=\lambda X.\phi'\).
  The result follows immediately from the induction hypothesis,
  as \(\lambda \lamTE.\lambda X.\phi' = \lambda (\lamTE,X\COL\PropTU{T}{U}).\phi'\).
\item Case for \rn{T-App}:
  In this case, \(\phi=\phi_1\phi_2\).
  By the induction hypothesis, we have:
  \( \semant{\lambda \lamTE.\phi_1}(\interpret) \precsim_{\lamTE \to \PropTU{T}{U}\to\decTy} \semant{\lambda \lamTE.\phi_1}(\interpret') \)
  and 
  \( \semant{\lambda \lamTE.\phi_2}(\interpret) \precsim_{\lamTE \to \PropTU{T}{U}} \semant{\lambda \lamTE.\phi_2}(\interpret') \).
  Suppose \(\lamTE=X_1\COL\PropTU{T_1}{U_1},\ldots,X_k\COL\PropTU{T_k}{U_k}\)
  and let \(x_i\in\semant{\PropTU{T_i}{U_i}}\) for each \(i\in\set{1,\ldots,k}\).
  Then, we have:
  \[
  \begin{array}{rl}
  \semant{\lambda \lamTE.\phi}(\interpret)\\
  =&
  \lambda x_1\in\semant{\Quant}.\cdots\lambda x_k\in\semant{\Quant}.
  \semant{\phi}{\interpret\set{X_1\mapsto x_1,\ldots,X_k\mapsto x_k}}\\
  =&
  \lambda x_1\in\semant{\Quant}.\cdots\lambda x_k\in\semant{\Quant}.
  (\semant{\lambda \lamTE.\phi_1}{\interpret}\,x_1\,\cdots,x_k)
  (\semant{\lambda \lamTE.\phi_2}{\interpret}\,x_1\,\cdots,x_k)\\
  \precsim_{\lamTE\to\decTy}&
  \lambda x_1\in\semant{\Quant}.\cdots\lambda x_k\in\semant{\Quant}.
  (\semant{\lambda \lamTE.\phi_1}{\interpret'}\,x_1\,\cdots,x_k)
  (\semant{\lambda \lamTE.\phi_2}{\interpret'}\,x_1\,\cdots,x_k)\\
  =&
  \semant{\lambda \lamTE.\phi}(\interpret')
  \end{array}
  \]
  as required.
  \item The remaining cases follow immediately from the induction hypothesis.
\end{itemize}

\end{document}